\documentclass[12pt]{article}
\usepackage{amsmath}
\usepackage{graphicx}
\usepackage{enumerate}
\usepackage{natbib}
\usepackage{url} 
\usepackage{caption}
\usepackage{subcaption}

\usepackage{bm}
\usepackage{amsfonts}

\usepackage{hyperref}
\hypersetup{
    colorlinks,
    citecolor=black,
    filecolor=black,
    linkcolor=black,
    urlcolor=black
}

\usepackage{multirow}
\usepackage{tablefootnote}

\usepackage{nicefrac}
\usepackage{mathtools}
\usepackage{enumitem}

\setlist[itemize]{noitemsep, topsep=0pt}

\usepackage{mathrsfs}

\usepackage{amssymb}

\usepackage{amsthm}
\newtheorem{theorem}{Theorem}
\newtheorem{lemma}{Lemma}

\newtheorem*{remark}{Remark}

\usepackage{physics}

\newcommand{\iid}{\overset{\mathrm{iid}}{\sim}}

\newcommand{\E}{\mathbb{E}}
\newcommand{\bP}{\mathbb{P}}
\newcommand{\V}{\mathbb{V}}

\newcommand*{\prob}{\mathsf{P}}

\newcommand{\talpha}{\tilde{\alpha}}
\newcommand{\tbeta}{\tilde{\beta}}
\newcommand{\tlambda}{\tilde{\lambda}}
\newcommand{\tLambda}{\tilde{\Lambda}}

\newcommand{\tzeta}{\tilde{\zeta}}

\newcommand{\tY}{\tilde{Y}}

\newcommand{\cS}{\mathcal{S}}

\newcommand{\sD}{\mathscr{D}}
\newcommand{\sT}{\mathscr{T}}

\newcommand*\dif{\mathop{}\!\mathrm{d}}

\newcommand{\hk}{\hat{k}}
\newcommand{\hK}{\hat{K}}

\newcommand{\tf}{\tilde{f}}
\newcommand{\tg}{\tilde{g}}
\newcommand{\hf}{\hat{f}}
\newcommand{\hg}{\hat{g}}

\newcommand{\fo}{f^{(1)}}
\newcommand{\fj}{f^{(j)}}
\newcommand{\hfo}{\hf^{(1)}}

\usepackage{authblk}

\begin{document}

\def\spacingset#1{\renewcommand{\baselinestretch}%
{#1}\small\normalsize} \spacingset{1}


\date{}
\title{\bf Sampling depth trade-off in function estimation under a two-level design}
\author[1]{Akira Horiguchi\thanks{\url{akira.horiguchi@duke.edu}}}
\author[2]{Li Ma\thanks{\url{li.ma@duke.edu}}}
\author[3]{Botond T. Szabó\thanks{\url{botond.szabo@unibocconi.it}}}
\affil[1,2]{Department of Statistical Science, Duke University}
\affil[3]{Department of Decision Sciences and Institute for Data Science and Analytics, Bocconi University}
\maketitle

\begin{abstract}
Many modern statistical applications involve a two-level sampling scheme that first samples subjects from a population and then samples observations on each subject.
These schemes often are designed to learn both the population-level functional structures shared by the subjects and the functional characteristics specific to individual subjects.
Common wisdom suggests that learning population-level structures benefits from sampling more subjects whereas learning subject-specific structures benefits from deeper sampling within each subject. 
Oftentimes these two objectives compete for limited sampling resources, which raises the question of how to optimally sample at the two levels.
We quantify such sampling-depth trade-offs by establishing the $L_2$ minimax risk rates for learning the population-level and subject-specific structures under a hierarchical Gaussian process model framework where we consider a Bayesian and a frequentist perspective on the unknown population-level structure. 
These rates provide general lessons for designing two-level sampling schemes given a fixed sampling budget. 
Interestingly, they show that subject-specific learning occasionally benefits more by sampling more subjects than by deeper within-subject sampling.
We show that the corresponding minimax rates can be readily achieved in practice through simple adaptive estimators without assuming prior knowledge on the underlying variability at the two sampling levels. We validate our theory and illustrate the sampling trade-off in practice through both simulation experiments and two real datasets. While we carry out all the theoretical analysis in the context of Gaussian process models for analytical tractability, the results provide insights on effective two-level sampling designs more broadly.
\end{abstract}

\noindent%
{\it Keywords:} Sampling design, nonparametric, Gaussian process, hierarchical model, Bayesian model

\section{Introduction}
\label{sec:intro}

Many statistical applications involve collecting data in which repeated sampling occurs at multiple levels. A two-level nested sampling design widely in use proceeds by first sampling at the level of subjects, and then collecting observations on each subject. 
In biomedical applications, for instance, the subjects of interest are often individual patients, and for each patient, the investigator may collect sequencing reads by a high-throughput sequencer in order to identify patient characteristics such as their gene regulation, microbiome composition, and metabolic or immunological profiles. 
In different studies the subjects of interest might be different (e.g., single-cell RNA-seq studies might involve individual cells as the subjects) and the number of sampling levels might be larger (e.g., in a longitudinal single-cell analysis involving multiple patients). Herein for simplicity we consider a two-level sampling setup that assumes the subjects are sampled exchangeably from some population of interest, and for each subject, observations are sampled from an underlying data-generative distribution that characterizes the relevant contextual features of the subject.

A key assumption in collecting data from multiple subjects is that they share similarities that characterize the underlying target population, and so a common objective is to learn the commonalities in the data-generative mechanism shared across the population. One might ask: what gene pathways are active among patients of a particular disease? What microbes are typically present in a breast-fed 3-month old? 
Intuitively, these shared structures can be inferred by pooling data from the subjects sampled from such a population. We refer to such inference as population-level learning. 
In addition, one may also wish to learn the characteristics specific to an individual subject. Questions such as  
``Are there immunological anomalies for this particular patient?''  might arise.
Such subject-specific learning is often the prime objective in precision medicine or personalized treatment.

Common wisdom suggests that learning the shared structures in the target population requires sampling enough subjects, whereas learning subject-specific characteristics requires enough observations from that particular subject. 
But because real studies are constrained by the available time and resources, the induced trade-off between the ``sampling depths'' at these two levels---that is, in terms of the number of subjects versus the number or precision of the observations made for each subject (respectively denoted as $m$ and $n$ later in the paper)---can pit these two learning objectives against each other. In order to garner insights on ``optimal'' sampling designs that prioritize one or balance both objectives, one must quantify how such a trade-off takes place. We aim to provide such a characterization in this paper for a benchmark hierarchical Gaussian process model.

Related ideas in two-level sampling have been considered from the early days of experiment designs. 
For example, \cite{novick1972estimating} performed a numerical study and found that subject-specific learning improved by pooling data from all subjects in the study, especially when the sample sizes were small \citep{lindley1972bayes}. 
However, such considerations mainly concerned finite-dimensional quantities of interest, whereas modern applications see increasingly complex population-level and subject-specific characteristics of interest which are often functional in nature. 
Here we review existing works that study minimax rates in the presence of multiple related functional observations.
For estimating a population-level function, \cite{bunea2006minimax} and \cite{chau2016functional} provide risk bounds for the population-mean log-spectrum in the context of time series, whereas \cite{giacofci2018minimax} assume a heteroscedastic nonparametric regression setting to estimate the function from multi-subject noisy curves.
For estimating subject-specific functions, \cite{wang2016nonparametric} establish risk upper bounds for the scenario where a densely sampled function is used to aid in the regression of a sparsely sampled function. 
\cite{koudstaal2018multiple} establish minimax rates for a multiple Gaussian sequence model where the subject-specific functions are realizations of a mean-zero Gaussian process and hence do not possess any shared structure.
Finally, \cite{bak2023minimax} establish minimax rates in a nonparametric regression setting where the functions are spanned by a smaller number of other orthonormal functions which by definition do not possess any shared structure.
However, none of the works we found on multi-subject designs for functional data seem to address the apparent sampling-depth trade-off for the objectives of learning both population-level and subject-specific functions of interest.

This paper aims to translate the intuition behind the mechanisms underlying two-level sampling into theory by quantifying this trade-off for the aforementioned two objectives.
The paper is organized as follows. 
Section~\ref{sec:mainresults} presents our theoretical results.
To ease technical computation, we assume the subject-specific functions are realizations of a Gaussian process whose mean function characterizes their population-level structure and whose covariance function specifies the cross-sample heterogeneity across the subjects.
On the unknown population-level structure, we take both a frequentist {\em and} a Bayesian perspective; the former assumes the structure is a fixed member of a Sobolev class of functions, while the latter endows it with a Gaussian process prior. 
We establish the rates of $L_2$ minimax risk (under the frequentist setup) and minimum average risk (under the Bayesian setup) for both the subject-specific {\em and} population-mean functions.
We also construct threshold estimators that adapt to the unknown smoothness at both subject and population levels to achieve the minimax rate by applying Lepskii's method \citep{lepskii1992asymptotically,lepskii1993asymptotically}. 

All of the estimators we build in this paper are very simple and would not necessarily be our recommended methods in various applications. Their purpose is to show that the established minimax rates can be readily attained in practice by estimators that can be easily constructed. As such, we hope the reader does not get distracted by the specific design of these simple estimators. Our focus herein shall be on the risk rates and the lessons one may draw from them regarding the choice of sample design---in particular $n$ and $m$---for an experiment with a fixed sampling budget.
Proofs of these results are in the Supplementary Material.
Section~\ref{sec:studyall} contains simulation studies that empirically validate our theoretical results and compare our proposed estimators to existing single-subject estimators. 
Sections \ref{sec:pinch} and \ref{sec:orthosis} further compare the estimators on two real-world data sets: one that records the force exerted during a brief pinch by the thumb and forefinger \citep{ramsay1995functional}, and one that records how different externally applied moments to the human knee affect the processes underlying movement generation \citep{cahouet2002static}.
Section~\ref{sec:discussion} provides concluding remarks. 

We finish the introduction with a preview of our findings. 
Our theoretical results in Sections \ref{sec:estg} and \ref{sec:estf} indeed corroborate the general wisdom of the trade-off between sampling depths for many settings.
In addition, in some settings our established multi-subject rates are faster than the single-subject rates, which confirms that in such settings the information pooling process can strongly improve estimation of the subject-specific functions. 
Our rates also delineate the mechanism in which the smoothnesses of the functions, the number of subjects, and how precisely each subject is sampled jointly influence the quality of the inference and the benefit of the information from the additional subjects. 
Interestingly, the theoretical finding of Figure~\ref{fig:negativegradientf} and numerical experiments in Section~\ref{sec:study2} suggest that if the only goal is to estimate some subject-specific function under a constrained budget, the goal might be best achieved not by spending the entire budget on sampling that subject, but instead by allocating some of the budget to sparingly sample a large number of other subjects; this strategy of pooling the subjects' inferred information allows precise inference on the population-level structure without sacrificing much subject-specific information when the population-level structure is complex and the between-subject variations are relatively small. 

\section{Main results}
\label{sec:mainresults}
\subsection{The two-level sampling model}
\label{sec:samplingmodel}

The paper assumes a two-level sampling model. The first level samples $m$ subjects from some target population. For the $j$th subject, $j=1,\ldots,m$, we use $f^{(j)}$ to represent a function of interest, while the observation for the subject is a sample path $Y^{(j)} \coloneqq \{Y_t^{(j)}, t \in [0,1]\}$ from the Gaussian white noise model given $f^{(j)}$. 
The white noise model is the idealized version of the nonparametric regression model and is used as an analytically tractable benchmark model for other more complicated nonparametric models. 
We idealize the sampling distribution of each $f^{(j)}$ to be a Gaussian process whose mean function characterizes the population structure of interest and whose covariance function $\tilde\Lambda$ specifies the cross-subject variability of $f^{(j)}$. 
By Mercer's theorem \citep[see e.g., ][]{williams2006gaussian}, there is an orthonormal eigenbasis $\{\psi_k\}_{k=1}^{\infty}$ in $L_2([0,1])$ such that the eigenvalues $\{\tilde\lambda_k\}$ of $\tilde\Lambda$ are all nonnegative and $\tilde\Lambda(s,t) = \sum_{k=1}^\infty \tilde\lambda_k \psi_k(s) \psi_k(t)$ for all $s,t \in [0,1]$.
Furthermore, we assume $\tilde\Lambda$ is characterized by a regularity hyper-parameter $\tilde\alpha$ describing the polynomial decaying rate of the eigenvalues (e.g., Matérn or Riemann-Liouville processes), i.e., $\tilde\lambda_k \asymp k^{-1-2\tilde\alpha}$. 
Thus the parameter $\tilde\alpha$ characterizes the ``smoothness'' of realizations from a mean-zero Gaussian process with covariance function $\tilde\Lambda$. 

We will consider both a Bayesian and a frequentist model on the population structure of interest. 
The Bayesian model endows the population-level function with a Gaussian process prior.
Our full model thus becomes a hierarchical Gaussian process model
\begin{align}  \label{eq:gwnmodelbayes}
\begin{split}
\dif Y^{(j)}_t &= \fj(t)\dif t + n^{-1/2} \dif W_t^{(j)},\quad t\in[0,1]   \\
f^{(j)}\mid g,\tilde\Lambda  &\iid \text{GP}(g,\tilde\Lambda), \\
g\mid\Lambda &\sim \pi_{\Lambda}
\end{split}
\end{align}
where $W^{(j)} \coloneqq \{W_t^{(j)}, t \in [0,1]\}$, $j=1,\ldots,m$, are i.i.d.\ Brownian motions, and $\pi_{\Lambda}$ is a mean-zero Gaussian process distribution whose covariance function $\Lambda$ shares the same eigenfunctions $\{\psi_k\}$ as the covariance function $\tilde\Lambda$ and has eigenvalues $\{\lambda_k\}$ with polynomial decaying rate $\alpha$, i.e. $\lambda_k \asymp k^{-1-2\alpha}$. 

In the frequentist framework, we instead assume the population-level function is unknown and fixed.
Our sampling model thus becomes
\begin{align}  \label{eq:gwnmodelfreq}
\begin{split}
\dif Y^{(j)}_t &= \fj(t)\dif t + n^{-1/2} \dif W_t^{(j)},\quad t\in[0,1]   \\
f^{(j)}\mid g^*,\tilde\Lambda  &\iid \text{GP}(g^*,\tilde\Lambda), \\
g^* &\in \cS(\alpha,R)
\end{split}
\end{align}
where $W^{(1)}, \ldots, W^{(m)}$ are i.i.d. Brownian motions, and the population-level function $g^*$ belongs to the Sobolev-type ball (i.e., ellipsoid with respect to the Euclidean norm and the basis functions $\{\psi_k\}$)  
\begin{align}  \label{eq:ellipsoid}
    \cS(\alpha,R) = \left\{g(\cdot)=\sum_{k=1}^\infty g_k\psi_{k}(\cdot)\in L^2[0,1]: \sum_k \left(g_k^2 k^{2\alpha}\right) \leq R^2\right\}
\end{align}
for some positive rate $\alpha$ and radius $R$.

Our primary technical goal is to find the Bayes risk (i.e., minimum average risk) under model~\eqref{eq:gwnmodelbayes} and the $L_2$ minimax risk under model~\eqref{eq:gwnmodelfreq} of the population-level function $g^*$ and each subject-specific function $\fj$ from the noisy observations. 
Because the subject-specific functions are exchangeable under either model, we can simplify this goal to finding the Bayes and minimax risks of $g^*$ and $\fo$.
These risks will be functions of $m$ and $n$ and hence will provide insight into the trade-off between the number of samples and the accuracy of each observed function on estimating $g^*$ and $\fo$. 
For both functions we will also produce an adaptive estimator that achieves the function's Bayes and minimax risk rates without knowledge of either $\alpha$ or $\tilde\alpha$. 
As noted earlier, we choose these simple estimators to show that the rates can be easily attained in practice. Our main focus will be on the rates and their implications on the sampling design.
Proofs are included as Supplemental Material.

\subsection{Notation} 
We write $a \wedge b = \min\{a, b\}$ and $a \vee b = \max\{a, b\}$. For two positive sequences $a_n,b_n$ we write $a_n\lesssim b_n$ if there exists a universal positive constant $C$ such that $a_n\leq C b_n$. We write $a_n\asymp b_n$ if $a_n \lesssim b_n$ and $b_n \lesssim a_n$ are satisfied simultaneously.
For two doubly indexed positive sequences $a_{n,m},b_{n,m}$ we write $a_{n,m} \lesssim b_{n,m}$ if there exists a universal positive constant $D$ such that $a_{n,m} \leq D b_{n,m}$. We write $a_{n,m}\asymp b_{n,m}$ if $a_{n,m} \lesssim b_{n,m}$ and $b_{n,m} \lesssim a_{n,m}$ are satisfied simultaneously.

For any function $g$ defined on the unit interval, we denote by $\bP_g$ the joint distribution of $Y^{(1)},\ldots,Y^{(m)},f^{(1)},\ldots,f^{(m)}|g$ and denote by $\E_g$ the expectation over this distribution.

\subsection{Estimating the population-level function}  \label{sec:estg}

First we find the Bayes risk for estimating the population-level function generated from a mean-zero Gaussian process with covariance function $\Lambda$.
The following theorem provides the squared Bayes risk's asymptotic rate as well as an estimator that achieves this rate. 

\begin{theorem} \label{thm:lbbayesriskg}
    Consider the hierarchical model \eqref{eq:gwnmodelbayes}.
    Then the corresponding minimum average mean integrated squared error (MISE) for $g$ with respect to the distribution $\pi_{\Lambda}$ is
    \begin{equation*}  \label{eq:bayesrisklowerboundg}
        \inf_{\hg} \int \E_g\|\hg-g\|_2^2\dif\pi_{\Lambda}(g) 
        = \int \E_g\|\tg_{\Lambda}-g\|_2^2\dif\pi_{\Lambda}(g)
        \asymp m^{-1}+(nm)^{-\frac{2\alpha}{1+2\alpha}},
    \end{equation*}
    where $\tilde{g}_\Lambda$ denotes the posterior mean in this model.
\end{theorem}

The Bayes risk's rate is thus $m^{-1/2}+(nm)^{-\alpha/(1+2\alpha)}$, where we see the number of subjects and the sampling precision of each observed function, i.e., $m$ and $n$, playing different roles.
As common wisdom suggests, this rate cannot be made arbitrarily small by fixing $m$ and appropriately increasing $n$, seeing as perfect knowledge of finitely many subject-specific functions still allows for much variation in what the population-level function could be. 
On the other hand, this rate shrinks to zero as $m \rightarrow \infty$ due to the mutual independence of the Brownian motions modeling the observation noise and of the subject-specific deviations $(f^{(j)}-g) \mid g$ for $j=1,\ldots,m$, which implies that we can view the $m$ observed functions $Y^{(1)}, \ldots, Y^{(m)} \mid g$ as $m$ mutually independent noisy versions of $g$.

\begin{figure}
\includegraphics[width=0.7\textwidth]{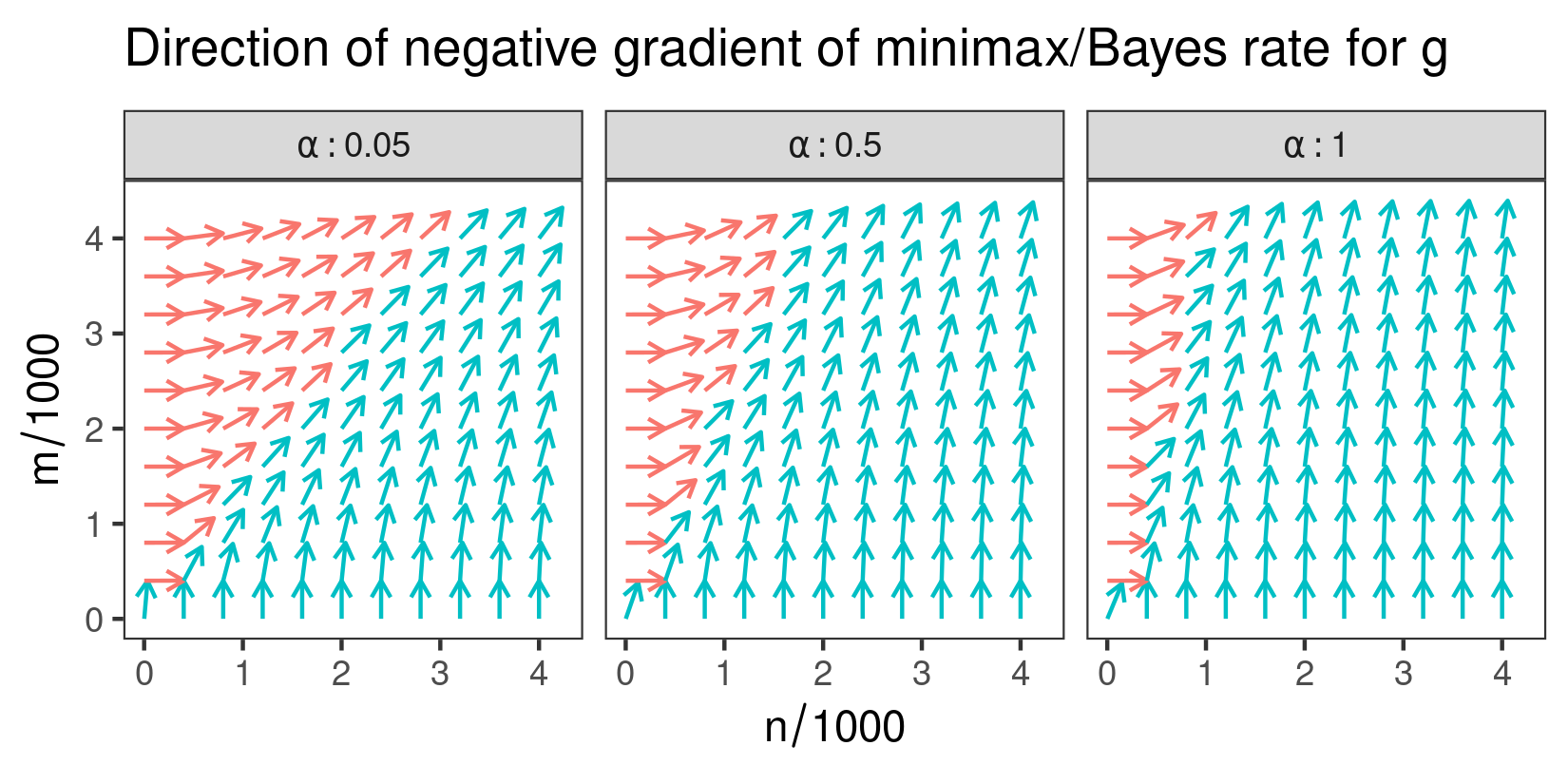}
\centering
\caption{Direction of the negative gradient of the Bayes rate $m^{-1/2}+(nm)^{-\alpha/(1+2\alpha)}$ (i.e., the direction of the greatest decrease of the rate) for estimating $g$. Arrow color indicates whether arrow slope is above or below 45 degrees.}
\label{fig:negativegradientg}
\end{figure}

Despite this, a sampling budget should still be spent partially on increasing $n$ even if the main focus is to estimate the population-level function $g$. 
We can see this in several ways.
The Bayes risk's rate becomes $m^{-\alpha/(1+2\alpha)}$ if $n$ is fixed but improves to the parametric speed $m^{-1/2}$ if $n$ grows with $m$ at least as fast as $m^{1/(2\alpha)}$ (i.e., if $n = n_m \gtrsim m^{1/(2\alpha)}$).
Also, Figure~\ref{fig:negativegradientg} shows fairly large regions of $(n,m)$ values where the direction of the rate's negative gradient is below 45 degrees, i.e., where the rate would decrease more by an infinitesimal increase of $n$ than by the same infinitesimal increase of $m$.
Finally, in the simulation study in Section~\ref{sec:study2}, the empirical out-of-sample risk of estimating $g$ decreases more by a unit increase of $n$ than by a unit increase of $m$ for many values of $(n,m,\alpha)$ (e.g., $(1,10,1)$ or $(1,100,0.05)$ in Figure~\ref{fig:pthmseg_fixta}).
Thus when designing a study where the main focus is to estimate the population-level function and the cost of a unit increase of $n$ is the comparable to the cost of a unit increase of $m$, the value of $n$ should be smaller than $m$, but not too much so. We note that in different applications, the cost of recruiting additional subjects might be higher (or lower) than measuring each individual subject more deeply; in order to arrive at the appropriate sampling design in such cases, one could apply similar reasoning through computing the risk gradients with respect to the unit cost for increasing $m$ and $n$ rather than to $m$ and $n$ directly.

Now we consider model \eqref{eq:gwnmodelfreq}, which assumes the population-level function $g^*$ is fixed, and derive the minimax risk for estimating $g^*$ over a Sobolev ball $\cS(\alpha,R)$ as defined in \eqref{eq:ellipsoid}.
We do so by providing matching lower and upper bounds up to a multiplicative constant. 
The following theorem provides a minimax lower bound that equals the Bayes risk's rate if the ball's radius grows logarithmically with $n$.

\begin{theorem} \label{thm:lbminimaxg}
    Under model \eqref{eq:gwnmodelfreq}, we have
    \begin{equation*}  \label{eq:minimaxrisklowerboundg}
    \inf_{\hg} \sup_{g^* \in \cS(\alpha, 2 \sqrt{\log n})} \E_{g^*}\|\hg-g^*\|_2^2
    \gtrsim 
     m^{-1}+(nm)^{-\frac{2\alpha}{1+2\alpha}}.
    \end{equation*}
\end{theorem}

\begin{remark}
The radius of the Sobolev ball contains a logarithmic multiplicative factor, which is of technical nature coming from our proof technique. Following the standard approach the minimax risk was bounded from below by the Bayes risk. To connect the Bayesian and frequentist models we used the GP prior on $g$ as given in model \eqref{eq:gwnmodelbayes}. This, however, results in a $\sqrt{\log n}$ factor coming from the tail behavior of the GP prior. The results could be made sharper by considering a Rademacher prior on $g$, see e.g., \cite{gine}. Here we do not follow this route as we wanted to exploit the connection between models \eqref{eq:gwnmodelbayes} and \eqref{eq:gwnmodelfreq}. 
\end{remark}

To prove that the above lower bound is sharp, we show it can be attained with or without prior knowledge about the underlying smoothness of the functions. We start by showing that two simple estimators can achieve this lower bound when the underlying smoothness is known, followed by building an adaptive estimator that accomplishes this without assuming such prior knowledge.
The Bayesian estimator is the posterior mean of $g$ given the observations and prior in model \eqref{eq:gwnmodelbayes}. 
An alternative estimator we consider is a thresholding estimator of the form
\begin{align} \label{eq:estg}
\tg^{K}(t) = \sum_{k=1}^K \hg_k \psi_k(t),
\end{align}
where $\hat{g}_k = m^{-1}\sum_{j=1}^{m} Y^{(j)}_k$ and $Y_k^{(j)} \coloneqq \langle Y^{(j)}, \psi_k \rangle = \int_0^1 Y^{(j)}(t) \psi_k(t) \dif t$ for each $j=1,\ldots,m$.
(In the left hand side of~\eqref{eq:estg}, we put the threshold $K$ in the superscript to avoid confusing the thresholding estimator $\tg^{K}(\cdot)$ with the coefficient estimator $\hg_K$ in~\eqref{eq:estg}.)
The following theorem provides the risk of the Bayesian and thresholding estimators maximized over the Sobolev ball.

\begin{theorem} \label{thm:upperboundg}
    If $\tilde{g}_{\alpha}$ is either the posterior mean associated with the prior $\pi_{\Lambda}$ defined in model \eqref{eq:gwnmodelbayes} or the estimator \eqref{eq:estg} using the threshold $K \asymp (nm)^{1/(1+2\alpha)}$, then under model \eqref{eq:gwnmodelfreq} we have
    \begin{align*}
        \sup_{g^* \in \cS(\alpha,R)} \E_{g^*} \|\tilde{g}_{\alpha}-g^*\|_2^2
        \lesssim 
        m^{-1} + (nm)^{-\frac{2\alpha}{1+2\alpha}}.
    \end{align*}
\end{theorem}

This upper bound matches the lower bound in Theorem~\ref{thm:lbminimaxg} up to a logarithmic factor (the radius was set to $\sqrt{\log n}$ in Theorem~\ref{thm:lbminimaxg} and to a constant in Theorem~\ref{thm:upperboundg}). We conjecture, however, that the rate does not contain the logarithmic term.

Both of the above estimators assume knowledge of the regularity hyperparameter $\alpha$ of the underlying true population-level function. However, this information is typically not available in practice. Furthermore, if the regularity hyperparameter is misspecified in either estimator and if the MISE of the estimator is maximized over $g$ in the Sobolev ball, the estimator's error rate becomes of larger order than the minimax rate; further detail is in Sections S.8 and S.11 of the Supplementary Material. Therefore, adaptive, data-driven choices should be considered for $\alpha$ or for the optimal threshold $K$ relying on it. We propose the following estimator based on Lepskii's method \citep{lepskii1992asymptotically,lepskii1993asymptotically} adapted to the present, hierarchical setting. 
We define the data-driven threshold
\begin{equation}  \label{eq:thresholdadaptg}
\hk = \min\, \left\{k: \forall\ell \in \left(k, \sqrt{nm}\right],\, \|\tilde{g}^{\ell} - \tilde{g}^k\|_2^2  \leq \tau \frac{\ell}{nm} \right\}
\end{equation}
for some $\tau > 6$. 
This threshold is well defined, seeing as the set on the right hand side always contains $\sqrt{nm}$ and hence is never empty.
The following theorem implies the minimax rate is achieved by the estimator~\eqref{eq:estg} with this data-driven threshold without any knowledge of $\alpha$ (or $\tilde\alpha$), which makes it an adaptive estimator.

\begin{theorem}  \label{thm:adaptiveg}
Take $\alpha > 0$, and recall the regularity parameter $\tilde\alpha$ of the covariance function $\tilde\Lambda$ from Section~\ref{sec:samplingmodel}.
If $m=O(\exp\{cn^{1/(1+2A)}\})$ for some $c>0$ and $A>\max\{\alpha,\tilde\alpha\}$, then there exists a constant $D_{\alpha,\tilde\alpha}>0$ such that under model \eqref{eq:gwnmodelfreq} the estimator \eqref{eq:estg} using the data-driven threshold \eqref{eq:thresholdadaptg} satisfies 
\begin{align*}
\sup_{g^* \in \cS(\alpha, R)} 
\E_{g^*}\|\tilde{g}_{\hk}-g^*\|_2^2 \leq D_{\alpha,\tilde\alpha} \left(m^{-1} + (nm)^{-\frac{2\alpha}{1+2\alpha}}\right).
\end{align*}
\end{theorem}

\subsection{Estimating the subject-specific function} \label{sec:estf}

Similar to Section~\ref{sec:estg}, we first find the Bayes risk for estimating the subject-specific function $\fo$ when the population-level function $g$ has a mean-zero GP prior distribution with covariance function $\Lambda$ from model~\eqref{eq:gwnmodelbayes}.
The following theorem provides the asymptotic rate of squared Bayes risk as well as an estimator that achieves this rate. 
\begin{theorem}  \label{thm:lbbayesriskf}
Consider the hierarchical model~\eqref{eq:gwnmodelbayes}. Then the minimum of the average MISE (with respect to the distribution of $g\sim\pi_{\Lambda}$) for estimating $f^{(1)}$ is equal to
    \begin{equation*}  \label{eq:bayesrisklowerboundf1}
    \inf_{\hfo} \int \E_g \|\hfo-f^{(1)}\|_2^2 \dif \pi_{\Lambda}(g)
    =  \int \E_g \|\tilde{f}_{\Lambda, \tilde\Lambda}-f^{(1)}\|_2^2 \dif \pi_{\Lambda}(g)
    \asymp  n^{-\frac{2\tilde\alpha}{1+2\tilde\alpha}}+(nm)^{-\frac{2\alpha}{1+2\alpha}},
    \end{equation*}
where $\tilde{f}_{\Lambda, \tilde\Lambda}$ denotes the posterior mean for $\fo$ given the data $Y^{(1)}, \ldots, Y^{(m)}$ and prior $\pi_{\Lambda}$. 
\end{theorem}

The Bayes risk's rate is thus $n^{-\tilde\alpha/(1+2\tilde\alpha)}+(nm)^{-\alpha/(1+2\alpha)}$, which we see cannot be made arbitrarily small by fixing $n$ and appropriately increasing $m$, seeing as infinitely many noisy subject-specific functions allows for much variation in what the subject-specific deviation $(\fo-g)\mid g$ could be.
Note also that $f^{(1)}$ is $\{\alpha\wedge\tilde\alpha\}$-smooth; hence the corresponding minimax rate using the single-subject data $Y^{(1)}$ is $n^{-\{\alpha\wedge\tilde\alpha\}/(1+2\{\alpha\wedge\tilde\alpha\})}$. We compare this with the above derived multi-subject rate to see how learning $g$ aids with learning $f^{(1)}$.
If $\alpha \geq \talpha$, i.e., if $f^{(1)}$'s roughness comes from the covariance function $\tLambda$ rather than from $g$, the two rates equal $n^{-\tilde\alpha/(1+2\tilde\alpha)}$, which implies estimating $g$ does not help with estimating $f^{(1)}$.
If $\alpha < \talpha$, i.e., if $f^{(1)}$'s roughness comes from $g$, the single-subject rate becomes $n^{-\alpha/(1+2\alpha)}$ which is strictly larger than either $n^{-\tilde\alpha/(1+2\tilde\alpha)}$ or $(nm)^{-\alpha/(1+2\alpha)}$ and hence implies that estimating $g$ helps with estimating $f^{(1)}$.
In this region, we now examine the degree to which learning $g$ aids with learning $f^{(1)}$. 
As such, let $\delta \coloneqq \log m / \log n$ and suppose $\talpha > \alpha$. 
Then $n^{-\tilde\alpha/(1+2\tilde\alpha)} > (nm)^{-\alpha/(1+2\alpha)}$ if and only if $\alpha \{1+\delta+2\talpha\delta\} > \talpha > \alpha$, which means increasing the number of subjects will help to learn $f^{(1)}$ only up to a certain point.
Specifically, any additional increase in $m$ beyond $n^{(\talpha\alpha^{-1}-1)/(1+2\talpha)}$ will not improve the estimation of $f^{(1)}$.

\begin{figure}
\includegraphics[width=0.7\textwidth]{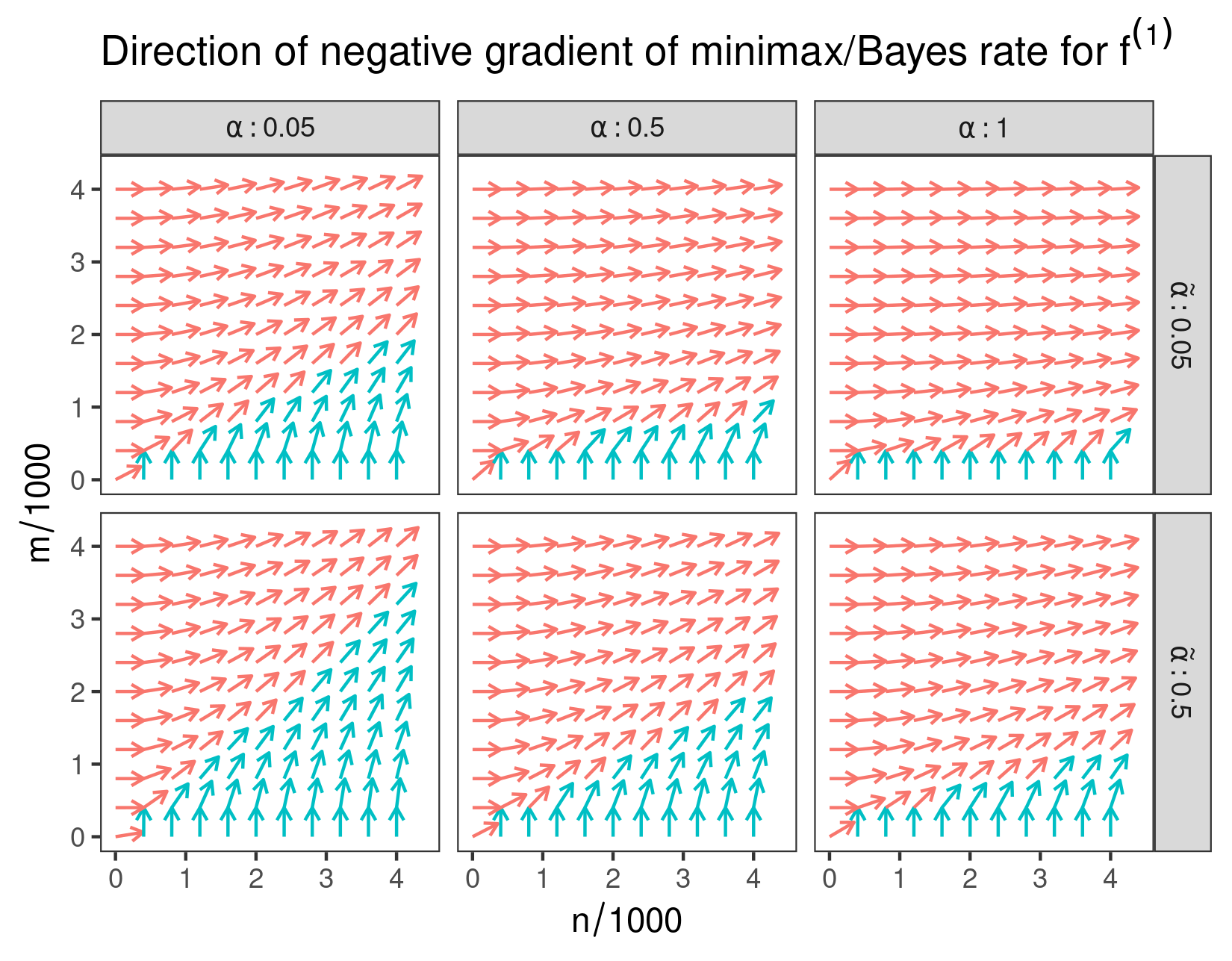}
\centering
\caption{Direction of the negative gradient of the Bayes rate $n^{-\tilde\alpha/(1+2\tilde\alpha)}+(nm)^{-\alpha/(1+2\alpha)}$ (i.e., the direction of the greatest decrease of the rate) for estimating $f^{(1)}$. Arrow color indicates whether arrow slope is above or below 45 degrees.}
\label{fig:negativegradientf}
\end{figure}

Similar to the previous section, the previous point implies a sampling budget should still be spent partially on increasing $m$ even if the main focus is to estimate the one subject-specific function $\fo$.
We emphasize that this does not even take into account that increasing $m$ allows us to learn about additional subjects (e.g., patients) whereas increasing $n$ has no such benefit. To bolster this argument, 
Figure~\ref{fig:negativegradientf} shows sizeable regions of $(n,m)$ values where the direction of the rate's negative gradient is above 45 degrees, i.e., where the rate would decrease more by an infinitesimal increase of $m$ than by the same increase of $n$.
This behavior is corroborated by the simulation study in Section~\ref{sec:study2}, where the empirical out-of-sample risk of estimating $\fo$ decreases more by a unit increase of $m$ than by a unit increase of $n$ for many values of $(n,m,\alpha,\tilde\alpha)$ (e.g., $(1,1000,0.2,0.05)$ or $(1,100,0.2,0.5)$ in Figure~\ref{fig:pmsef_fixta}).
Thus when designing a study where the main focus is to estimate the subject-specific function $\fo$ and the cost of a unit increase of $n$ is the similar to that for $m$, the value of $m$ should be smaller than $n$, but not too much so. Again, one can generalize the reasoning to cases in which the sampling cost for a unit increase in $m$ and $n$ are different through computing the risk gradient in terms of the corresponding unit cost.

Now we find the minimax risk for estimating $\fo$ for fixed $g^*$ in the Sobolev ball $\cS(\alpha,R)$ as defined in \eqref{eq:ellipsoid}.
As in the previous section, we do so by providing matching lower and upper bounds up to a multiplicative constant. 
The following theorem provides a minimax lower bound that equals the Bayes risk's rate if the ball's radius grows logarithmically with $n$. 

\begin{theorem} \label{thm:lbminimaxf}
    Under model \eqref{eq:gwnmodelfreq}, we have
    \begin{equation*}  
    \inf_{\hf} \sup_{g^* \in \cS(\alpha, 2 \sqrt{\log n})} \E_{g^*}\|\hf-\fo\|_2^2
    \gtrsim 
    n^{-\frac{2\tilde\alpha}{1+2\tilde\alpha}}+(nm)^{-\frac{2\alpha}{1+2\alpha}}.
    \end{equation*}
\end{theorem}

The remark about the radius from the previous section also applies here. 

To show that the above lower bound is sharp, we again construct estimators that attain this lower bound, first with prior knowledge about the underlying smoothness then without such knowledge.
The first is the posterior mean of $\fo$ given the observations and prior in model~\eqref{eq:gwnmodelbayes}. 
The second is a double-thresholding estimator of the form
\begin{align} \label{eq:estf}
\hfo_{k_1,k_2}(t) = \sum_{k=1}^{k_1} Y^{(1)}_k \psi_k(t) +\sum_{k=k_1+1}^{k_2}\hat{g}_k \psi_k(t), 
\end{align}
where $Y_k^{(1)}$ is defined in \eqref{eq:estg} and $\hg_k = (m-1)^{-1}\sum_{j=2}^{m} Y^{(j)}_k$ contains data from subjects $j=2,\ldots,m$. 
(To improve the legibility of the thresholds, we will sometimes denote by $\hfo(k_1,k_2)$ the double-thresholding estimator with thresholds $k_1$ and $k_2$ if we are thinking of it as a function yet to be evaluated.)
The following theorem provides the risk of these two estimators maximized over the ball.

\begin{theorem} \label{thm:upperboundf}
    If $\tf_{\talpha,\alpha}$ is either the posterior mean of $\fo$ associated with the prior $\pi_{\Lambda}$ or the estimator \eqref{eq:estf} using the thresholds $k_1 \asymp n^{1/(1+2\talpha)}$ and 
    \begin{align*}
        k_2 \asymp \left(k_1 \vee (nm)^{1/(1+2\alpha)}\right) \asymp 
        \begin{cases}
            (nm)^{1/(1+2\alpha)} & \text{for } \alpha < \talpha + \delta(1/2+\talpha), \\
            n^{1/(1+2\talpha)} & \text{for } \alpha \geq \talpha + \delta(1/2+\talpha), 
        \end{cases} 
    \end{align*}
    then under model~\eqref{eq:gwnmodelfreq} we have
    \begin{align*}
        \sup_{g^* \in \cS(\alpha,R)} \E_{g^*} \| \tf_{\alpha,\talpha}-\fo \|_2^2
        &\lesssim 
        (nm)^{-\frac{2\alpha}{1+2\alpha}} + n^{-\frac{2\talpha}{1+2\talpha}}
    \end{align*}
    where $\cS(\alpha,R)$ is a Sobolev ball defined in \eqref{eq:ellipsoid} and $\delta \coloneqq \log m / \log n$.
\end{theorem}

The upper bound matches the lower bound in Theorem~\ref{thm:lbminimaxf} (up to a logarithmic factor), and hence the two estimators from Theorem~\ref{thm:upperboundf} achieve the minimax rate.
But this requires knowing what values of $\alpha$ and $\tilde\alpha$ appropriately characterize, respectively, the population-level function and the subject-specific deviations. 
Furthermore, if either regularity hyperparameter is misspecified in either estimator and if the MISE of the estimator is maximized over $g^*$ in the Sobolev ball, the estimator's error rate becomes polynomially larger order than the minimax rate; further detail is contained in Sections S.6 and S.10 of the Supplementary Material.
Therefore, to establish an estimator that achieves the minimax rate without being provided $\alpha$ or $\tilde\alpha$, we define the data-driven thresholds 
\begin{equation}\label{eq:thresholdadaptf}
\begin{split}
\hk_2 &= \min\, \left\{k: \forall\ell \in \left(k, \sqrt{nm}\right],\, \|\hf^{(1)}(1,k)-\hf^{(1)}(1,\ell) \|_2^2\leq \tau_2\frac{\ell}{nm} \right\}\\
\hk_1^{(1)} &= \min\, \left\{k: \forall \ell \in \left(k, \sqrt{n}\right],\, \|\hf^{(1)}(k,\sqrt{n})-\hf^{(1)}(\ell,\sqrt{n}) \|_2^2\leq \tau_1\frac{\ell}{n} \right\},
\end{split}
\end{equation}
for some positive $\tau_1 > 4$ and $\tau_2 > 6$.
Both $\hk_2$ and $\hk_1^{(1)}$ are well defined, seeing as the corresponding sets on the right hand side of the preceding display are not empty (they contain $\sqrt{nm}$ and $\sqrt{n}$, respectively).
Though $\hk_1^{(1)}$ is specific to subject $j=1$, the threshold $\hk_2$ is the same for all subjects and hence does not change if, for example, we are estimating $f^{(2)}$ rather than $f^{(1)}$.
The following theorem states that the estimator~\eqref{eq:estf} with these data-driven thresholds can achieve the minimax rate without any knowledge of $\alpha$ or $\tilde\alpha$.

\begin{theorem}  \label{thm:adaptivef}
Take $\alpha > 0$, and recall the regularity parameter $\tilde\alpha$ of the covariance function $\tilde\Lambda$ from Section~\ref{sec:samplingmodel}.
If $m=O(\exp\{cn^{1/(1+2A)}\})$ for some $c>0$ and $A>\max\{\alpha,\tilde\alpha\}$, then there exists a constant $D_{\alpha,\tilde\alpha}>0$ such that under model \eqref{eq:gwnmodelfreq} the estimator \eqref{eq:estf} using the data-driven thresholds \eqref{eq:thresholdadaptf} satisfies 
\begin{align*}
\sup_{g^* \in \cS(\alpha, R)} \E_{g^*} \|\hf(\hk_1^{(1)},\hk_1^{(1)}\vee\hk_2)-\fo\|_2^2
\leq D_{\alpha,\tilde\alpha} \left(n^{-\frac{2\tilde\alpha}{1+2\tilde\alpha}}+(nm)^{-\frac{2\alpha}{1+2\alpha}}\right).
\end{align*}
\end{theorem}

\section{Simulation results}
\label{sec:studyall}

We next carry out two simulation studies.
As mentioned in Section~\ref{sec:mainresults}, the white noise model is the idealized version of the nonparametric regression model.
Here each study generates data from the following nonparametric regression model which is a modification of the Gaussian white noise model~\eqref{eq:gwnmodelbayes}.
Assuming $m$ subjects and $n$ observations per subject, for each subject $j=1,\ldots,m$ and observation number $i=1,\ldots,n$ we observe the random variable $Y^{(j)}_i$ generated from the following model
\begin{align*}
\begin{split}
    Y^{(j)}_i &= f^{(j)}(t_i^{(j)}) + Z^{(j)}_i, \qquad Z_i^{(j)} \iid \text{N}(0,1) \\
    f^{(1)}, \ldots, f^{(m)} \mid g, \tLambda &\iid \text{GP}(g, \tLambda), \\ 
    g \mid \Lambda &\sim \text{GP}(0, \Lambda),
\end{split}
\end{align*} 
where the grid points $t_i^{(j)} \in [0,1]$ are fixed. 
Here we use fixed, equidistant grid points to drastically reduce the number of large covariance matrices to compute while still ensuring each observation $Y_i^{(j)}$ is evaluated at its own unique grid point.
(Random grid points would require computing a new covariance matrix for each replicate data set generated from the above model.)
Both studies use the Fourier basis as the set of eigenfunctions $\psi_k$ and covariance matrices with entries $\Sigma^{(\beta)}_{i,j} = \sum_{k=1}^{N} k^{-1-2\beta} \psi_k(i/N) \psi_k(j/N)$ for $i,j=1,\ldots,N$ and various values of $\beta$, where $N=20000$.

Each study will use the thresholding estimators \eqref{eq:estg} and \eqref{eq:estf} whose initial definitions assumed the observations $Y^{(1)}, \ldots, Y^{(m)}$ were functions. 
In view of the nonparametric regression model above, we must modify the definitions of $\hat{g}_k$ and $Y_k^{(j)}$, which we do by replacing the inner product $\langle Y^{(j)}, \psi_k \rangle \coloneqq \int_0^1 Y^{(j)}(t) \psi_k(t) \dif t$ with the empirical version $\langle Y^{(j)}, \psi_k \rangle_n^{(j)} \coloneqq \sum_{i=1}^n Y^{(j)}_i \psi_k(t_i^{(j)})$.
In the following numerical studies, the thresholding estimators \eqref{eq:estg} and \eqref{eq:estf} are defined using the modified $\hat{g}_k$ and $Y_k^{(j)}$.

\subsection{Study 1}
\label{sec:study1}

This study illustrates the performance of the threshold estimator \eqref{eq:estg} for $g$ with adaptive threshold \eqref{eq:thresholdadaptg} and of the two-threshold estimator \eqref{eq:estf} for $\fo$ with adaptive thresholds \eqref{eq:thresholdadaptf}.
For estimating $g$ we compare the performance of the adaptive estimator to that of the estimator \eqref{eq:estg} with one of four fixed thresholds. 
For estimating $\fo$ we compare the performance of the two-threshold adaptive estimator to that of an estimator that uses data only from subject $j=1$. 

This study uses $\talpha=0.5$ and $\alpha \in \{0.05, 0.2, 0.5, 2\}$, which provides values less than, equal to, and greater than $\talpha$; as such we compute the four $N \times N$ covariance matrices $\Sigma^{(0.05)}$, $\Sigma^{(0.2)}$, $\Sigma^{(0.5)}$, and $\Sigma^{(2)}$ exactly once in order to generate realizations of GPs. 
This study also uses the pairs $(n,m) \in \{(20,500), (50,200), (100,100), (200,50), (500,20)\}$, where we see $nm=N/2$.
For each combination of $\alpha$, $(n,m)$, and $r \in \{1,\ldots,200\}$ (here $r$ indexes the 200 replicates), we produce a dataset $\sD_{\alpha, \talpha, n, m, r}$ as follows. 
First the realization $g_{\alpha, r}$ of $\text{GP}(0,\Lambda)$ is evaluated at the grid points $\{i/N:i=1,\ldots,N\}$;
to reduce computation time and to facilitate comparison, the same (discrete) realization of $g_{\alpha, r}$ is used regardless of the value of $(n,m)$. 
Then for each $j \in \{1,\ldots,m\}$ we use the realization of $g_{\alpha, r}$ to evaluate the function $\fj_{\alpha,n,m,r}$ at the grid points $\sT^{train}_{n,m,j} \cup \sT^{eval}$, 
where $\sT^{train}_{n,m,j} \coloneqq \{2(mi + j) / N: i=0,\ldots,n-1\}$ for each $j$ and $\sT^{eval} \coloneqq \{(20i+1)/N: i=0,\ldots,999\}$ are all mutually disjoint subsets of the set $\{i/N:i=1,\ldots,N\} \subset [0,1]$. 
Finally, for each $j \in \{1,\ldots,m\}$, we use the realized values of $\fj_{\alpha,n,m,r}$ to generate the observations $Y_{i,\alpha,\talpha,n,m,r}^{(j)}$ for all $i \in \sT^{train}_{n,m,j}$.
Thus each dataset $\sD_{\alpha, \talpha, n, m, r}$ contains $nm=10000$ observed values that the estimators are trained on. 

\begin{figure}
\includegraphics[width=0.85\textwidth]{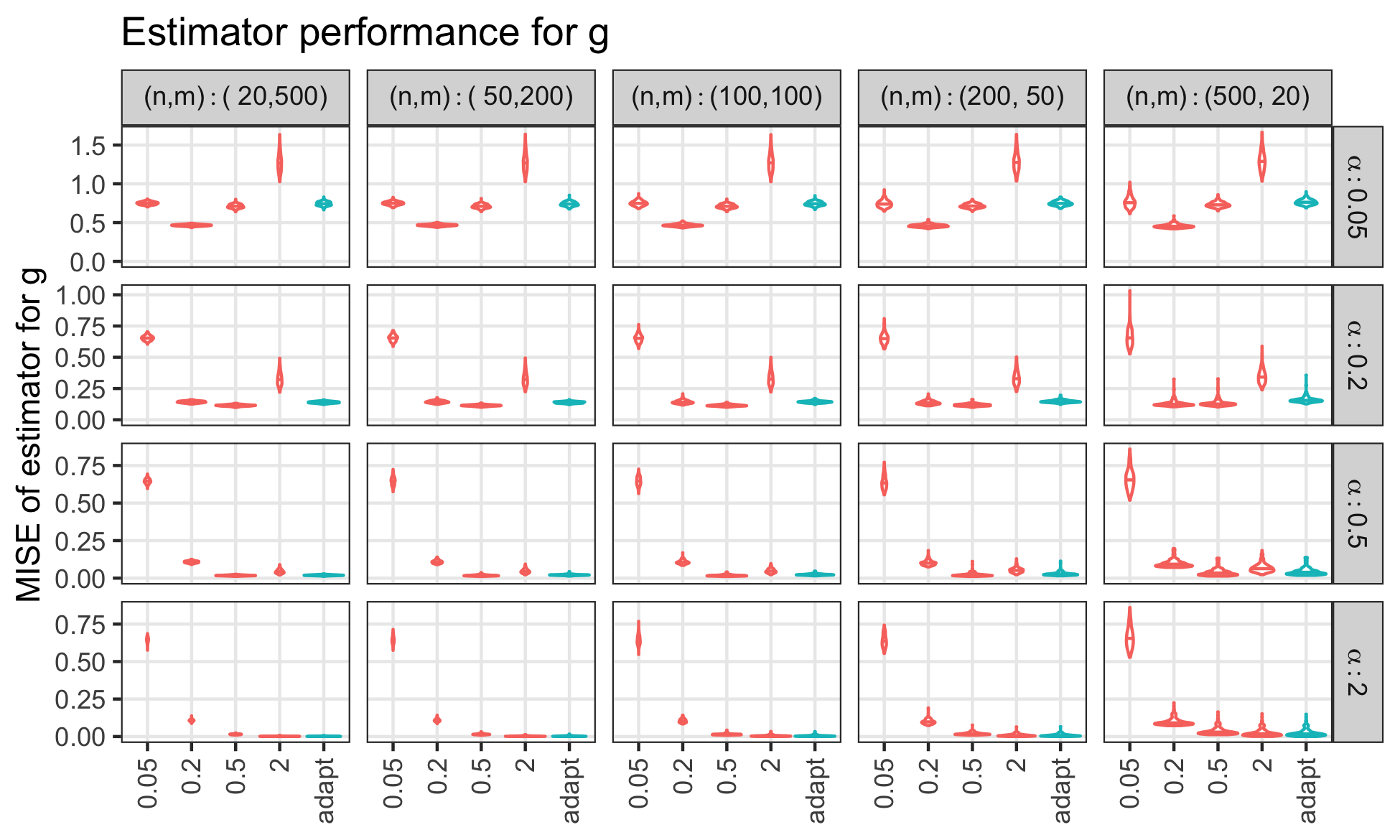}
\centering
\caption{The out-of-sample MISE \eqref{eq:emiseg} of the thresholding estimator \eqref{eq:estg} for the threshold $K = (nm)^{1/(1+2\beta)}$ for all $\beta \in \{0.05,0.2,0.5,2\}$ and for the adaptive threshold \eqref{eq:thresholdadaptg}. Here $\tilde\alpha = 0.5$. The variability comes from the 200 datasets for each $(\alpha,n,m)$ combination.}
\label{fig:pmseg_n100_m100}
\end{figure}

For estimating $g$ we compare the adaptive estimator to the nonadaptive estimator \eqref{eq:estg} with fixed threshold $K_{n,m} = (nm)^{1/(1+2\beta)}$ for all $\beta \in \{0.05,0.2,0.5,2\}$ so that for any of the four $\alpha$ values, exactly one of the fixed thresholds produces an estimator that achieves the minimax rate.
For each dataset $\sD_{\alpha, \talpha, n, m, r}$ and each of the five estimators (denoted by $\tilde{g}$) trained on the observations in $\sD_{\alpha, \talpha, n, m, r}$, we compute the empirical MISE of $\tilde{g}$ over $N/2=10000$ unobserved grid points:
\begin{align} \label{eq:emiseg}
    \text{MISE}(\tilde{g}, g_{\alpha,r}) = 
    (10000)^{-1} \sum_{i=1}^{10000} \big[\tilde{g}(t_i) - g_{\alpha,r}(t_i)\big]^2, \quad t_i=(i-0.5)/10000.
\end{align}

The (empirical) MISEs are shown in Figure~\ref{fig:pmseg_n100_m100}, where the variability comes from the 200 datasets for each $(\alpha,n,m)$ combination.
First we see for any $\alpha$ that although the value of $(n,m)$ does not seem to affect the median MISE of any estimator, the spread of any estimator's MISE increases as $m$ decreases, though to an increasingly smaller degree as $\alpha$ decreases.
Second, the $\beta=0.2$ and $\beta=0.5$ thresholding estimators never perform too poorly in any of the investigated cases, despite the thresholds not depending on the data. (We emphasize that their performance is specific to this particular simulation specification; they might perform worse for different values of $n$ and $m$.)
Finally, we find the adaptive threshold is somewhere between 100 and 6 (i.e., between the thresholds induced by $\beta=0.5$ and $\beta=2$), which turns out to be too small when $\alpha=0.05$ or $\alpha=0.2$; this could be remedied by decreasing the chosen value of the tuning parameter $\tau$ in \eqref{eq:thresholdadaptg}.

\begin{figure}
\includegraphics[width=0.85\textwidth]{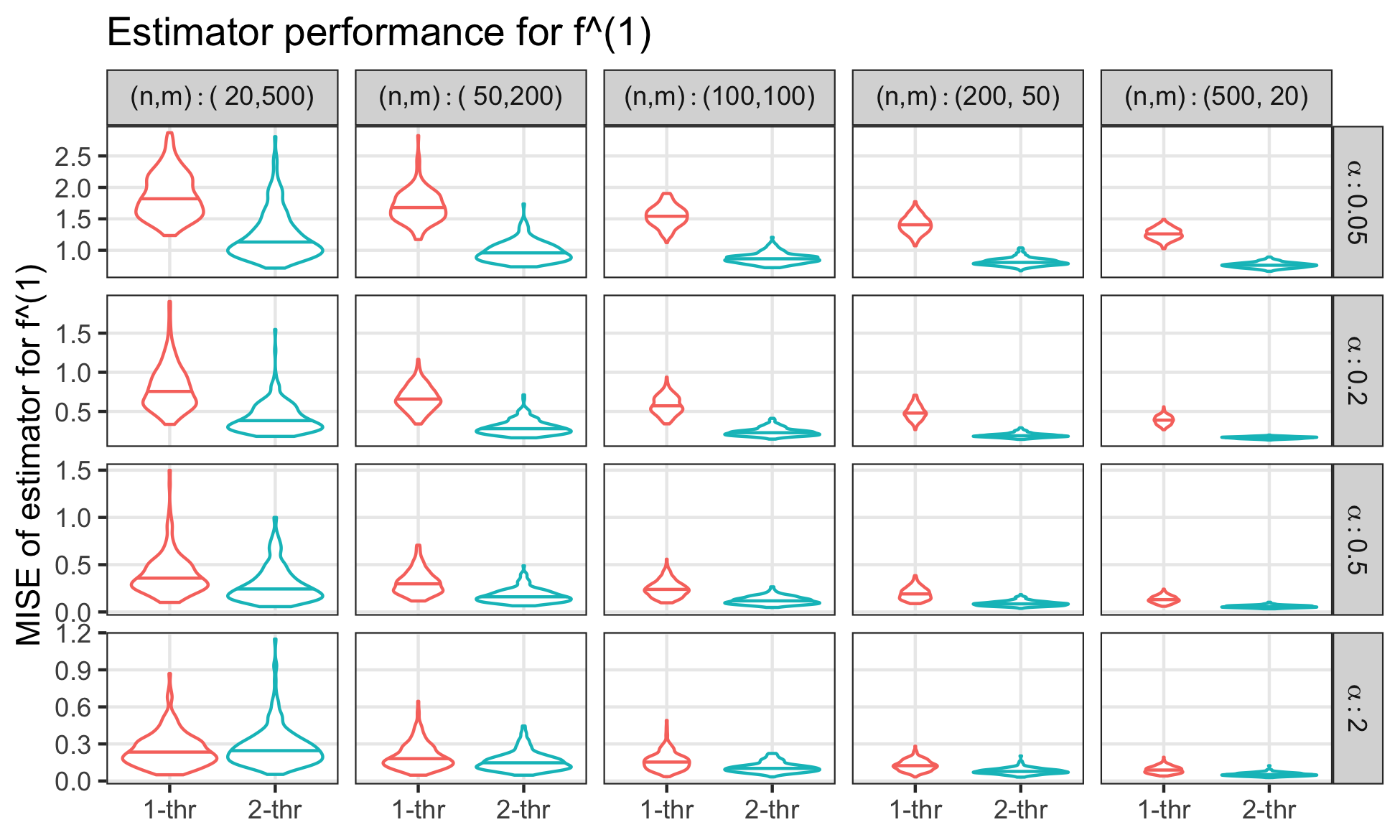}
\centering
\caption{The out-of-sample MISE \eqref{eq:emisef} of the one-threshold (i.e., single-subject) estimator \eqref{eq:est1f} and the two-threshold estimator \eqref{eq:estf} using the adaptive thresholds \eqref{eq:thresholdadaptf}.  Here $\tilde\alpha = 0.5$. The variability comes from the 200 datasets for each $(\alpha,n,m)$ combination.}
\label{fig:pmsef_n100_m100}
\end{figure}

To illustrate how much the estimation of $\fo$ improves by the addition of data from $m-1$ other subjects, we compare the two-threshold adaptive estimator (which uses data from all $m$ subjects) to the single-subject estimator 
\begin{align} \label{eq:est1f}
\hf_{\hat{k}_1}(t) = \sum_{k=1}^{\hat{k}_1} Y^{(1)}_k \psi_k(t), \text{ where } 
\hat{k}_1 = \min\left\{k: \forall \ell\in\left(k,\sqrt{n}\right], \sum_{i=k+1}^{\ell} \big(Y_k^{(1)}\big)^2 \leq 2\frac{\ell}{nm} \right\}
\end{align}
which uses data from only subject $j=1$ and whose threshold is data-driven (i.e., assumes no knowledge of $\alpha$ or $\tilde\alpha$).
For each dataset $\sD_{\alpha, \talpha, n, m, r}$ and each of the two estimators (denoted by $\hfo$) trained on the observations in $\sD_{\alpha, \talpha, n, m, r}$, we compute the empirical MISE of $\hfo$ over the $1000$ unobserved grid points in $\sT^{eval}$:
\begin{align} \label{eq:emisef}
    \text{MISE}(\hfo, \fo_{\alpha,n,m,r}) = 
    (1000)^{-1} \sum_{t \in \sT^{eval}} \Big[\hfo(t) - \fo_{\alpha, n, m, r}(t)\Big]^2.
\end{align}

The (empirical) MISEs are shown in Figure~\ref{fig:pmsef_n100_m100}, where the variability comes from the 200 datasets for each $(\alpha,n,m)$ combination.
The two-threshold estimator performs visibly better than the single-subject estimator for most of the tested $(n,m,\alpha)$ combinations.
When $\alpha=0.05$, the single-subject estimator's MISE is visibly larger than the two-threshold estimator's, but as $\alpha$ increases, the MISE difference decreases until it is essentially nullified at $\alpha=2$. 
This observation is supported by our intuition that information about $g$ can improve the inference of $\fo$ when $g$ is rougher than any of the functions $\fj$ are, and that the inferential improvement should decrease as the roughness difference does.
It is also supported by the comparison of the multi-subject rate to the single-subject rate in Section~\ref{sec:estf}, which implies that the two theoretical rates are equal at $\alpha=2$ and $\alpha=0.5$, and that the multi-subject rate is smaller than the single-subject rate at $\alpha=0.2$ and $\alpha=0.05$.
We also see at any $\alpha$ that the MISE of either estimator decreases as $n$ increases (where we also recall that $nm$ is always $10000$), which suggests the $n^{-2\talpha/(1+2\talpha)}$ term dominates the estimator's theoretical MISE for estimating $\fo$ in this simulation scenario.

\subsection{Study 2}
\label{sec:study2}

This study considers the scenario of a practitioner (e.g., a biomedical researcher) choosing the value of $n$ and $m$ for an experiment given a fixed sampling budget.
A primal goal might be to minimize the MISE of estimators of $\fo$ and $g$. 
The dual goal might instead be to minimize the experiment's cost within a predetermined range of acceptable MISE values. 

This study uses different values of $(\alpha,\talpha)$ shown in Figures \ref{fig:pthmseg} and \ref{fig:pmsef}; these values are chosen to best illustrate certain phenomena described below.
This study also uses many different values of $(n,m)$ under a budget constraint of $nm \leq 5000$.
For each combination of $(\alpha,\talpha)$, $(n,m)$, and $r=1,\ldots,50$, we produce a dataset $\sD_{\alpha, \talpha, n, m, r}$ as described in Section~\ref{sec:study1}.

Figure~\ref{fig:pthmseg} shows the mean log MISE of the threshold estimator \eqref{eq:estg} for $g$ with adaptive threshold \eqref{eq:thresholdadaptg}, whereas Figure~\ref{fig:pmsef} shows the mean log MISE the two-threshold estimator \eqref{eq:estf} for $\fo$ with adaptive thresholds \eqref{eq:thresholdadaptf};
here ``mean log MISE'' is the mean of the respective log MISE taken over the 50 datasets. 
(Figures \ref{fig:pmsef_fixta} and Figure~\ref{fig:pmsef_fixa} focus on a smaller range of $(n,m)$ values to better distinguish changes in mean log MISE values.)
Figure~\ref{fig:pthmseg} shows that increasing $n$ can often improve the estimation of $g$.
Figure~\ref{fig:pthmseg} also shows that when $m=1$, the log MISE never dips below roughly $0.1$ regardless of $\alpha$, $\talpha$, or $n$, which reflects the $m^{-1}$ term in the minimax rate.
If we look more closely, Figure~\ref{fig:pthmseg_fixta} shows that the boundaries between colors becomes increasingly horizontal as $\alpha$ increases whereas Figure~\ref{fig:pthmseg_fixa} shows that the boundaries between colors become increasingly horizontal as $\talpha$ decreases; 
both observations match the intuition that $n$ become increasingly irrelevant to estimate the mean function $g$ as it becomes smoother than the subject-specific deviations. 
Similarly, Figure~\ref{fig:pmsef} shows that increasing $m$ can often improve the estimation of $f^{(1)}$, but Figure~\ref{fig:pmsef_fixta} shows that the boundaries between colors become increasingly \emph{vertical} as $\alpha$ increases whereas Figure~\ref{fig:pmsef_fixa} shows that the boundaries between colors become increasingly vertical as $\talpha$ decreases; 
both observations match the intuition that $m$ becomes increasingly irrelevant to estimate the subject-specific functions as they becomes rougher than the mean function. 

\begin{figure}
     \begin{subfigure}[b]{0.44\textwidth}
         \centering
         \includegraphics[width=\textwidth]{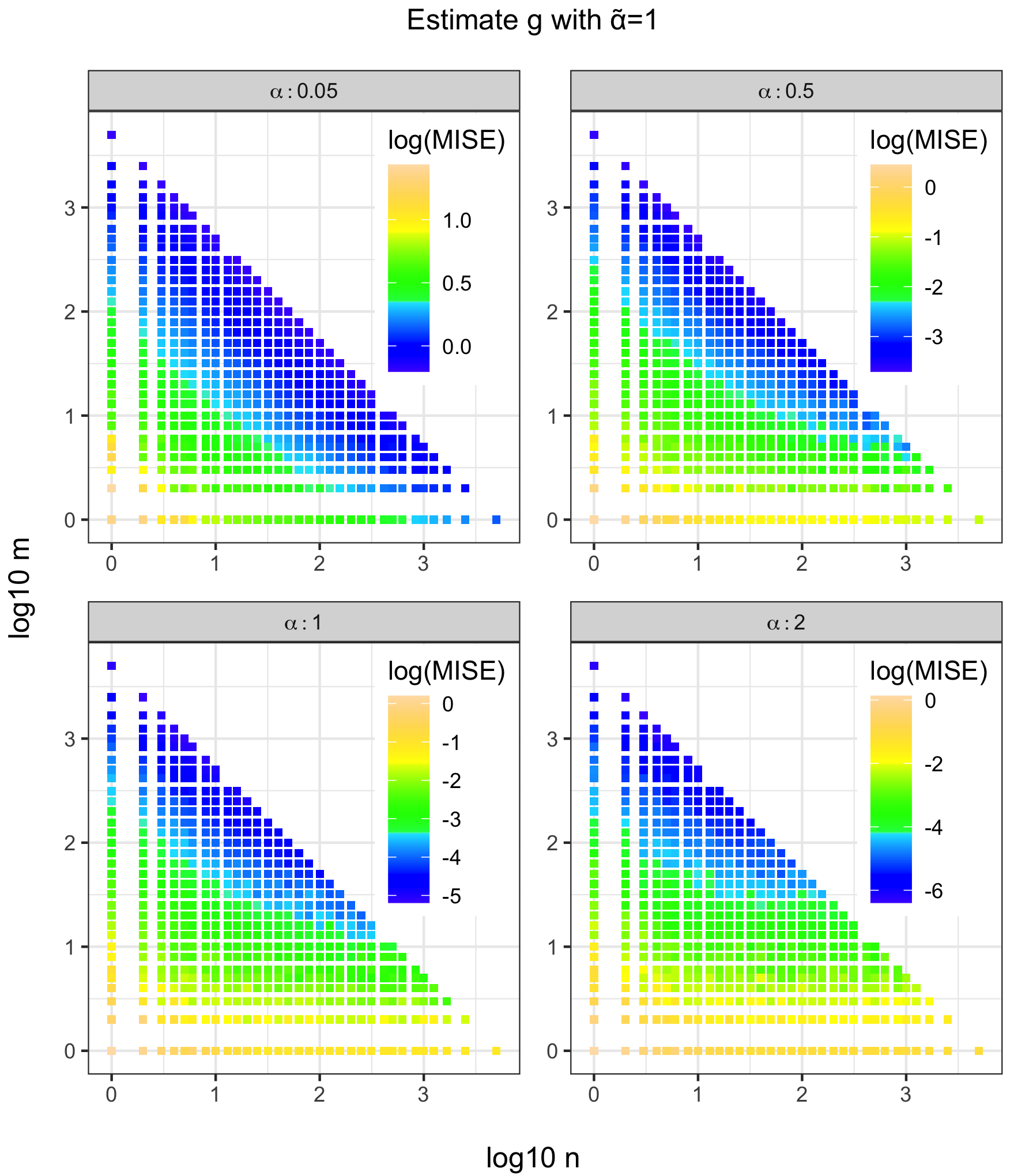}
         \caption{Fix $\talpha=1$ and vary $\alpha$.}
         \label{fig:pthmseg_fixta}
     \end{subfigure}
     \hfill
     \begin{subfigure}[b]{0.44\textwidth}
         \centering
         \includegraphics[width=\textwidth]{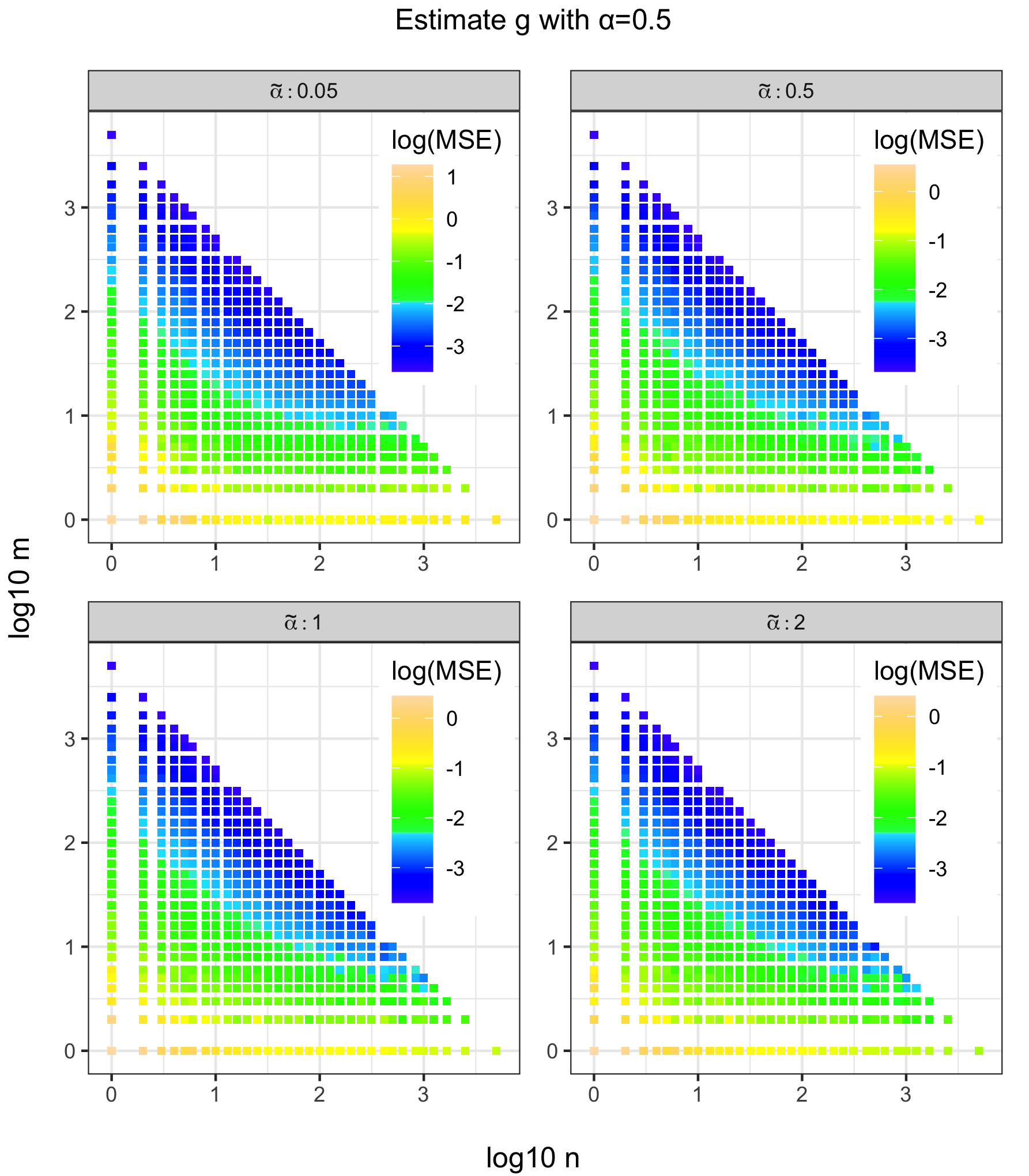}
         \caption{Fix $\alpha=0.5$ and vary $\talpha$.}
         \label{fig:pthmseg_fixa}
     \end{subfigure}
    \centering
    \caption{Mean $\log \text{MISE}$ values for the threshold estimator \eqref{eq:estg} of $g$ with adaptive threshold \eqref{eq:thresholdadaptg} for various $(n,m)$. Here ``mean $\log \text{MISE}$'' is the $\log \text{MISE}$ averaged over 50 datasets.}
    \label{fig:pthmseg}
\end{figure}

\begin{figure}
     \begin{subfigure}[b]{0.44\textwidth}
         \centering
         \includegraphics[width=\textwidth]{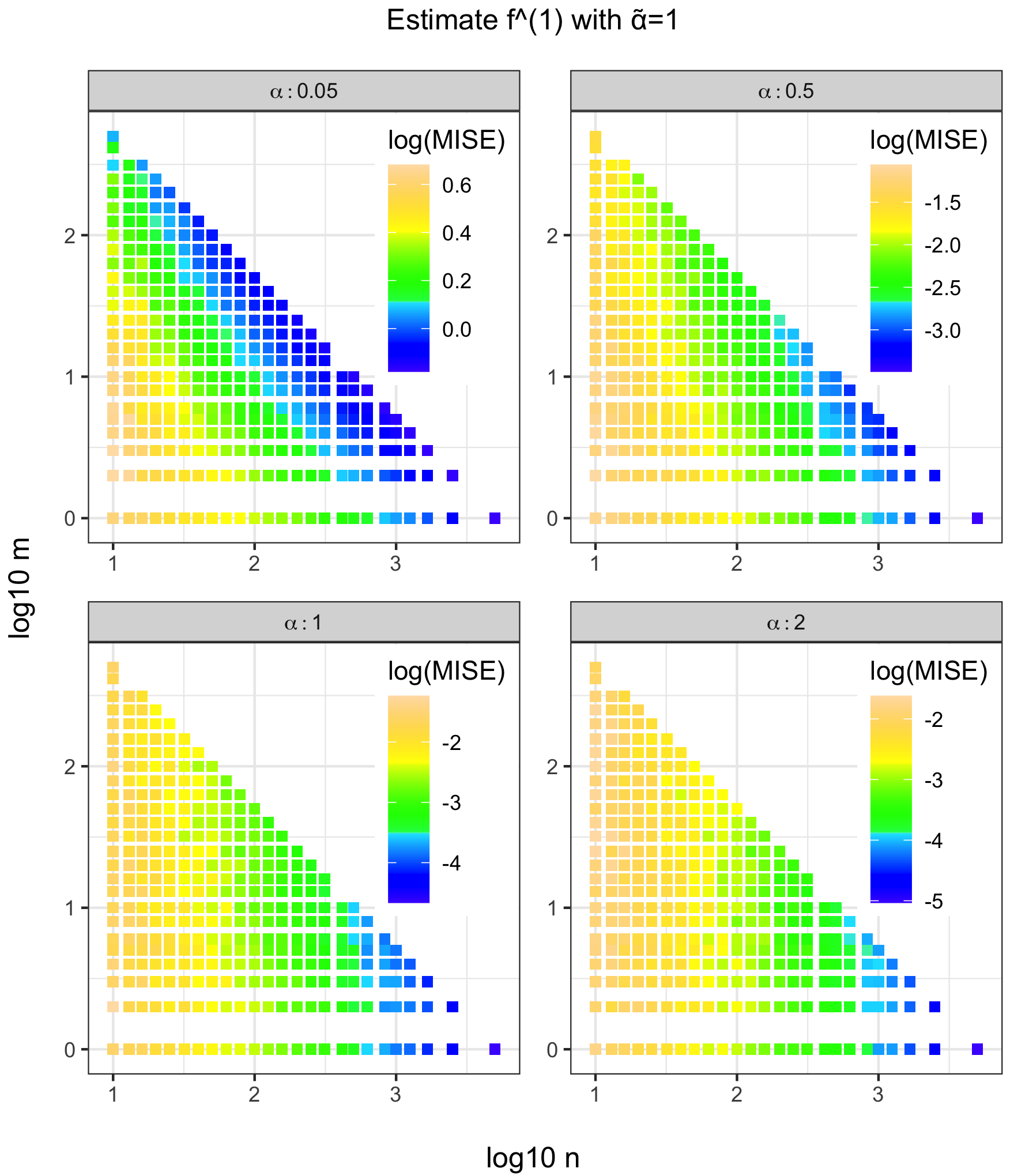}
         \caption{Fix $\talpha=1$ and vary $\alpha$.}
         \label{fig:pmsef_fixta}
     \end{subfigure}
     \hfill
     \begin{subfigure}[b]{0.44\textwidth}
         \centering
         \includegraphics[width=\textwidth]{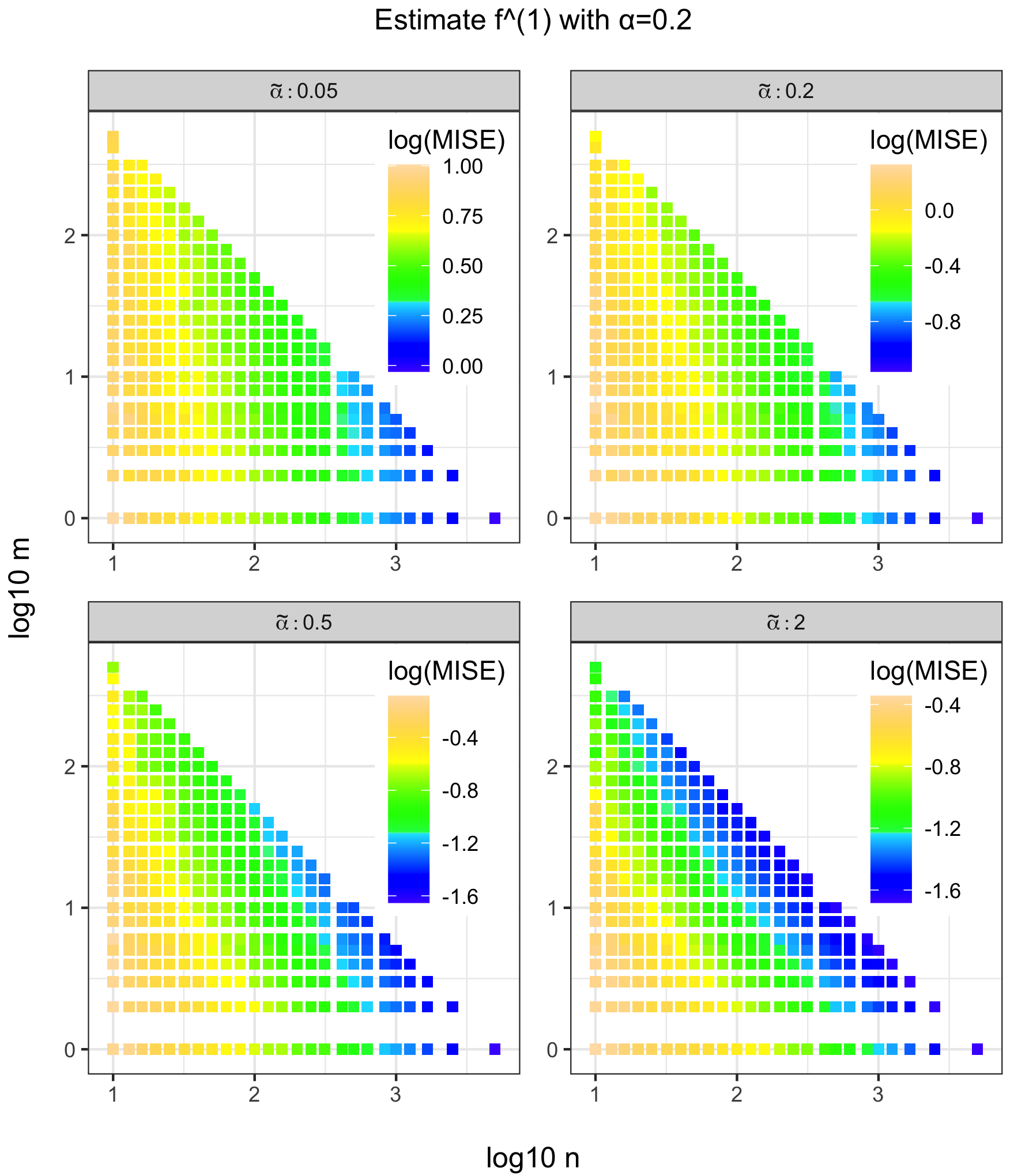}
         \caption{Fix $\alpha=0.2$ and vary $\talpha$.}
         \label{fig:pmsef_fixa}
     \end{subfigure}
    \centering
    \caption{Mean $\log \text{MISE}$ values for the two-threshold estimator \eqref{eq:estf} of $f^{(1)}$ with adaptive thresholds \eqref{eq:thresholdadaptf} for various $(n,m)$, where again the $\log \text{MISE}$ is averaged over 50 datasets.}
    \label{fig:pmsef}
\end{figure}

Now we consider how the practitioner designing the sampling scheme might use this information.
For the primal goal, she can overlay contour curves of the experiment's cost as a function of $\log n$ and $\log m$ on the appropriate subplot and then follow the contour line to the $(\log n, \log m)$ value that minimizes the mean log MISE.
For the dual goal, she can consider just the squares whose color corresponds to an acceptable mean log MISE value and then find the lowest-cost contour line that intersects one of these squares.

\section{Illustration on real data}

Similar to Section~\ref{sec:studyall}, the thresholding estimators used in this section are defined using the regression versions of $\hat{g}_k$ and $Y_k^{(j)}$.

\subsection{Pinch force experiment}
\label{sec:pinch}

We compare the single-subject estimator \eqref{eq:est1f} and the two-threshold estimator \eqref{eq:estf} for the subject-specific functions to the pinch force data originally collected by \cite{ramsay1995functional} but preprocessed in a form provided by the \texttt{fda} R package \citep{ramsay2023package}.
As explained in \cite{ramsay2002applied}, the preprocessed data recorded is $m=20$ replications of ``the force exerted on a meter during a brief pinch by the thumb and forefinger. The subject was required to maintain a certain background force on a force meter and then to squeeze the meter aiming at a specified maximum value, returning afterwards to the background level.''
The preprocessing aligns the replications so that their maximum value occurs at the same time.
Each replication is measured at 151 equally spaced time points; we scale the time points to fit in the unit interval. 

We fit both the single-subject estimator and the two-threshold estimator to the data. 
More specifically, for any $j \in \{1,\ldots,m\}$, each estimator for $\fj$ is trained on the data $Y_i^{(j')}$ for $j' \in \{1,\ldots,m\}$ and $i \in \mathcal{N}^{train} \coloneqq \{1,\ldots,151\} \setminus \mathcal{N}^{test}$, where $\mathcal{N}^{test} \coloneqq \{3i-1: i=1,\ldots,50\}$.
The top half of Figure~\ref{fig:pinchdataestf} shows the data with the estimators fit to the training data. 
Because the noise variance is small, for the estimators we set all $\tau$ values to $0.01$. 
The largest modes of the replicates are similar to each other and the two estimators (the single-subject estimator and the two-threshold estimator) fit similarly to each other at these modes, but the data has a lot of variance to the left and right of the modes.
The bottom half of Figure~\ref{fig:pinchdataestf} zooms into these areas, where we can better see that the two-threshold estimator tends to be smoother (i.e. less wiggly) than the single-subject estimator in these areas. 
For each $j=1,\ldots,m$ we numerically evaluate each estimator's predictive performance on the remaining $|\mathcal{N}^{test}|$ time points using the root mean squared prediction error
\begin{align} \label{eq:RMSPE}
    \text{RMSPE}^{(j)} = 
    \sqrt{\frac{1}{|\mathcal{N}^{test}|} \sum_{i \in \mathcal{N}^{test}} \Big( \tilde{f}^{(j)}(t_i) - Y_i^{(j)} \Big)^2}
\end{align}
where $t_i = (i-1)/150$
and $\tilde{f}^{(j)}(\cdot)$ is one of the two estimators computed using the training data.
Figure~\ref{fig:pinchrmspediff} shows that the two-threshold estimator produces smaller RMSPE values for 18 of the 20 replicates.
This provides further evidence that the two-threshold estimator should be used even in finite-sample scenarios. 

\begin{figure}
    \includegraphics[width=0.9\textwidth]{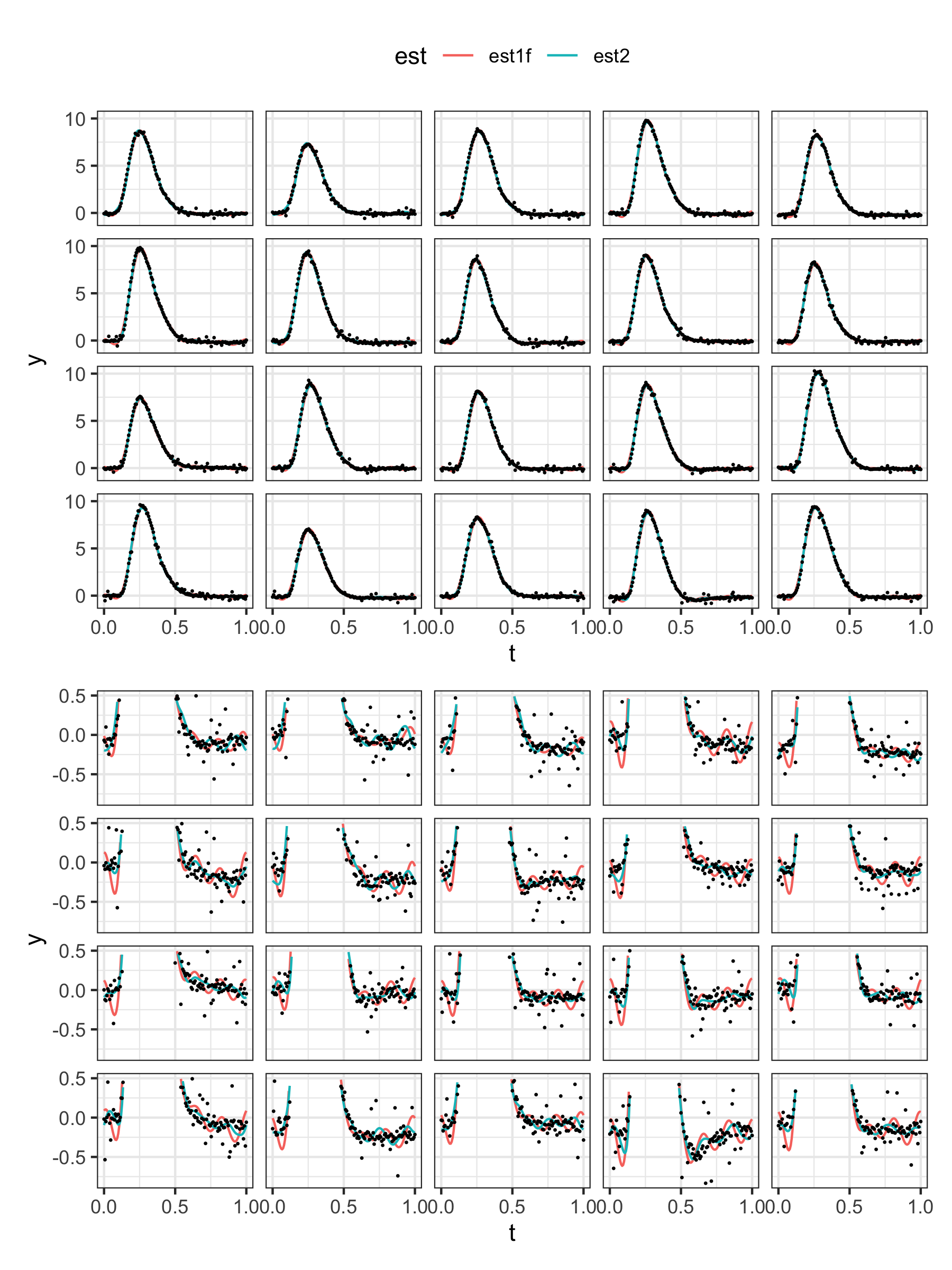}
    \centering
    \caption{Pinch force data (points), single-subject estimator \eqref{eq:est1f} (red line), and two-threshold estimator \eqref{eq:estf} using the adaptive thresholds \eqref{eq:thresholdadaptf} (blue line). Each panel displays one of $m=20$ replicates. Top: all data. Bottom: zooms into ``lower signal'' regions.}
    \label{fig:pinchdataestf}
\end{figure}

\begin{figure}
    \includegraphics[width=0.3\textwidth]{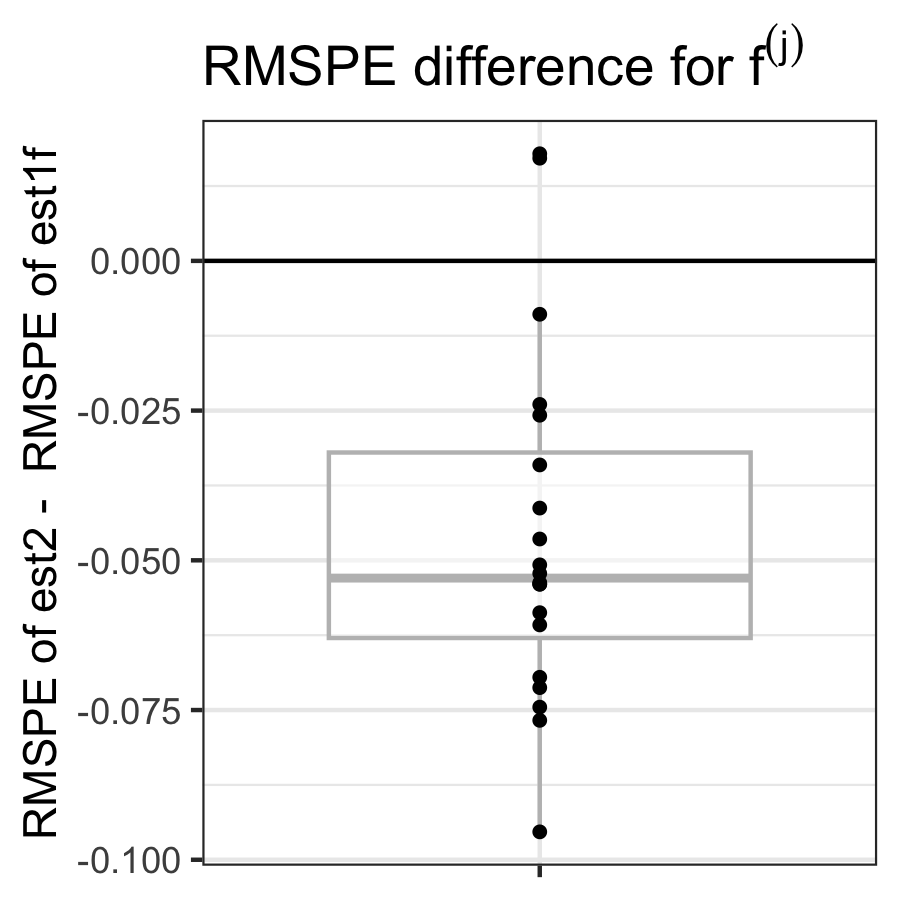}
    \centering
    \caption{Root mean squared prediction error \eqref{eq:RMSPE} difference between the single-subject estimator \eqref{eq:est1f} and two-threshold adaptive estimator of $\fo,\ldots,f^{(20)}$.}
    \label{fig:pinchrmspediff}
\end{figure}

\subsection{Orthosis experiment}
\label{sec:orthosis}

Similar to Section~\ref{sec:pinch}, we compare the single-subject estimator \eqref{eq:est1f} and the two-threshold estimator \eqref{eq:estf} for $\fj$ to the orthosis dataset, which was acquired by the Laboratoire Sport et Performance Motrice, Grenoble University, France. 
The original experiment studied how different externally applied moments to the human knee affect the processes underlying movement generation.
The data recorded were the resultant moments of seven young male volunteers stepping in place for $m=10$ cycles under four experimental conditions: a control condition (without orthosis), an orthosis condition (with the orthosis only), and two conditions (spring 1 and spring 2) where a spring is loaded into the orthosis. 
Each cycle's moments are recorded at 256 equally spaced time points; we scale the time points to fit in the unit interval. 
Figure~\ref{fig:orthosisdata} shows the full dataset.
\cite{cahouet2002static} contains further details on the experiment.

For each subject-condition pair, we fit both the single-subject estimator \eqref{eq:est1f} and the two-threshold estimator \eqref{eq:estf} to the data. 
More specifically, for any $j \in \{1,\ldots,m\}$, each estimator for $\fj$ is trained on the data $Y_i^{(j')}$ for $j' \in \{1,\ldots,m\}$ and $i \in \mathcal{N}^{train} \coloneqq \{1,\ldots,256\} \setminus \mathcal{N}^{test}$, where $\mathcal{N}^{test} \coloneqq \{5i: i=1,\ldots,51\}$.
All estimators use $\tau=6$.
For each $j=1,\ldots,m$ we numerically evaluate each estimator's predictive performance on the remaining $|\mathcal{N}^{test}|$ time points using the root mean squared prediction error \eqref{eq:RMSPE}
where $t_i = 0.002 + 0.0039(i-1)$
and $\tilde{f}^{(j)}(\cdot)$ is one of the two estimators computed using the training data.
Figure~\ref{fig:orthosisRMSPE} shows that for most of the 28 subject-condition pairs, the two-threshold estimator has a smaller RMSPE than does the single-subject estimator.
There are a handful of subject-condition pairs for which the opposite is true; these pairs correspond to cases where the replicates are quite different from each other and hence prediction for a replicate is not helped by pooling information from the other replicates.

\begin{figure}
\includegraphics[width=0.9\textwidth]{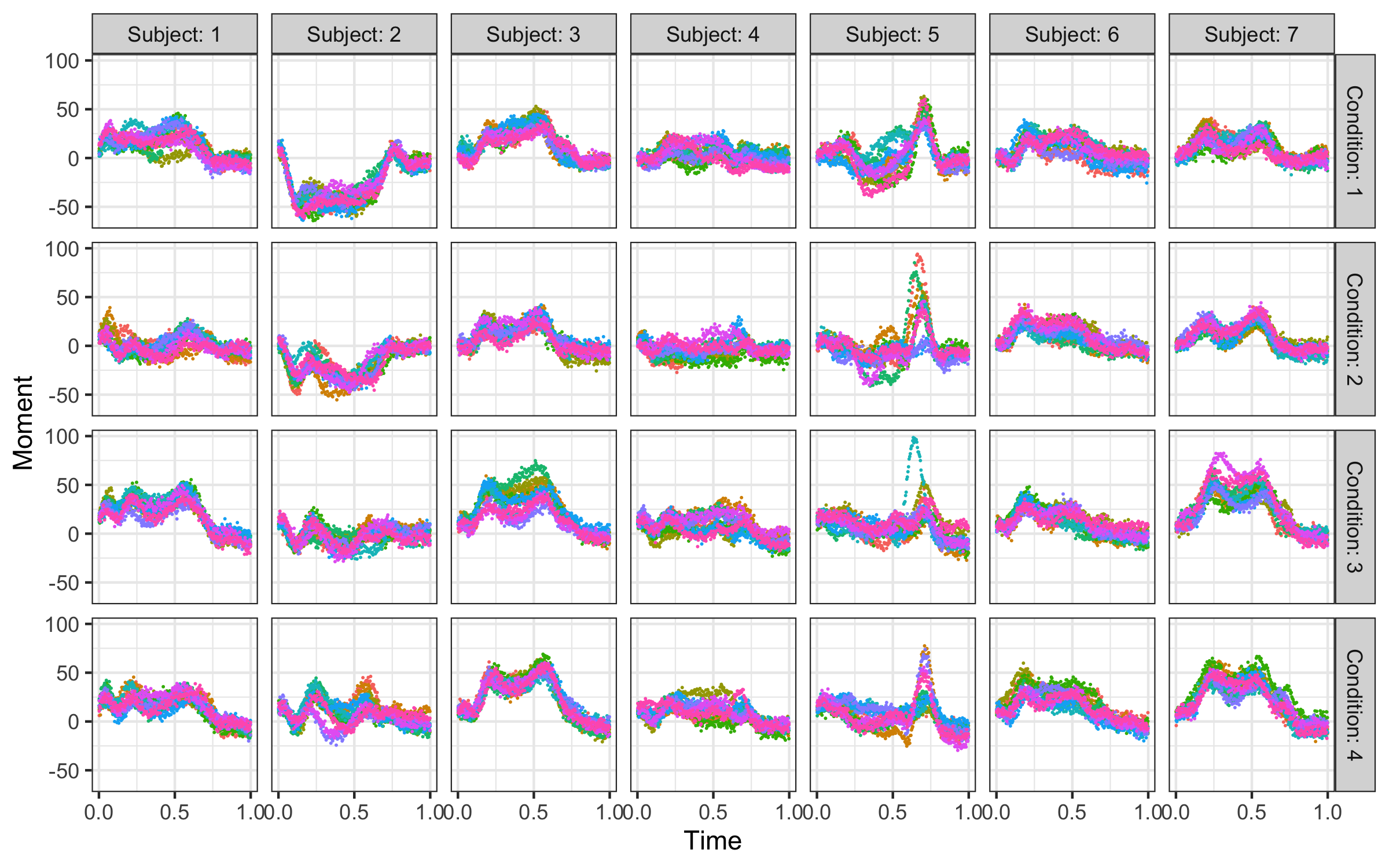}
\centering
\caption{Orthosis data. Each panel contains 10 replicates each consisting of 256 points.}
\label{fig:orthosisdata}
\end{figure}

\begin{figure}
\includegraphics[width=0.9\textwidth]{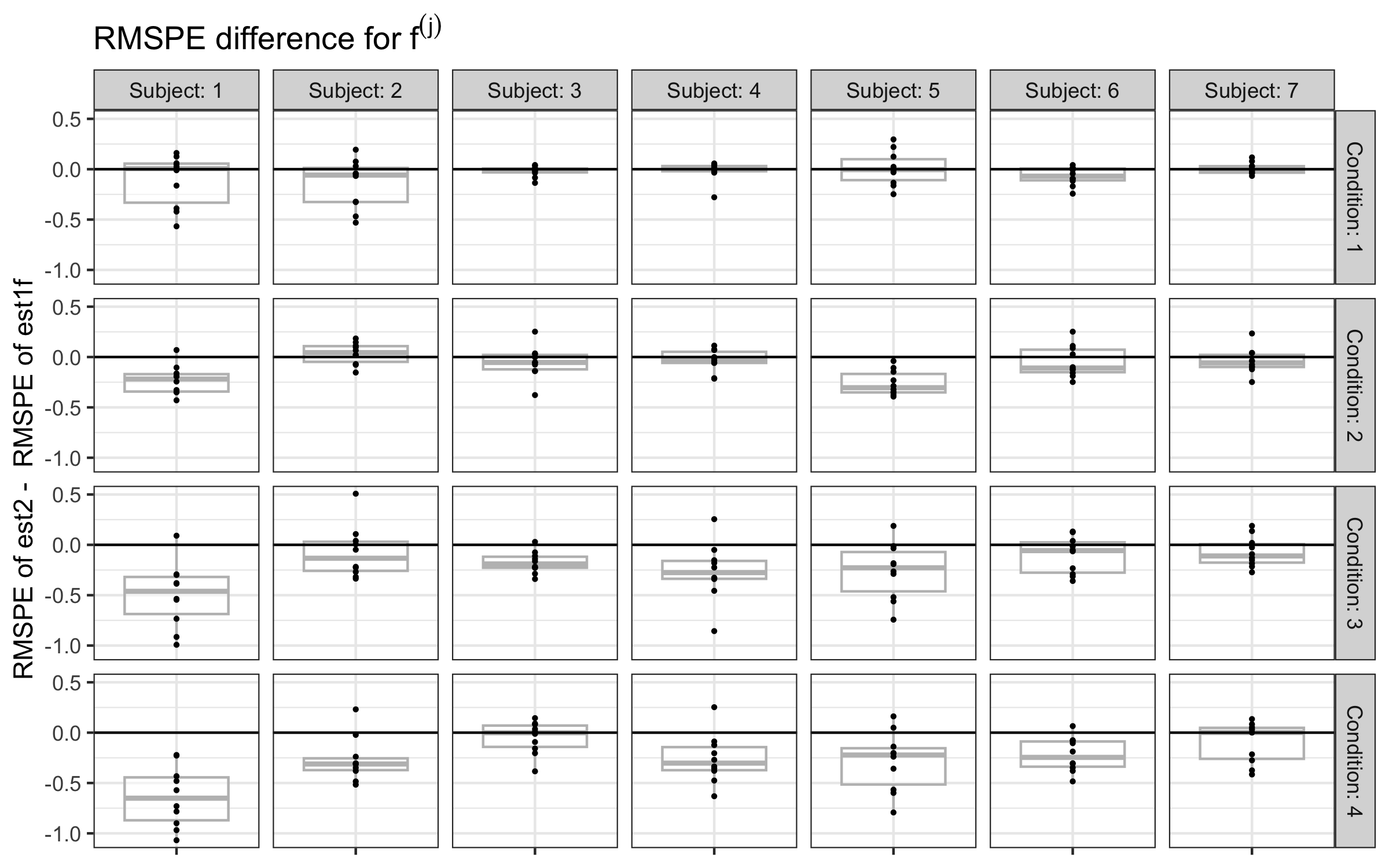}
\centering
\caption{Root mean squared prediction error \eqref{eq:RMSPE} difference between the single-subject estimator \eqref{eq:est1f} and two-threshold adaptive estimator of $\fo,\ldots,f^{(10)}$ for each subject-condition pair.}
\label{fig:orthosisRMSPE}
\end{figure}

\section{Discussion}
\label{sec:discussion}

This paper established the rate of convergence of the $L_2$ minimax risk (when the mean function is assumed to be a fixed member of a Sobolev class of functions) and minimum average risk (when the mean function is equipped with a mean-zero Gaussian process prior) for learning the population and subject-specific functions of interest in a nested Gaussian process model. 
Because our established rates are faster than the single-subject rates, we conclude that the information pooling process can strongly improve estimation of either the subject-specific functions and mean function. 
Our rates also provide practical insights into the trade-off between the sample sizes at the subject level and the observation level under a fixed sampling budget.
We show that these insights hold not only theoretically when the two sample sizes are large but also empirically in some small-sample scenarios.

Our work considered real-valued observations, but many nonparametric applications involve discrete-valued outputs.
Binary observations, for example, would involve $(0,1)$-valued functions to be estimated. 
Though some existing works have studied estimation rates for binary observations in the single-subject case, to our knowledge there is no published work studying such rates under nested sampling designs.  
Due to the increasing use of such designs in nonparametric settings for discrete-valued outputs, it would be of interest to establish minimax rates for binary observations which would in turn provide a foundation to establish rates for other discrete-valued observations. 

\section*{Acknowledgement}
AH is partly supported by NSF grant DMS-2013930 and NIGMS grant R01-GM135440.
LM is partly supported by NSF grants DMS-2013930, DMS-2152999, and NIGMS grant R01-GM135440.
This project has received funding from the European Research Council (ERC) under the European Union's Horizon 2020 research and innovation programme (grant agreement No. 101041064).

\bibliographystyle{chicago} 
\bibliography{gpbib}

\section*{Acknowledgement}
AH is partly supported by NSF grant DMS-2013930 and NIGMS grant R01-GM135440.
LM is partly supported by NSF grants DMS-2013930, DMS-2152999, and NIGMS grant R01-GM135440.
This project has received funding from the European Research Council (ERC) under the European Union's Horizon 2020 research and innovation programme (grant agreement No. 101041064).

\bibliographystyle{chicago} 
\bibliography{gpbib}

\begin{thebibliography}{}

\bibitem[\protect\citeauthoryear{Bak and Koo}{Bak and
  Koo}{2023}]{bak2023minimax}
Bak, K.-Y. and J.-Y. Koo (2023).
\newblock Minimax estimation in multi-task regression under low-rank
  structures.
\newblock {\em Journal of Nonparametric Statistics\/}~{\em 35\/}(1), 122--144.

\bibitem[\protect\citeauthoryear{Bunea, Ombao, and Auguste}{Bunea
  et~al.}{2006}]{bunea2006minimax}
Bunea, F., H.~Ombao, and A.~Auguste (2006).
\newblock Minimax adaptive spectral estimation from an ensemble of signals.
\newblock {\em IEEE transactions on signal processing\/}~{\em 54\/}(8),
  2865--2873.

\bibitem[\protect\citeauthoryear{Cahou{\"e}t, Luc, and David}{Cahou{\"e}t
  et~al.}{2002}]{cahouet2002static}
Cahou{\"e}t, V., M.~Luc, and A.~David (2002).
\newblock Static optimal estimation of joint accelerations for inverse dynamics
  problem solution.
\newblock {\em Journal of biomechanics\/}~{\em 35\/}(11), 1507--1513.

\bibitem[\protect\citeauthoryear{Chau and von Sachs}{Chau and von
  Sachs}{2016}]{chau2016functional}
Chau, J. and R.~von Sachs (2016).
\newblock Functional mixed effects wavelet estimation for spectra of replicated
  time series.
\newblock {\em Electronic Journal of Statistics\/}~{\em 10\/}(2), 2461--2510.

\bibitem[\protect\citeauthoryear{Giacofci, Lambert-Lacroix, and
  Picard}{Giacofci et~al.}{2018}]{giacofci2018minimax}
Giacofci, M., S.~Lambert-Lacroix, and F.~Picard (2018).
\newblock Minimax wavelet estimation for multisample heteroscedastic
  nonparametric regression.
\newblock {\em Journal of Nonparametric Statistics\/}~{\em 30\/}(1), 238--261.

\bibitem[\protect\citeauthoryear{Gin{\'e} and Nickl}{Gin{\'e} and
  Nickl}{2016}]{gine}
Gin{\'e}, E. and R.~Nickl (2016).
\newblock {\em Mathematical foundations of infinite-dimensional statistical
  models}.
\newblock Cambridge university press.

\bibitem[\protect\citeauthoryear{Koudstaal and Yao}{Koudstaal and
  Yao}{2018}]{koudstaal2018multiple}
Koudstaal, M. and F.~Yao (2018).
\newblock From multiple gaussian sequences to functional data and beyond: a
  stein estimation approach.
\newblock {\em Journal of the Royal Statistical Society Series B: Statistical
  Methodology\/}~{\em 80\/}(2), 319--342.

\bibitem[\protect\citeauthoryear{Lepskii}{Lepskii}{1992}]{lepskii1992asymptotically}
Lepskii, O. (1992).
\newblock Asymptotically minimax adaptive estimation. i: Upper bounds.
  optimally adaptive estimates.
\newblock {\em Theory of Probability \& Its Applications\/}~{\em 36\/}(4),
  682--697.

\bibitem[\protect\citeauthoryear{Lepskii}{Lepskii}{1993}]{lepskii1993asymptotically}
Lepskii, O. (1993).
\newblock Asymptotically minimax adaptive estimation. ii. schemes without
  optimal adaptation: Adaptive estimators.
\newblock {\em Theory of Probability \& Its Applications\/}~{\em 37\/}(3),
  433--448.

\bibitem[\protect\citeauthoryear{Lindley and Smith}{Lindley and
  Smith}{1972}]{lindley1972bayes}
Lindley, D.~V. and A.~F. Smith (1972).
\newblock Bayes estimates for the linear model.
\newblock {\em Journal of the Royal Statistical Society Series B: Statistical
  Methodology\/}~{\em 34\/}(1), 1--18.

\bibitem[\protect\citeauthoryear{Novick, Jackson, Thayer, and Cole}{Novick
  et~al.}{1972}]{novick1972estimating}
Novick, M.~R., P.~H. Jackson, D.~T. Thayer, and N.~S. Cole (1972).
\newblock Estimating multiple regressions in m groups: A cross-validation
  study.
\newblock {\em British Journal of Mathematical and Statistical
  Psychology\/}~{\em 25\/}(1), 33--50.

\bibitem[\protect\citeauthoryear{Ramsay}{Ramsay}{2023}]{ramsay2023package}
Ramsay, J. (2023).
\newblock {\em fda: Functional Data Analysis}.
\newblock R package version 6.1.4.

\bibitem[\protect\citeauthoryear{Ramsay, Wang, and Flanagan}{Ramsay
  et~al.}{1995}]{ramsay1995functional}
Ramsay, J., X.~Wang, and R.~Flanagan (1995).
\newblock A functional data analysis of the pinch force of human fingers.
\newblock {\em Journal of the Royal Statistical Society: Series C (Applied
  Statistics)\/}~{\em 44\/}(1), 17--30.

\bibitem[\protect\citeauthoryear{Ramsay and Silverman}{Ramsay and
  Silverman}{2002}]{ramsay2002applied}
Ramsay, J.~O. and B.~W. Silverman (2002).
\newblock {\em Applied functional data analysis: methods and case studies}.
\newblock Springer.

\bibitem[\protect\citeauthoryear{Wang, Oliva, Schneider, and P{\'o}czos}{Wang
  et~al.}{2016}]{wang2016nonparametric}
Wang, X., J.~B. Oliva, J.~G. Schneider, and B.~P{\'o}czos (2016).
\newblock Nonparametric risk and stability analysis for multi-task learning
  problems.
\newblock In {\em IJCAI}, pp.\  2146--2152.

\bibitem[\protect\citeauthoryear{Williams and Rasmussen}{Williams and
  Rasmussen}{2006}]{williams2006gaussian}
Williams, C.~K. and C.~E. Rasmussen (2006).
\newblock {\em Gaussian processes for machine learning}.
\newblock MIT press Cambridge, MA.

\end{thebibliography}

\bigskip
\begin{center}
{\large\bf SUPPLEMENTARY MATERIAL}
\end{center}

\renewcommand*{\thetheorem}{S.\arabic{theorem}}
\renewcommand\theequation{S.\arabic{equation}}
\renewcommand\thesection{S.\arabic{section}}

\appendix

The supplemental material contains the proofs of Theorems 1-8. 
Here we order the eight proofs in the main text by first pairing the related results for $f^{(1)}$ and $g$. 
Because a proof for $f^{(1)}$ typically contains much of the technical content contained in the corresponding proof for $g$, we present the former before the latter. 
Once the related results for $f^{(1)}$ and $g$ are paired, we present the resulting four pairs in the order they appear in the main text. 

\section{Proof of Theorem~5: Bayes risk $\fo$}
\label{sec:proofbayesriskf}

\begin{proof}[Proof of Theorem~5]

This proof considers $m+1$ samples $f^{(1)},...,f^{(m+1)}$ for notational convenience, but this does not influence the derived results qualitatively.
With respect to the distribution $\pi_\alpha=\pi_{\Lambda}$ on $g$, denote the marginal distribution of $f^{(1)}$ by $\pi^{(1)}$ and the expectation of the distribution $Y^{(1)},\ldots,Y^{(m+1)}|f^{(1)}$ by $E_{f^{(1)}}$.

Our goal is to find the estimator $\hf$ minimizing
\begin{align*}
    \E^{\pi_{\alpha}} \E_g \|\hf-f^{(1)}\|_2^2 
= \E^{\pi^{(1)}} E_{\fo} \|\hf- \fo\|_2^2 
= \int E_{f^{(1)}}\|\hf-f^{(1)}\|_2^2\dif \pi^{(1)}(f^{(1)}).
\end{align*}
The above Bayes risk is minimized by the posterior mean for $f^{(1)}|Y^{(1)},\ldots,Y^{(m+1)}$. 
It can also be decomposed into a squared bias and variance term
\begin{align*}
\int E_{f^{(1)}}\|\hf-f^{(1)}\|_2^2\dif \pi^{(1)}(f^{(1)}) = \int B_{\fo}(\hf) \dif \pi^{(1)}(f^{(1)}) + \int V_{\fo}(\hf) \dif \pi^{(1)}(f^{(1)})
\end{align*}
where we use the notations $B_{\fo}(\hf) = \| E_{f^{(1)}}\hf-f^{(1)}\|_2^2$ and $V_{\fo}(\hf) = E_{f^{(1)}}\| \hf- E_{f^{(1)}}\hf\|_2^2$.

Next we consider the sequence representation of $Y^{(j)}$ with respect to matching eigenbasis $\psi_1,\psi_2,\ldots,$ of $\Lambda$ and $\tLambda$, 
\begin{align*}
Y^{(j)}_k=\sum_{k=1}^{\infty} f^{(j)}_k+\sqrt{\frac{1}{n}}Z_k^{(j)},
\end{align*}
where $Z_k^{(j)} \iid \text{N}(0,1)$, $f^{(j)}_k=\langle f^{(j)},\psi_k \rangle$, and $Y^{(j)}_k= \int \psi_k(t) dY^{(j)}_t$. 
We also introduce the notation $g_k=\langle g,\psi_k \rangle$. 
From now on we slightly abuse our notations and denote by $f^{(j)}$, $g$, and $Y^{(j)}$ both the functional and the $\ell_2$ sequence representation of the (random) functions.

The conditional independence of $f^{(j)}|g$, $j=2,\ldots,m+1$ implies
\begin{align*}
Y^{(2)},\ldots,Y^{(m+1)}|g \sim
\bigotimes_{j=2}^{m+1} \bigotimes_{k=1}^{\infty} \text{N}\Big(g_k,\tlambda_k+1/n\Big)
\quad \text{and} \quad \bar{Y}|g \sim \bigotimes_{k=1}^{\infty} \text{N}\bigg(g_k, \frac{\tlambda_k+1/n}{m}\bigg)
\end{align*}
where $\bar{Y} \coloneqq m^{-1}\sum_{j=2}^{m+1}Y^{(j)}$.
Finally, the marginal distribution of $f^{(1)}$ is
\begin{align}
f^{(1)}\sim \bigotimes_{k=1}^{\infty} \text{N}\Big(0,\lambda_k+\tlambda_k\Big).\label{eq:prior:f}
\end{align}

\subsection{Posterior for $f^{(1)}$}
\label{sec:posteriorf}
In this section we derive the posterior distribution for $f^{(1)}$ given $(Y^{(j)})_{j=1,\ldots,m+1}$. For notational convenience we omit the superindex $(1)$ from $f^{(1)}$, i.e., use the notation $f=\fo$.

As a first step we derive the joint conditional distribution of $f$ and $g$ given $Y^{(1)},\ldots,Y^{(m+1)}$ using Bayes' formula
\begin{align}
&\pi(f,g \mid Y^{(1)},\ldots,Y^{(m+1)})\propto \pi(g)\pi(f|g)p(Y^{(1)}|f)\prod_{j=2}^{m+1} p(Y^{(j)}|g)\nonumber\\
&\propto \exp\Bigg\{-\frac{1}{2}\sum_{k=1}^{\infty} \frac{g_k^2}{\lambda_k}+\frac{(f_k-g_k)^2}{\tlambda_k}+n(Y^{(1)}_k-f_k)^2 + \frac{1}{\tlambda_k+\frac{1}{n}}\sum_{j=2}^{m+1} (Y^{(j)}_k-g_k)^2 \Bigg\}.\label{eq:joint}
\end{align}

We will derive the posterior for $f$ by marginalizing out $g$ from the above formula. 
Letting $c_k = \lambda_k^{-1}+\tlambda_k^{-1}+m(\tlambda_k+1/n)^{-1}$ and $d_k = \tlambda_k^{-1}f_k+(\tlambda_k+1/n)^{-1}m \bar{Y}_k$, 
we note
\begin{align*}
g_k^2 c_k - 2g_k d_k 
= c_k \big( g_k- d_k c_k^{-1} \big)^2 - d_k^2 c_k^{-1}.
\end{align*}
Letting $a_k \coloneqq \tlambda_k^{-1} m(\tlambda_k+1/n)^{-1} c_k^{-1}$ and $b_k \coloneqq \big\{\lambda_k^{-1}+m(\tlambda_k+1/n)^{-1}\big\} c_k^{-1}$, by \eqref{eq:joint} we have 
\begin{align*}
&\pi(f,g \mid Y^{(1)},\ldots,Y^{(m+1)}) \\
&\propto \exp\bigg\{-\frac{1}{2}\sum_{k=1}^{\infty}c_k \big( g_k- d_k c_k^{-1} \big)^2 - d_k^2c_k^{-1} + f_k^2\big(\tlambda_k^{-1}+n\big) - 2f_k Y^{(1)}_k n \bigg\} \\
&\propto \exp\bigg\{-\frac{1}{2}\sum_{k=1}^{\infty}c_k \big( g_k- d_k c_k^{-1} \big)^2 \bigg\} \exp\Bigg\{ -\frac{1}{2}\sum_{k=1}^{\infty}\big(n+\tlambda_k^{-1} b_k \big) \bigg( f_k -\frac{Y^{(1)}_k n+\bar{Y}_k a_k}{n+\tlambda_k^{-1}b_k}\bigg)^2 \Bigg\}.
\end{align*}
Since $\pi(f,g|Y^{(1)},\ldots,Y^{(m+1)})=\pi(g|f,Y^{(1)},\ldots,Y^{(m+1)})\pi(f|Y^{(1)},\ldots,Y^{(m+1)})$, we conclude $g|f,Y^{(1)},\ldots,Y^{(m+1)}$ has a Gaussian distribution 
and $f|Y^{(1)},\ldots,Y^{(m+1)}$ has a Gaussian distribution given by 
\begin{align*}
f|Y^{(1)},\ldots,Y^{(m+1)}\sim \bigotimes_{k=1}^{\infty} \text{N}\Bigg(\frac{Y^{(1)}_k n+\bar{Y}_ka_k}{n+\tlambda_k^{-1}b_k},
\frac{1}{n+\tlambda_k^{-1}b_k} \Bigg).
\end{align*}

\subsection{Bias term}
\label{sec:biasterm}
This section provides the asymptotic rate for the integrated squared bias $\E^{\pi^{(1)}} B_{f}(\hf)$ where $\hf$ is the posterior mean derived above.
Noting $E_f Y^{(1)}_k = \cdots = E_f Y^{(m+1)}_k = f_k$, we have 
\begin{align*}
E_f\hf_k &= f_k
\bigg(\frac{n + a_k}{n+\tlambda_k^{-1}b_k} \bigg) 
= f_k \Bigg(1 - \frac{1}{n\tlambda_k+1+n(m+1)\lambda_k} \Bigg),
\end{align*}
where $a_k$ and $b_k$ are defined in Section~\ref{sec:posteriorf}, and thus
\begin{align}
B_{f}(\hf) 
=\sum_{k=1}^{\infty}\frac{f_k^2}{(n\tlambda_k+1+n(m+1)\lambda_k)^2 } 
\asymp \sum_{k=1}^{\infty}\frac{f_k^2}{(n\tlambda_k+1+nm\lambda_k)^2 }. \label{eq:denominator}
\end{align}

To simplify computation, we partition $\mathbb{N}$ into three intervals depending on which term in the denominator dominates (i.e., which denominator term is the largest). 
In view of $\lambda_k \asymp k^{-1-2\alpha}$ and $\tlambda_k \asymp k^{-1-2\talpha}$, the term $nm\lambda_k$ dominates on the interval $[1, t_1]$, where $t_1=t_1(n,m,\alpha,\tilde\alpha)$, the term $n\tlambda_k$ dominates on $(t_1, t_2]$, where $t_2=t_2(n,m,\alpha,\tilde\alpha)$ might be equal to $t_1$ corresponding to an empty interval, and for large integers the constant term dominates. We deal with these three intervals separately.
For $k \in [1,t_1]$, i.e., for $k \in \mathbb{N}$ where $nm\lambda_k$ dominates the denominator, we have 
\begin{align*}
\sum_{k=1}^{t_1}\frac{f_k^2}{(n\tlambda_k+1+nm\lambda_k)^2}
\asymp (nm)^{-2}\sum_{k=1}^{t_1} f_k^2 \lambda_k^{-2}.
\end{align*}
For $k \in ({t_1}, {t_2}]$ i.e. for $k \in \mathbb{N}$ where $n \tlambda_k$ dominates the denominator, we have
\begin{align*}
\sum_{k=t_1+1}^{t_2} \frac{f_k^2}{(n\tlambda_k+1+nm\lambda_k)^2}
\asymp n^{-2} \sum_{k=t_1+1}^{t_2} f_k^2 \tlambda_k^{-2}.
\end{align*}
Finally, for $k \in [t_2,\infty)$ i.e. for $k \in \mathbb{N}$ where $1$ dominates the denominator, we have
\begin{align*}
\sum_{k=t_2}^{\infty} \frac{f_k^2}{(n\tlambda_k+1+nm\lambda_k)^2}
\asymp \sum_{k=t_2}^{\infty} f_k^2.
\end{align*}
Next we substitute the values $\lambda_k\asymp k^{-1-2\alpha}$ and $\tlambda_k\asymp k^{-1-2\talpha}$, resulting in
\begin{align*}
B_{f}(\hf) &\asymp (nm)^{-2}\sum_{k=1}^{t_1} f_k^2 k^{2+4\alpha} + n^{-2} \sum_{k=t_1+1}^{t_2} f_k^2 k^{2+4\tilde\alpha} + \sum_{k=t_2}^{\infty} f_k^2.
\end{align*}
Since \eqref{eq:prior:f} implies $\E^{\pi^{(1)}} f_k^2 = \lambda_k+\tilde{\lambda}_k\asymp k^{-1-2(\alpha \wedge \talpha)}$, we get
\begin{align*}
\E^{\pi^{(1)}} B_{f}(\hf) &\asymp (nm)^{-2}\sum_{k=1}^{t_1} k^{1+4\alpha-2(\alpha\wedge\talpha)} + n^{-2} \sum_{k=t_1+1}^{t_2} k^{1+4\tilde\alpha-2(\alpha\wedge\talpha)} + \sum_{k=t_2}^{\infty} k^{-1-2(\alpha\wedge\talpha)}.
\end{align*}

Now we find the values of the interval boundaries $t_1$ and $t_2$, which depend on $n,m,\alpha,\talpha$.
In view of $\lambda_k\asymp k^{-1-2\alpha}$ and $\tlambda_k\asymp k^{-1-2\talpha}$, there is always $K$ large enough that $1 > (nmk^{-1-2\alpha} \vee nk^{-1-2\talpha})$ for all $k \geq K$, which implies $t_1 \leq t_2$ at any $n,m,\alpha,\talpha$.
Note $t_1 = t_2$ if and only if $nk^{-1-2\talpha} \leq (1 \vee nmk^{-1-2\alpha})$ for all $k$.
This inequality holds for all $k$ if $\talpha \geq \alpha$, i.e. if $\tlambda_k$ decays faster than $\lambda_k$ does.
If $\talpha < \alpha$, the inequality holds if and only if $k$ is outside the interval $(m^{1/(2\alpha-2\talpha)}, n^{1/(1+2\talpha)})$, which itself is empty if and only if $\alpha \leq \talpha+(1/2+\talpha)\delta$ where $\delta \coloneqq (\log m / \log n)$.
In summary, $t_1=t_2$ if and only if $\alpha \leq \talpha+(1/2+\talpha)\delta$.

Then we distinguish two cases according to the value of $\alpha$ and $\talpha$. 
For $\alpha \leq \talpha+(1/2+\talpha)\delta$ the boundaries are $t_1=t_2=(nm)^{1/(1+2\alpha)}$, in which case $\E^{\pi^{(1)}} B_{f}(\hf) \asymp (nm)^{-2(\alpha \wedge \talpha)/(1+2\alpha)}$. 
For the remaining case $\alpha > \talpha+(1/2+\talpha)\delta$ the boundaries are $t_1=m^{1/(2\alpha-2\talpha)}$ and $t_2= n^{1/(1+2\talpha)}$, in which case $\E^{\pi^{(1)}} B_{f}(\hf) \asymp n^{-2\talpha/(1+2\talpha)}$. 
We summarize the two cases as
\begin{equation}\label{UB: bias}
\E^{\pi^{(1)}} B_{f}(\hf)\asymp
\begin{cases}
(nm)^{-\frac{2(\alpha\wedge\talpha)}{1+2\alpha}} & \text{for } \talpha+(1/2+\talpha)\delta > \alpha,\\
n^{-\frac{2\talpha}{1+2\talpha}} & \text{for } \talpha+(1/2+\talpha)\delta\leq \alpha. 
\end{cases} 
\end{equation}

\subsection{Variance term}
\label{sec:varianceterm}
This section provides the asymptotic rate for the integrated variance $\E^{\pi^{(1)}} \V_{f}(\hf)$. 
First recall that $V(X+Y) = V(X) + V(Y) + 2Cov(X,Y)$ and $V(X+Y) \leq 2V(X) + 2V(Y)$ for arbitrary random variables $X$ and $Y$. 
Furthermore, in view of $V_f(\bar{Y}_k)=(\tlambda_k+1/n)m^{-1}+\tlambda_k$, $V_f(Y^{(1)}_k)=1/n$, and $Cov_f(\bar{Y}_k, Y^{(1)}_k) \geq 0$ we have 
\begin{align*}
V_f(\hf)
\asymp
\sum_{k=1}^{\infty} \frac{n+a_k^2\{ (\tlambda_k+1/n)m^{-1}+\tlambda_k\}}{(n+\tlambda_k^{-1}b_k)^2}.
\end{align*}
We multiply both the numerator and denominator by  $\lambda_k^2\tlambda_k^2(\lambda_k^{-1}+\tlambda_k^{-1}+m(\tlambda_k+1/n)^{-1})^2$. Then the denominator is handled the same way as for the bias term, while the numerator (after straightforward but somewhat cumbersome computations) is bounded from below and above by a multiple of 
\begin{align*}
    A_{k,n,m}:=n\tilde\lambda_k^2+ m^2 \lambda_k^2(\tilde\lambda_k+1/n)^{-1}(n\tilde\lambda_k+1/m),
\end{align*}
which, by combining the above computations, leads to
\begin{align*}
V_{f}(\hf) \asymp \sum_{k=1}^{\infty}\frac{A_{k,n,m}}{(nm\lambda_k+n\tlambda_k+1)^2}.
\end{align*}

The denominator in the preceding panel is the same as the one in \eqref{eq:denominator}.
Thus we follow the same strategy of partitioning $\mathbb{N}$ into three intervals. 
For $k \in [1,t_1]$, i.e., for $k \in \mathbb{N}$ where $nm\lambda_k$ dominates the denominator, we have 
\begin{align*}
\sum_{k=1}^{t_1}\frac{A_{k,n,m}}{(n\tlambda_k+1+nm\lambda_k)^2}
\asymp \frac{1}{nm^2}\sum_{k=1}^{t_1}\tlambda_k^2\lambda_k^{-2}+ \frac{1}{n}\sum_{k=1}^{t_1}\frac{\tilde\lambda_k+1/(nm)}{\tilde\lambda_k+1/n}.
\end{align*}
For $k \in ({t_1}, {t_2}]$ i.e. for $k \in \mathbb{N}$ where $n \tlambda_k$ dominates the denominator, we have
\begin{align*}
\sum_{k=t_1+1}^{t_2} \frac{A_{k,n,m}}{(n\tlambda_k+1+nm\lambda_k)^2}
\asymp \frac{t_2-t_1}{n}+ \frac{m^2}{n}\sum_{k=t_1+1}^{t_2}\lambda_k^2\tlambda_k^{-2}.
\end{align*}
Finally, for $k \in [t_2,\infty)$ i.e. for $k \in \mathbb{N}$ where $1$ dominates the denominator, we have
\begin{align*}
\sum_{k=t_2}^{\infty} \frac{A_{k,n,m}}{(n\tlambda_k+1+nm\lambda_k)^2}
\asymp n\sum_{k=t_2}^{\infty}\tilde\lambda_k^2+ nm\sum_{k=t_2}^{\infty} \lambda_k^2+ n^2m^2\sum_{k=t_2}^{\infty} \lambda_k^2\tlambda_k .
\end{align*}
Next we substitute the values $\lambda_k\asymp k^{-1-2\alpha}$ and $\tlambda_k\asymp k^{-1-2\talpha}$, resulting in
\begin{align*}
V_{f}(\hf) 
&\asymp \frac{1}{nm^2}(t_1^{1+4(\alpha-\talpha)}\vee 1)+\frac{1}{n}\sum_{k=1}^{t_1}\frac{k^{-1-2\talpha}+1/(nm)}{k^{-1-2\talpha}+1/n}+\frac{t_2-t_1}{n}\\
&\quad+\frac{m^2}{n}\sum_{k=t_1}^{t_2}  k^{4(\talpha-\alpha)}+n t_2^{-1-4\tilde\alpha}+nmt_2^{-1-4\alpha}+n^2m^2 t_2^{-2-4\alpha-2\talpha}.
\end{align*}

Then similarly to the bias term in the previous section we distinguish two cases. 
For $\alpha \leq \talpha+(1/2+\talpha)\delta$ the boundaries are $t_1=t_2=(nm)^{1/(1+2\alpha)} \geq n^{1/(1+2\talpha)}$. By noting 
\begin{align*}
    \frac{1}{n}\sum_{k=1}^{t_1}\frac{k^{-1-2\talpha}+1/(nm)}{k^{-1-2\talpha}+1/n}
    \asymp \sum_{k=1}^{n^{1/(1+2\talpha)}} \frac{1}{n} +\sum_{k=n^{1/(1+2\talpha)}}^{t_1}\Big(k^{-1-2\talpha}+\frac{1}{nm}\Big)\asymp n^{-\frac{2\talpha}{1+2\talpha}}
\end{align*}
and since $(nm)^{-2\talpha/(1+2\alpha)} \leq n^{-2\talpha/(1+2\talpha)}$ in this case, we get 
\begin{align*}
V_{f}(\hf) 
&\asymp \frac{(nm)^{\frac{1+4(\alpha-\talpha)}{1+2\alpha}}}{nm^2}+n^{-\frac{2\talpha}{1+2\talpha}}+(nm)^{-\frac{2\alpha}{1+2\alpha}} +n(nm)^{-\frac{1+4\talpha}{1+2\alpha}}\\
&\qquad+nm(nm)^{-\frac{1+4\alpha}{1+2\alpha}}+ n^2m^2 (nm)^{-\frac{2+4\alpha+2\talpha}{1+2\alpha}} \\
&\asymp n^{-\frac{2\talpha}{1+2\talpha}} + (nm)^{-\frac{2(\alpha\wedge\talpha)}{1+2\alpha}}
\asymp n^{-\frac{2\talpha}{1+2\talpha}} + (nm)^{-\frac{2\alpha}{1+2\alpha}}.
\end{align*}
Note $(nm)^{-2\alpha/(1+2\alpha)} \geq n^{-2\talpha/(1+2\talpha)}$ if and only if $\alpha \leq \talpha (1+\delta+2\talpha\delta)^{-1}$, which is a subcase of $\alpha \leq \talpha+(1/2+\talpha)\delta$.
Thus the rate in the preceding panel can be written as 
\begin{align*}
    \Big[n^{-\frac{2\talpha}{1+2\talpha}} + (nm)^{-\frac{2\alpha}{1+2\alpha}}\Big] 1_{\alpha \leq \talpha+(1/2+\talpha)\delta} \asymp
    \begin{cases}
    (nm)^{-\frac{2\alpha}{1+2\alpha}} & \text{for $\alpha\leq \frac{\talpha}{1+\delta+2\talpha\delta}$,}\\
    n^{-\frac{2\talpha}{1+2\talpha}} & \text{for $\frac{\talpha}{1+\delta+2\talpha\delta} \leq \alpha \leq \talpha+(1/2+\talpha)\delta$.}
    \end{cases} 
\end{align*}
For the remaining case $\alpha > \talpha+(1/2+\talpha)\delta$ the boundaries are  $t_1=m^{1/(2\alpha-2\talpha)}$ and $t_2= n^{1/(1+2\talpha)}$. 
Here $(nm)^{1/(1+2\alpha)} \leq n^{1/(1+2\talpha)}$, which implies $m^{1/(2\alpha-2\talpha)}\leq n^{1/(1+2\talpha)}$. 
Then
\begin{align*}
V_{f}(\hf) &\asymp  \frac{ m^{\frac{1+4(\alpha-\talpha)}{2(\alpha-\talpha)}} }{nm^2}+ n^{-1}m^{\frac{1}{2(\alpha-\talpha)}}+\frac{n^{1/(1+2\talpha)}}{n}+ \frac{m^2}{n} m^{\frac{1+4(\talpha-\alpha)}{2(\alpha-\talpha)}}\\
&\qquad+n n^{-\frac{1+4\talpha}{1+2\talpha}}+nm n^{-\frac{1+4\alpha}{1+2\talpha}}+ n^2m^2 n^{-\frac{2+4\alpha+2\talpha}{1+2\talpha}}\nonumber\\
&\asymp n^{-\frac{2\talpha}{1+2\talpha}}.
\end{align*}
We conclude the two cases in the display below
\begin{equation}\label{UB: variance}
V_{f}(\hf) = \E^{\pi^{(1)}} V_{f}(\hf) \asymp
\begin{cases}
(nm)^{-\frac{2\alpha}{1+2\alpha}} & \text{for } \alpha\leq \talpha (1+\delta+2\talpha\delta)^{-1}, \\
n^{-\frac{2\talpha}{1+2\talpha}}  & \text{for } \alpha\geq \talpha (1+\delta+2\talpha\delta)^{-1}. 
\end{cases} 
\end{equation}

\subsection{Conclusion}
From \eqref{UB: bias} and \eqref{UB: variance} we conclude $\E^{\pi^{(1)}} B_f(\hf) \lesssim V_{f}(\hf)$ and thus
\begin{equation}  \label{eq:gwnbayesrisklowerboundf1}
\inf_{\hf}\int E_{f^{(1)}}\|\hf-f^{(1)}\|_2^2\dif \pi^{(1)}(f^{(1)})
\asymp
V_{f}(\hf)
\asymp (nm)^{-\frac{2\alpha}{1+2\alpha}} + n^{-\frac{2\talpha}{1+2\talpha}}.
\end{equation}
\end{proof}

\section{Proof of Theorem~1: Bayes risk for $g$}
\label{sec:proofbayesriskg}

\begin{proof}[Proof of Theorem~1]
First we consider the Bayes risk associated with the GP prior $\pi=\pi_{\alpha}$ on $g$ with regularity hyperparameter $\alpha$, i.e.
\begin{align*}
\E^{\pi} \E_g\|\hg-g\|_2^2
= \int \E_g\|\hg-g\|_2^2 \dif\pi(g).
\end{align*}
The above Bayes risk is minimized by the posterior mean for $g$ given the data $Y^{(1)},\ldots,Y^{(m)}$. The likelihood of the sufficient statistic $\tY_k \coloneqq m^{-1} \sum_{j=1}^m Y_k^{(j)}$ is $\text{N}(g_k, m^{-1} (\tlambda_k + n^{-1}))$.
Under the prior $g_k \sim \text{N}(0, \lambda_k)$, the posterior distribution of $g_k$ given $\tY_k$ is Gaussian with mean 
\begin{equation}  \label{eq:gwnpostmeang}
    (\lambda_k^{-1} m^{-1} (\tlambda_k + n^{-1}) + 1)^{-1} \tY_k \eqqcolon \hat{g}_k.
\end{equation}
Under the standard bias-variance decomposition, the above integrated MISE becomes
\begin{align*}
    \E^{\pi} \E_g\|\hg-g\|_2^2 = \E^{\pi} \| \E_g\hg-g\|_2^2 + \E^{\pi} \E_g\| \hg-\E_g\hg\|_2^2
\end{align*}
and for simplicity we use the notations $B_g(\hg) \coloneqq \| \E_g\hg-g\|_2^2$ and $V_g(\hg) \coloneqq \E_g\| \hg- \E_g\hg\|_2^2$. We deal with these two terms separately.

For the squared bias term, we get
\begin{align} \label{eq:gwnbayesriskgbias}
    B_g(\hg)
    = \sum_{k=1}^{\infty} \Big(\frac{n\tlambda_k + 1}{nm\lambda_k + n\tlambda_k + 1}\Big)^2 g_k^2,
\end{align}
where we recognize the denominator is the same as the one in \eqref{eq:denominator} and thus we follow the same strategy of partitioning $\mathbb{N}$ into three intervals. 
For $t_1$ and $t_2$ as before, we have 
\begin{align*}
    B_g(\hg)
    \asymp m^{-2} \sum_{k=1}^{t_1} \frac{\tlambda_k^2 + n^{-2}}{\lambda_k^2} g_k^2
    + \sum_{k=t_1}^{t_2} \Big(1+n^{-2} \tlambda_k^{-2} \Big) g_k^2
    + \sum_{k=t_2}^{\infty} \Big(n^2 \tlambda_k^2 + 1\Big) g_k^2.
\end{align*}
Then using $\E^{\pi} g_k^2 = \lambda_k$ and substituting $\lambda_k \asymp k^{-1-2\alpha}$ and $\tlambda_k \asymp k^{-1-2\talpha}$, we get 
\begin{align*}
    \E^{\pi} B_g(\hg)
    \asymp m^{-2} \sum_{k=1}^{t_1} \big(k^{2\alpha-1-4\tilde\alpha} + n^{-2} k^{1+2\alpha} \big) 
    + \sum_{k=t_1}^{t_2} \big(k^{-1-2\alpha}+n^{-2} k^{1+4\tilde\alpha-2\alpha}\big) \\
    + \sum_{k=t_2}^{\infty} \big(n^2 k^{-3-4\tilde\alpha-2\alpha} + k^{-1-2\alpha}\big).
\end{align*}
Then we distinguish two cases. For $\alpha \leq \talpha+(1/2+\talpha)\delta$ the boundaries are $t_1=t_2=(nm)^{1/(1+2\alpha)} \geq n^{1/(1+2\talpha)}$, in which case $\E^{\pi} B_g(\hg) \asymp (nm)^{-2\alpha/(1+2\alpha)}$.
For the remaining case $\alpha > \talpha+(1/2+\talpha)\delta$ the boundaries are $t_1=m^{1/(2\alpha-2\talpha)}$ and $t_2= n^{1/(1+2\talpha)}$, in which case $\E^{\pi} B_g(\hg) \asymp m^{-\alpha/(\alpha-\talpha)}$. 
We summarize these two cases by writing 
\begin{align}  \label{eq:gwnbayesriskgintegratedbias}
    \E^{\pi}B_g(\hg)
    &\asymp 
    \begin{cases}
        m^{-\alpha/(\alpha-\talpha)} &\qquad \text{if } \alpha > \talpha + \delta (1+2\talpha)/2, \\
        (nm)^{-2\alpha/(1+2\alpha)} &\qquad \text{otherwise.}
    \end{cases}
\end{align}
Note $m^{-\alpha/(\alpha-\talpha)} = (nm)^{-2\alpha/(1+2\alpha)}$ at the boundary $\alpha = \talpha + \delta (1+2\talpha)/2$.

For the variance term $V_g(\hg)$, by \eqref{eq:gwnpostmeang} and $\V [\tY_k \mid g_k] = m^{-1} (\tlambda_k + n^{-1})$ we have
\begin{align*}
    V_g(\hg)
    = \sum_{k=1}^{\infty} \frac{m^{-1} (\tlambda_k + n^{-1})}{(\lambda_k^{-1} m^{-1} (\tlambda_k + n^{-1}) + 1)^2}
    = n^2 m \sum_{k=1}^{\infty} \frac{\lambda_k^2(\tlambda_k + n^{-1})}{(n\tlambda_k+1+nm\lambda_k)^2},
\end{align*}
where again we recognize the denominator and thus follow the same partitioning strategy as before. 
For $t_1$ and $t_2$ as before, and substituting $\lambda_k \asymp k^{-1-2\alpha}$ and $\tlambda_k \asymp k^{-1-2\talpha}$, we get
\begin{align*}
    V_g(\hg)
    \asymp m^{-1} \sum_{k=1}^{t_1} \big(k^{-1-2\talpha} + n^{-1}\big) 
    + m \sum_{k=t_1}^{t_2} \big(k^{-1-4\alpha+2\talpha}+n^{-1} k^{4\tilde\alpha-4\alpha}\big) \\
    + n m \sum_{k=t_2}^{\infty} \big(n k^{-3-4\alpha-2\tilde\alpha} + k^{-2-4\alpha}\big).
\end{align*}
Then we distinguish two cases. For $\alpha \leq \talpha+(1/2+\talpha)\delta$ the boundaries are $t_1=t_2=(nm)^{1/(1+2\alpha)} \geq n^{1/(1+2\talpha)}$, in which case $V_g(\hg) \asymp m^{-1} + (nm)^{-2\alpha/(1+2\alpha)}$. 
For the remaining case $\alpha > \talpha+(1/2+\talpha)\delta$ the boundaries are $t_1=m^{1/(2\alpha-2\talpha)}$ and $t_2= n^{1/(1+2\talpha)}$, in which case $V_g(\hg) \asymp m^{-1}$. 
Combining the results of the two scenarios, we get 
\begin{align} \label{eq:gwnbayesriskgvar}
    V_g(\hg) = \E^{\pi} V_g(\hg) \asymp 
    \begin{cases}
        m^{-1} &\qquad \text{if } \alpha > \talpha + \delta (1+2\talpha)/2 \\
        m^{-1} + (nm)^{-2\alpha/(1+2\alpha)} &\qquad \text{otherwise.}
    \end{cases}
\end{align}
Then $V_g(\hg) \gtrsim \E^{\pi} B_g(\hg)$ since $m^{-\alpha/(\alpha-\talpha)} < m^{-1}$ if $\alpha > \talpha + \delta (1+2\talpha)/2$.

From \eqref{eq:gwnbayesriskgintegratedbias} and \eqref{eq:gwnbayesriskgvar}, we conclude 
\begin{align} \label{eq:MISE:Bayesg}
\int  \E_g\|\hg-g\|_2^2\dif\pi(g)\asymp m^{-1}+(nm)^{-\frac{2\alpha}{1+2\alpha}}.
\end{align}
\end{proof}

\section{Proof of Theorem~6: Minimax LB for $\fo$}
\label{sec:proofminimaxriskf}

\begin{proof}[Proof of Theorem~6]
Let $\pi = \pi_{\alpha}$ be the GP prior on $g$ with regularity hyperparameter $\alpha$.
To reduce visual clutter we let $f$ and $\hf$ respectively denote $\fo$ and $\hfo$, $\|\cdot\|$ denote the $L_2$ norm, and let $R_n \coloneqq 2\sqrt{\log n}$. First note
\begin{align*}
   \inf_{\hf} \sup_{g \in \cS(\alpha, R_n)} \E_g\|\hf-f\|^2 
    &\geq \inf_{\hf}\int_{\cS(\alpha, R_n)} \E_g\|\hf-f\|^2 \dif\pi(g) \\
    &=\inf_{\hf}\bigg\{\int \E_g\|\hf-f\|^2 \dif\pi(g)-
    \int_{\cS(\alpha, R_n)^c} \E_g\|\hf-f\|^2 \dif\pi(g)\bigg\}.
\end{align*}
If $\hf$ minimizes the above expression, then $\E_g\|\hf\|^2 \leq 3R_n^2 + \sum_{k=1}^{\infty} \tlambda_k$; since otherwise $ \E_g\|\hf-f\|^2 \geq (1/2)R_n^2 -\sum_{k=1}^{\infty}\tilde\lambda_k \gtrsim R_n^2$ for all $g\in \cS(\alpha, R_n)$, which leads to a contradiction in view of the upper bound derived in Theorem~\ref{thm:upperboundf}. 
Since $\sum_{k=1}^{\infty} \tlambda_k=O(1)=o(R_n)$, we get $\E_g\|\hf\|^2 \leq \big(3+o(1) \big)R_n^2$, which together with
\eqref{eq:gwnbayesrisklowerboundf1} and $\E_g\|\hf-f\|^2 \leq 2\E_g\|\hf\|^2+2\E_g\|f\|^2$ imply that the right hand side of the preceding display is bounded below by a multiple of
\begin{align*}
n^{-\frac{2\talpha}{1+2\talpha}} + (nm)^{-\frac{2\alpha}{1+2\alpha}} - C\int_{\cS(\alpha, R_n)^c}\left(R_n^2+\E_g\|f\|^2\right)\dif\pi(g),
\end{align*}
for some universal constant $C>0$. We show below that for $R_n=\rho\sqrt{\log n}$, 
\begin{align} 
\int_{\cS(\alpha, R_n)^c} \left(R_n^2+\E_g\|f\|^2\right) \dif\pi(g)=O(n^{-C_\rho}),\label{eq:help}
\end{align}
for any $C_\rho<\rho^2/4$ and hence the integral term is negligible for $\rho \geq 2$, finishing the proof of the statement.

To prove the preceding panel, we let $\|g\|_{\cS_\alpha}^2 = \sum_k g_k^2 k^{2\alpha}$ and first note 
\begin{align*}
\int_{\cS(\alpha, R_n)^c}\E_g\|f\|^2\dif\pi(g)
&= \sum_{r=0}^\infty \int_{2^{r}R_n^2\leq \|g\|_{\cS_\alpha}^2\leq 2^{r+1}R_n^2 }\E_g\|f\|^2\dif\pi(g)\\
&\lesssim R_n^2\sum_{r=0}^\infty  2^{r+1} \pi\{\cS(\alpha, 2^{r/2} R_n)^c\},
\end{align*}
where in the last line we have used that $\E_g\|f\|^2=\|g\|^2+\sum_{k=1}^\infty\tlambda_k\lesssim \|g\|_{S_\alpha}^2+O(1)$.
To upper bound $\pi\{\cS(\alpha, 2^{r/2} R_n)^c\}$, we note $g_j^2\stackrel{d}{=}\lambda_j Z_j^2$, with $Z_j\stackrel{iid}{\sim}\text{N}(0,1)$ and recall Theorem~3.1.9 from \cite{gine}, which states
\begin{align*}
P\bigg(\sum_{j=1}^\infty \lambda_j (Z_j^2-1)\geq t\bigg) \leq \exp\bigg\{-\frac{t^2}{4\sum_{j=1}^\infty \lambda_j+ t\lambda_1}\bigg\} 
\end{align*}
(in the above theorem only a finite sequence is considered, but the proof goes through with an infinite sequence of $Z_j$s as well). Hence 
\begin{align}
\pi\{\cS(\alpha, 2^{r/2} R_n)^c\}&= P\bigg(\sum_{j=1}^\infty \lambda_j (Z_j^2-1)\geq 2^r R_n^2 -\sum_{j=1}^\infty \lambda_j\bigg)\nonumber\\
&\leq  \exp\bigg\{-\frac{2^{2r-1} R_n^4}{4\sum_{j=1}^\infty \lambda_j+ 2^rR_n^2 \lambda_1}\bigg\}\leq \exp\{-2^{r-2}R_n^2\}.\label{eq: UB:complement}
\end{align}
Combining the preceding displays, we get
\begin{align*}
\int_{\cS(\alpha, R_n)^c}\E_g\|f\|^2\dif\pi(g)
\lesssim R_n^2\sum_{r=0}^\infty 2^{r+1} \exp\{-2^{r-2}R_n^2\}
\lesssim R_n^2 \exp\{-R_n^2/4\}
=o(n^{-C_\rho}),
\end{align*}
for any $C_\rho<\rho^2/4$. Similarly, we get that
\begin{align*}
    \int_{\cS(\alpha, R_n)^c} R_n^2 \dif\pi(g) \lesssim R_n^2 \exp\{-R_n^2\} = o(n^{-C_\rho}),
\end{align*}
finishing the proof of the statement.
\end{proof}

\section{Proof of Theorem~2: Minimax LB for $g$}
\label{sec:proofminimaxriskg}

\begin{proof}[Proof of Theorem~2]
Let $\pi = \pi_{\alpha}$ be the GP prior on $g$ with regularity hyperparameter $\alpha$.
Similarly to before, to reduce visual clutter we let $\|\cdot\|$ denote the $L_2$ norm and set $R_n \coloneqq 2\sqrt{\log n}$. Then, as before, note that 
\begin{align*}
    \inf_{\hg} \sup_{g \in \cS(\alpha, R_n)} \E_g\|\hg-g\|^2 
    &\geq \inf_{\hg}\Big\{\int \E_g\|\hg-g\|^2 \dif\pi(g)-
    \int_{\cS(\alpha, R_n)^c} \E_g\|\hg-g\|^2 \dif\pi(g)\Big\}.
\end{align*}
If $\hg$ minimizes the above expression, then $\E_g\|\hg\|^2\leq 3R_n^2$; since otherwise $\E_g\|\hat{g}-g\|^2\geq 2^{-1}\E_g\|\hat{g}\|^2-\|g\|^2\geq (3/2-1)R_n^2$ for all $g\in \cS(\alpha, R_n)$, which would mean $\hg$ does not minimize the expression in view of the upper bound derived in Theorem~\ref{thm:upperboundg}. 
This inequality, \eqref{eq:MISE:Bayesg}, and $\E_g\|\hg-g\|^2 \leq 2\E_g\|\hg\|^2+2\|g\|^2$ together imply that the right hand side of the preceding display is bounded below by a multiple of
\begin{align}
    m^{-1}+(nm)^{-\frac{2\alpha}{1+2\alpha}}- C\int_{\cS(\alpha, R_n)^c}\left(R_n^2+\|g\|^2\right)\dif\pi(g),\label{eq:LB:g}
\end{align}
for some universal constant $C>0$. 
We let $\|g\|_{\cS_\alpha}^2 = \sum_k g_k^2 k^{2\alpha}$ and note
\begin{align*}
\int_{\cS(\alpha, R_n)^c}\|g\|^2\dif\pi(g)
&= \sum_{r=0}^\infty \int_{2^{r}R_n^2\leq \|g\|_{\cS_\alpha}^2\leq 2^{r+1}R_n^2 }\|g\|^2\dif\pi(g)
\leq R_n^2\sum_{r=0}^\infty 2^{r+1} \pi\{\cS(\alpha, 2^{r/2} R_n)^c\}.
\end{align*}
Furthermore, similarly to \eqref{eq:help} (where in the first step we have used that $\|f\|^2\lesssim \|g\|^2+\sum_k\tilde\lambda_k$) the last term in \eqref{eq:LB:g} is $O(n^{-C_\rho})$ for any $C_\rho<\rho^2/4$. Since this term is negligible compared to the first two terms in \eqref{eq:LB:g}, it concludes our proof.
\end{proof}

\section{Proof of the posterior mean portion of Theorem~7: Posterior-mean UB for $\fo$}
\label{sec:proof_ubpostmean_f}

\begin{proof}
For notational convenience we omit the super index from $f^{(1)}$, i.e. we write $f=f^{(1)}$, and consider sample size $m+1$ instead of $m$.
We let $\hat{f}_k$ be component $k$ of the mean of the posterior distribution of $f|Y^{(1)},\ldots,Y^{(m+1)}$ derived in Section~\ref{sec:posteriorf}, i.e., we let 
\begin{align*}
    \hat{f}_k \coloneqq \frac{Y^{(1)}_k n+\bar{Y}_ka_k}{n+\tlambda_k^{-1}b_k},
\end{align*}
where $\bar{Y}_k = m^{-1} \sum_{j=2}^{m+1} Y_k^{(j)}$, $c_k \coloneqq \lambda_k^{-1}+\tlambda_k^{-1}+m(\tlambda_k+1/n)^{-1}$, 
\begin{align*}
    a_k = \tlambda_k^{-1} m(\tlambda_k+1/n)^{-1} c_k^{-1} \quad\text{and}\quad b_k = \big\{\lambda_k^{-1}+m(\tlambda_k+1/n)^{-1}\big\} c_k^{-1}.
\end{align*}

For any given function $g$, we denote by $\pi_g^{(1)}$ the distribution of $f^{(1)}|g$.
We also recall $\E_{f,g}$ is the conditional expectation given $f^{(1)}$ and $g$.
Then 
\begin{align*}
    \E_{g}(\hat{f}_k-f_k)^2 
    = \E^{\pi_g^{(1)}} \E_{f,g}(\hat{f}_k-f_k)^2 
    &= \E^{\pi_g^{(1)}} \Big[ \big(\E_{f,g}\hat{f}_k-f_k\big)^2  + \E_{f,g}\big(\E_{f,g}\hat{f}_k-\hat{f}_k\big)^2  \Big]
\end{align*}
where $\E_{f,g}\big(\E_{f,g}\hat{f}_k-\hat{f}_k\big)^2 = (n+\tlambda_k^{-1}b_k)^{-2} \big[n + a_k^2m^{-1}(\tlambda_k+1/n)\big]$ does not depend on either $f$ or $g$, and 
\begin{align*}
	(n+\tlambda_k^{-1}b_k)^2 \E^{\pi_g^{(1)}}\big[\E_{f,g}\hat{f}_k-f_k\big]^2
	&= \E^{\pi_g^{(1)}}\Big[\big(\tlambda_k^{-1}b_k\big)^2f_k^2 + a_k g_k^2 - 2\tlambda_k^{-1}b_k a_k f_kg_k \Big] \\
	&= \big(\tlambda_k^{-1}b_k - a_k\big)^2 g_k^2 + \tlambda_k^{-1}b_k^2.
\end{align*}

We have $\E_{g}(\hat{f}_k-f_k)^2= \gamma_k^2 g_k^2 + \delta_k$, where we introduce the terms
\begin{align*}
    \gamma_k 
    = \frac{\tlambda_k^{-1}b_k - a_k }{n+\tlambda_k^{-1}b_k}
    = \frac{1}{n \tlambda_k + 1 + n(m+1)\lambda_k}
    \quad \text{and} \quad
    \delta_k = \frac{n + a_k^2m^{-1}(\tlambda_k+1/n) + \tlambda_k^{-1}b_k^2}{(n+\tlambda_k^{-1}b_k)^2}.
\end{align*}
By plugging in $\lambda_k=k^{-1-2\alpha}$ into $\gamma_k$ in the preceding panel, we get
\begin{align*}
    \gamma_k k^{-\alpha} < \frac{ k^{-\alpha} }{1+nmk^{-1-2\alpha}}<
    \begin{cases}
        (nm)^{-1}k^{1+\alpha} &\text{if } k \leq (nm)^{1/(1+2\alpha)},\\
        k^{-\alpha} &\text{if } k \geq (nm)^{1/(1+2\alpha)},
    \end{cases}
\end{align*}
where both terms on the right hand side are $O((nm)^{-\alpha/(1+2\alpha)})$. 
We show below that 
\begin{align}
\sum_{k=1}^{\infty} \delta_k &\lesssim (nm)^{-\frac{2\alpha}{1+2\alpha}} + n^{-\frac{2\talpha}{1+2\talpha}}\label{eq: UB:postmean:2}
\end{align}
holds. 
The two preceding panels together then imply that
\begin{align*}
    \sup_{g \in \cS(\alpha,R)} \sum_{k=1}^{\infty} (\gamma_k^2 g_k^2 + \delta_k) = \sup_{g \in \cS(\alpha,R)} \sum_{k=1}^{\infty} \gamma_k^2 k^{-2\alpha} g_k^2 k^{2\alpha} + \sum_{k=1}^{\infty} \delta_k 
    \leq C_R \Big[(nm)^{-\frac{2\alpha}{1+2\alpha}} + n^{-\frac{2\talpha}{1+2\talpha}}\Big]
\end{align*}
for some universal constant $C_R>0$ only depending on $R$, providing our statement.

It remains to prove \eqref{eq: UB:postmean:2}. Note that 
\begin{align*}
    \delta_k 
    &= \frac{n\{\lambda_k^{-1}+\tlambda_k^{-1}+\frac{m}{\tlambda_k+1/n}\}^2 + \tlambda_k^{-2}\frac{m}{\tlambda_k+1/n}+ \tlambda_k^{-1}\big\{\lambda_k^{-1}+\frac{m}{\tlambda_k+1/n}\big\}^2}{\Big[n\lambda_k^{-1}+\tlambda_k^{-1}\lambda_k^{-1}+(m+1)n\tlambda_k^{-1}\Big]^2}.
\end{align*}
As in Section~\ref{sec:biasterm}, we partition the natural numbers into two or three intervals which are defined according to the dominating term in the denominator.

For $\alpha \leq \talpha+(1/2+\talpha)\delta$, the boundaries are $t_1=t_2=(nm)^{1/(1+2\alpha)} \geq n^{1/(1+2\talpha)}$.
Note also that $k^{-1-2\talpha} < n^{-1}$ is equivalent to $k > n^{1/(1+2\talpha)}$. Then, for $k \leq t_1$,
\begin{align*}
 \sum_{k=1}^{t_1} \delta_k 
    &\asymp \sum_{k=1}^{t_1} \frac{n\tlambda_k^2\{\lambda_k^{-2}+\tlambda_k^{-2}+\frac{m^2}{(\tlambda_k+1/n)^2} \} + \frac{m}{\tlambda_k+1/n} + \tlambda_k\big\{\lambda_k^{-2}+\frac{m^2}{(\tlambda_k+1/n)^2}\big\}}{(nm)^2} \\
      &\asymp \sum_{k=1}^{t_1} \frac{\tlambda_k^2\lambda_k^{-2}}{nm^2} + \frac{1}{nm^2} + \frac{\tlambda_k^2}{n(\tlambda_k+1/n)^{2}} + \frac{1}{n^2m(\tlambda_k+1/n)} + \frac{\tlambda_k\lambda_k^{-2}}{(nm)^2} + \frac{\tlambda_k}{n^2(\tlambda_k+1/n)^{2}}  \\
      &\leq  \sum_{k=1}^{t_1} \Bigg[ \frac{k^{4\alpha-4\talpha}}{nm^2} + \frac{1}{nm} + \frac{k^{1+4\alpha-2\talpha}}{(nm)^2} \Bigg] \\
      &\qquad+ \sum_{k=1}^{n^{1/(1+2\talpha)}} \Bigg[ \frac{1}{n} +\frac{k^{1+2\talpha}}{n^2m}+ \frac{k^{1+2\talpha}}{n^2} \Bigg] + \sum_{k=n^{1/(1+2\talpha)}}^{t_1} \bigg[ nk^{-2-4\talpha} +\frac{1}{nm}+k^{-1-2\talpha} \bigg] \\
    &\lesssim n^{-2\talpha/(1+2\talpha)} + (nm)^{-2\alpha/(1+2\alpha)}.
\end{align*}
For $k \geq t_2=t_1= (nm)^{1/(1+2\alpha)} \geq n^{1/(1+2\talpha)}$ we have $\tlambda_k+1/n \asymp 1/n$ and thus
\begin{align*}
    \sum_{k=t_2}^{\infty} \delta_k 
    &\asymp \sum_{k=t_2}^{\infty} \frac{n\tlambda_k^2\{\lambda_k^{-2}+\tlambda_k^{-2}+\frac{m^2}{(\tlambda_k+1/n)^2}\} + \frac{m}{\tlambda_k+1/n} + \tlambda_k\big\{\lambda_k^{-2}+\frac{m^2}{(\tlambda_k+1/n)^2}\big\}}{\lambda_k^{-2}} \\
    &\asymp \sum_{k=t_2}^{\infty} n\tlambda_k^2 + n\lambda_k^2 + n^3m^2\tlambda_k^2 \lambda_k^2 + nm\lambda_k^2 + \tlambda_k + n^2m^2\tlambda_k\lambda_k^2 \\
    &=\sum_{k= (nm)^{1/(1+2\alpha)}}^{\infty} nk^{-2-4\talpha}+ n^3m^2k^{-4-4\talpha-4\alpha} + nm k^{-2-4\alpha} +  k^{-1-2\talpha} + n^2m^2 k^{-3-4\alpha-2\talpha} \\
    &\lesssim n^{-2\talpha/(1+2\talpha)} + (nm)^{-2\alpha/(1+2\alpha)}.
\end{align*}

For the remaining case $\alpha > \talpha+(1/2+\talpha)\delta$ the boundaries are $t_1=m^{1/(2\alpha-2\talpha)}$ and $t_2= n^{1/(1+2\talpha)}$. Note that $n^{(2\talpha-4\alpha)/(1+2\talpha)} m^2 < n^{-2\talpha/(1+2\talpha)}$ in this case.
Then, for $k \leq t_1$,
\begin{align*}
    \sum_{k=1}^{t_1} \delta_k 
    &\asymp \sum_{k=1}^{t_1} \frac{n\tlambda_k^2\{\lambda_k^{-2}+\tlambda_k^{-2}+\frac{m^2}{(\tlambda_k+1/n)^2}\} + \frac{m}{\tlambda_k+1/n} + \tlambda_k\big\{\lambda_k^{-2}+\frac{m^2}{(\tlambda_k+1/n)^2}\big\}}{(nm)^2} \\
    &\asymp \sum_{k=1}^{t_1} \frac{\tlambda_k^2\lambda_k^{-2}}{nm^2} + \frac{1}{nm^2} + \frac{1}{n} +\frac{\tlambda_k^{-1}}{n^2m} + \frac{\tlambda_k\lambda_k^{-2}}{(nm)^2} + \frac{\tlambda_k^{-1}}{n^2}\\ 
    &<\sum_{k=1}^{m^{1/(2\alpha-2\talpha)}} \frac{k^{4(\alpha-\talpha)}}{nm^2} + \frac{2}{n} + \frac{k^{1+4\alpha-2\talpha}}{(nm)^2} +\frac{2k^{1+2\talpha}}{n^2}
    \lesssim n^{-2\talpha/(1+2\talpha)} .
\end{align*}
For $t_1 < k < t_2$ we have 
\begin{align*}
    \sum_{k=t_1+1}^{t_2}\delta_k 
    &\asymp \sum_{k=t_1+1}^{t_2} \frac{n\{\lambda_k^{-2}+\tlambda_k^{-2}+m^2\tlambda_k^{-2}\} + \tlambda_k^{-2}m\tilde\lambda_k^{-1} + \tlambda_k^{-1}\big\{\lambda_k^{-2}+m^2\tlambda_k^{-2}\big\}}{n^2 \lambda_k^{-2}} \\
    &\asymp \sum_{k=t_1+1}^{t_2} n^{-1} + n^{-1}\tlambda_k^{-2}\lambda_k^2 + n^{-1}m^2\tlambda_k^{-2}\lambda_k^2 + n^{-2}m^2\tlambda_k^{-3}\lambda_k^2 + n^{-2} \tlambda_k^{-1} \\
    &\lesssim  \sum_{k=m^{1/(2\alpha-2\talpha)}}^{n^{1/(1+2\talpha)}} n^{-1}+n^{-1}m^2k^{4(\talpha-\alpha)}+ n^{-2}m^2 k^{1+6\talpha-4\alpha}+n^{-2}k^{1+2\talpha}\lesssim n^{-2\talpha/(1+2\talpha)}.
\end{align*}
Finally, for $k \geq t_2 > (nm)^{1/(1+2\alpha)}$,
\begin{align*}
    \sum_{k=t_2}^{\infty}\delta_k 
    &\asymp \sum_{k=t_2}^{\infty} n\tlambda_k^2\lambda_k^2\{\lambda_k^{-2}+\tlambda_k^{-2}+m^2n^2\} + nm\lambda_k^2 + \tlambda_k\lambda_k^2\big\{\lambda_k^{-2}+m^2n^2\big\} \\
    &\asymp \sum_{k=t_2}^{\infty} n\tlambda_k^2 + n\lambda_k^2 + n^3m^2\tlambda_k^2 \lambda_k^2 + nm\lambda_k^2 + \tlambda_k + n^2m^2\tlambda_k\lambda_k^2 \\
    &=\sum_{k=n^{1/(1+2\talpha)}}^{\infty} nk^{-2-4\talpha}+ n^3m^2k^{-4-4\talpha-4\alpha} + n(m+1)k^{-2-4\alpha} + k^{-1-2\talpha} + n^2m^2k^{-3-2\talpha-4\alpha}\\
    &\lesssim n^{-2\talpha/(1+2\talpha)} + (nm)^{-2\alpha/(1+2\alpha)}, 
\end{align*}
finishing the proof of the statement.
\end{proof}

\section{Posterior-mean LB for misspecified $\fo$}
\label{sec:proof_lbpostmeanmisspecified_f}

\begin{theorem} 
    Let us consider model \eqref{eq:gwnmodelfreq}. 
    Let us set $\zeta_k \asymp k^{-1-2\beta}$ and $\tzeta_k \asymp k^{-1-2\tbeta}$ to be the (possibly) misspecified versions of the eigenvalues $\lambda_k$ and $\tlambda_k$ and denote by $\hat{f}_{\beta,\tilde\beta}^{(1)}$ the corresponding, (possibly) misspecified posterior mean.\\
  If   $\beta=\alpha$, then for any $\tilde\beta>0$ we have 
    \begin{align*}
        \sup_{g \in \cS(\alpha,R)} \E_g \|\hfo_{\alpha,\tilde\beta}-\fo\|_2^2
        & 
        \gtrsim n^{-\frac{2(\tilde\beta\wedge\talpha)}{1+2\tilde\beta}} + (nm)^{-\frac{2\alpha}{1+2\alpha}}.
    \end{align*}
    If $\tilde\beta = \tilde\alpha$, then for any $\beta>0$ we have 
    \begin{align*}
        \sup_{g \in \cS(\alpha,R)} \E_g \|\hfo_{\beta,\tilde\alpha}-\fo\|_2^2 &\gtrsim
        \begin{cases}
            n^{-\frac{2\talpha}{1+2\talpha}} + (nm)^{-\frac{2(\beta\wedge\alpha)}{1+2\beta}} & \text{for } \beta \leq \talpha + \delta(1/2+\talpha), \\
            n^{-\frac{2(\alpha\wedge\talpha)}{1+2\talpha}}  & \text{for } \beta > \talpha + \delta(1/2+\talpha).
        \end{cases} 
    \end{align*}
\end{theorem}

\begin{proof}
In view of Section~\ref{sec:posteriorf}, the posterior mean $\hat{f}_{\beta,\tilde\beta}^{(1)} \equiv (\hat{f}_k^{(1)})_{k=1}^{\infty}$ has coordinates
\begin{align*}
    \hat{f}_k^{(1)} \coloneqq \frac{Y^{(1)}_k n+\bar{Y}_ka_k}{n+\tzeta_k^{-1}b_k},
\end{align*}
with $\bar{Y}_k = m^{-1} \sum_{j=2}^{m+1} Y_k^{(j)}$,
$a_k = \tzeta_k^{-1} m(\tzeta_k+1/n)^{-1} c_k^{-1}$, 
$b_k = \big\{\zeta_k^{-1}+m(\tzeta_k+1/n)^{-1}\big\} c_k^{-1}$ and  
$c_k \coloneqq \zeta_k^{-1}+\tzeta_k^{-1}+m(\tzeta_k+1/n)^{-1}$. Then as before, for any given function $g$, we denote by $\pi_g^{(1)}$ the distribution of $f^{(1)}|g$ (which depends on the true $\alpha$ and $\talpha$ values).
We also omit the index $(1)$ from the notation $f^{(1)}$ for convenience and consider sample size $m+1$ instead of $m$, as in the preceding sections.
Then by the bias-variance decomposition
\begin{align*}
    \E_{g}(\hat{f}_k-f_k)^2 
    = \E^{\pi_g^{(1)}} \E_{f,g}(\hat{f}_k-f_k)^2 
    &= \E^{\pi_g^{(1)}} \Big[ \big(\E_{f,g}\hat{f}_k-f_k\big)^2  + \E_{f,g}\big(\E_{f,g}\hat{f}_k-\hat{f}_k\big)^2  \Big]
\end{align*}
where $\E_{f,g}\big(\E_{f,g}\hat{f}_k-\hat{f}_k\big)^2 = (n+\tzeta_k^{-1}b_k)^{-2} \big[n + a_k^2m^{-1}(\tlambda_k+1/n)\big]$ not depending on either $f$ or $g$, and 
\begin{align*}
	(n+\tzeta_k^{-1}b_k)^2 \E^{\pi_g^{(1)}}\big[\E_{f,g}\hat{f}_k-f_k\big]^2
	&= \E^{\pi_g^{(1)}}\Big[\big(\tzeta_k^{-1}b_k\big)^2f_k^2 + a_k^2 g_k^2 - 2\tzeta_k^{-1}b_k a_k f_kg_k \Big] \\
	&= \big(\tzeta_k^{-1}b_k - a_k\big)^2 g_k^2 + \tzeta_k^{-2}b_k^2\tlambda_k.
\end{align*}
Then, similarly to the well-posed case, $\E_{g}(\hat{f}_k-f_k)^2 \asymp \gamma_k^2 g_k^2 + \delta_k$, where $\gamma_k \coloneqq [n \tzeta_k + 1 + n(m+1)\zeta_k]^{-1}$
and 
\begin{align*}
    \delta_k &\coloneqq \frac{n + a_k^2m^{-1}(\tlambda_k+1/n) + \tzeta_k^{-2}b_k^2\tlambda_k}{(n+\tzeta_k^{-1}b_k)^2} \\
    &= \frac{n\{\zeta_k^{-1}+\tzeta_k^{-1}+\frac{m}{\tzeta_k+1/n} \}^2 + \tzeta_k^{-2}\frac{m(\tlambda_k+1/n)}{(\tzeta_k+1/n)^2} + \tzeta_k^{-2}\tlambda_k\big\{\zeta_k^{-1}+\frac{m}{\tzeta_k+1/n} \big\}^2}{\Big[n\zeta_k^{-1}+\tzeta_k^{-1}\zeta_k^{-1}+(m+1)n\tzeta_k^{-1}\Big]^2}.
\end{align*}

\subsection{If $\beta=\alpha$} 
In this case $\zeta_k=\lambda_k$.
Let us distinguish two cases according to the value of $\alpha$. If $\alpha \leq \tbeta+(1/2+\tbeta)\delta$, then $(nm)^{1/(1+2\alpha)}\geq n^{1/(1+2\tilde\beta)}$, hence the term $(m+1)n\tzeta_k^{-1}$ dominates the denominator in $\delta_k$ for $k\leq n^{1/(1+2\tilde{\beta})}$. Therefore,
\begin{align*}
    \sum_{k=1}^{\infty} \delta_k 
    \gtrsim \sum_{k=1}^{n^{1/(1+2\tbeta)}} \frac{(n+\tzeta_k^{-2}\tlambda_k )\frac{m^2}{(\tzeta_k+1/n)^{2}}}{\big[(m+1)n\tzeta_k^{-1}\big]^2}\asymp \sum_{k=1}^{n^{1/(1+2\tbeta)}}  \frac{n+k^{1+4\tbeta-2\talpha}}{n^2}
    \gtrsim n^{-\frac{2(\tbeta\wedge\talpha)}{1+2\tbeta}}.
\end{align*}
Furthermore, if $g(\cdot) = g_{\ell}\psi_{\ell}(\cdot)$ with $g_{\ell}^2=R^2 \ell^{-2\alpha}$ and $\ell=\ell_{n,m}=(nm)^{1/(1+2\alpha)}$, then $g \in S(\alpha,R)$ and $\gamma_\ell \asymp 1$, which gives us
\begin{align*}
    \E_g \|\hfo_{\alpha,\tilde\beta}-\fo\|_2^2 \gtrsim \gamma_{\ell}^2g_{\ell}^2 \asymp R^2 \ell^{-2\alpha}=R^2 (nm)^{-2\alpha/(1+2\alpha)}.
\end{align*}

Next we consider the case $\alpha \geq \tbeta+(1/2+\tbeta)\delta$, where for $k\geq n^{1/(1+2\tbeta)}$ the term $\tzeta_k^{-1}\zeta_k^{-1}$ dominates the sum in the denominator of $\delta_k$. Then
\begin{align*}
    \sum_{k=1}^{\infty} \delta_k 
    \gtrsim \sum_{k=n^{1/(1+2\tbeta)}}^{\infty} \frac{n\zeta_k^{-2}+ \tzeta_k^{-2}\tlambda_k\zeta_k^{-2}}{\tzeta_k^{-2}\zeta_k^{-2}} 
    \gtrsim  \sum_{k=n^{1/(1+2\tbeta)}}^{\infty} nk^{-2-4\tbeta}+k^{-1-2\talpha}
    \asymp n^{-\frac{2(\talpha\wedge\tbeta)}{1+2\tbeta}}. 
\end{align*}
Furthermore, note that $(nm)^{-2\alpha/(1+2\alpha)}\leq (nm)^{-2\tbeta/(1+2\tbeta)}\leq  n^{-2\tbeta/(1+2\tbeta)}$, hence concluding the proof of the statement.

\subsection{If $\tbeta=\talpha$} 
In this case  $\tlambda_k \asymp \tzeta_k$ and similarly to above we distinguish two cases according to the value of $\beta$.  If $\beta \leq \talpha+(1/2+\talpha)\delta$,  then $(nm)^{1/(1+2\beta)}\geq n^{1/(1+2\tilde\alpha)}$, hence for $k\leq n^{1/(1+2\talpha)}$ the term $(m+1)n\tzeta_k^{-1}$ dominates the denominator in $\delta_k$, implying
 \begin{align*}
    \sum_{k=1}^{\infty} \delta_k 
    \gtrsim \sum_{k=1}^{n^{1/(1+2\talpha)}} \frac{n\frac{m^2}{(\tzeta_k+1/n)^{2}}}{\big[(m+1)n\tzeta_k^{-1}\big]^2}\asymp \sum_{k=1}^{n^{1/(1+2\talpha)}}  n^{-1}\asymp n^{-\frac{2\talpha}{1+2\talpha}}.
\end{align*}
Furthermore, for $k\geq (nm)^{1/(1+2\beta)}$ the term $\tzeta_k^{-1}\zeta_k^{-1}$ denominates the sum, hence
\begin{align*}
    \sum_{k=1}^{\infty} \delta_k 
    \gtrsim \sum_{k=(nm)^{1/(1+2\beta)}}^{\infty} \frac{ \tzeta_k^{-2}\frac{m(\tlambda_k+1/n)}{(\tzeta_k+1/n)^2} }{\big[\tzeta_k^{-1}\zeta_k^{-1}\big]^2}\asymp \sum_{k=(nm)^{1/(1+2\beta)}}^{\infty} n \zeta_k^{2}  \asymp (nm)^{-\frac{2\beta}{1+2\beta}}.
\end{align*}
Finally, if $g(\cdot) = g_{\ell}\psi_{\ell}(\cdot)$ with $g_{\ell}^2=R^2 \ell^{-2\alpha}$ and $\ell=\ell_{n,m}=(nm)^{1/(1+2\beta)}$, then $g \in S(\alpha,R)$ and $\gamma_\ell \asymp 1$, which gives us
\begin{align*}
    \E_g \|\hfo_{\beta,\tilde\alpha}-\fo\|_2^2\gtrsim\gamma_{\ell}^2g_{\ell}^2\asymp R^2 \ell^{-2\alpha}=R^2 (nm)^{-2\alpha/(1+2\beta)}.
\end{align*}

In the remaining case $\beta \geq \talpha+(1/2+\talpha)\delta$, for $k\geq n^{1/(1+2\talpha)}$ the term $\tzeta_k^{-1}\zeta_k^{-1}$ dominates the sum in the denominator of $\delta_k$. Then
\begin{align*}
    \sum_{k=1}^{\infty} \delta_k 
    \gtrsim \sum_{k=n^{1/(1+2\talpha)}}^{\infty} \frac{n\zeta_k^{-2}}{\tzeta_k^{-2}\zeta_k^{-2}} 
    \gtrsim  \sum_{k=n^{1/(1+2\talpha)}}^{\infty} nk^{-2-4\talpha}
    \asymp n^{-\frac{2\talpha}{1+2\talpha}}. 
\end{align*}
Furthermore, if $g(\cdot) = g_{\ell}\psi_{\ell}(\cdot)$ with $g_{\ell}^2=R^2 \ell^{-2\alpha}$ and $\ell=\ell_{n}=n^{1/(1+2\talpha)}\geq (nm)^{1/(1+2\beta)}$, then $g \in S(\alpha,R)$ and $\gamma_\ell \asymp 1$, which gives us
\begin{align*}
    \E_g \|\hfo_{\beta,\tilde\alpha}-\fo\|_2^2\gtrsim\gamma_{\ell}^2g_{\ell}^2\asymp R^2 \ell^{-2\alpha}=R^2 (nm)^{-2\alpha/(1+2\beta)},
\end{align*}
finishing the proof of the statement.
\end{proof}

\section{Proof of posterior mean portion of Theorem~3: Posterior-mean UB for $g$}
\label{sec:proof_ubpostmean_g}

\begin{proof}
Letting $\hg$ be the posterior mean under prior $\pi_\alpha = \pi_{\Lambda}$ and recalling the notation $V_g$ and $B_g$ from Section~\ref{sec:proofbayesriskg}, we will show
\begin{align*}
    \sup_{g \in \cS(\alpha, R)} V_g(\hg) + B_g(\hg) \lesssim m^{-1} + (nm)^{-\frac{2\alpha}{1+2\alpha}}.
\end{align*}
Since \eqref{eq:gwnbayesriskgvar} implies the term $V_g(\hg) \asymp m^{-1} + (nm)^{-2\alpha/(1+2\alpha)}$ and does not depend on $g$, it suffices to show $\sup_{g \in \cS(\alpha, R)} B_g(\hg) \lesssim m^{-1} + (nm)^{-2\alpha/(1+2\alpha)}$, where by Section~\ref{sec:proofbayesriskg}  
\begin{align*}
    B_g(\hg)
    \asymp m^{-2} \sum_{k=1}^{t_1} \Big(k^{4\alpha-4\tilde\alpha} + n^{-2} k^{2+4\alpha} \Big)g_k^2
    + \sum_{k=t_1}^{t_2} \Big(1+n^{-2} k^{2+4\tilde\alpha}\Big) g_k^2
    + \sum_{k=t_2}^{\infty} \Big(n^2 k^{-2-4\tilde\alpha} + 1\Big) g_k^2
\end{align*}
with $t_1$ and $t_2$ defined in Section~\ref{sec:proofbayesriskf}. 
If we express the preceding panel as $B_g(\hg) \asymp \sum_{k=1}^{t_1} a_k g_k^2 + \sum_{k=t_1}^{t_2} b_k g_k^2 + \sum_{k=t_2}^{\infty} c_k g_k^2$, 
we will show below that 
\begin{align*}
    a_k k^{-2\alpha} &= m^{-2} \Big(k^{2\alpha-4\tilde\alpha} + n^{-2} k^{2+2\alpha} \Big) \qquad \text{for all }k \in [1, t_1], \\
    b_k k^{-2\alpha} &= k^{-2\alpha}+n^{-2} k^{2+4\tilde\alpha-2\alpha} \qquad \text{for all }k \in (t_1, t_2], \\
    c_k k^{-2\alpha} &= n^2 k^{-2-4\tilde\alpha-2\alpha} + k^{-2\alpha} \qquad \text{for all }k \geq t_2
\end{align*}
are all bounded above by $2 \{m^{-1} \vee (nm)^{-2\alpha/(1+2\alpha)}\}$.
This will allow us to conclude $B_g(\hg) \lesssim \{m^{-1} \vee (nm)^{-2\alpha/(1+2\alpha)}\} \sum_{k=1}^{\infty} g_k^2 k^{2\alpha}$, where $\sup_{g \in \cS(\alpha,R)} \sum_{k=1}^{\infty} g_k^2 k^{2\alpha} \leq R^2$ by definition of the Sobolev ball. 

We distinguish two cases. For $\alpha \leq \talpha+(1/2+\talpha)\delta$ the boundaries are $t_1=t_2=(nm)^{1/(1+2\alpha)} \geq n^{1/(1+2\talpha)}$, in which case 
\begin{align*}
    a_k k^{-2\alpha} &\leq m^{-2} \left(1 \vee (nm)^{(2\alpha-4\tilde\alpha)/(1+2\alpha)}\right) + (nm)^{-2\alpha/(1+2\alpha)} \qquad \text{for all }k \in [1, t_1], \\
    c_k k^{-2\alpha} &\leq 2 (nm)^{-2\alpha/(1+2\alpha)} \qquad \text{for all }k \geq t_2.
\end{align*}
If $\alpha>2\talpha$, then $m^{-2} k^{2\alpha-4\tilde\alpha} \leq (nm)^{-2\alpha/(1+2\alpha)}$ for all $k \leq t_1$.
For the remaining case $\alpha > \talpha+(1/2+\talpha)\delta$ the boundaries are $t_1=m^{1/(2\alpha-2\talpha)}$ and $t_2= n^{1/(1+2\talpha)}$, in which case we have $n^{-2} < m^{-(1+2\tilde\alpha)/(\alpha-\tilde\alpha)}$ and thus
\begin{align*}
    a_k k^{-2\alpha} &\leq m^{-2} \left(1 \vee m^{1-\frac{\tilde\alpha}{\alpha-\talpha}}\right) + m^{-2} m^{1-\frac{\tilde\alpha}{\alpha-\talpha}} \leq 2 m^{-\frac{\alpha}{\alpha-\talpha}} \qquad \forall k \in [1, t_1], \\
    b_k k^{-2\alpha} &\leq m^{-\frac{\alpha}{\alpha-\talpha}} + n^{-2} \left(m^{\frac{1+2\tilde\alpha-\alpha}{\alpha-\talpha}} \vee n^{\frac{2+4\tilde\alpha-2\alpha}{1+2\talpha}}\right) \leq m^{-\frac{\alpha}{\alpha-\talpha}} + \left(m^{-\frac{\alpha}{\alpha-\talpha}} \vee n^{-\frac{2\alpha}{1+2\talpha}}\right) \; \forall k \in (t_1, t_2], \\
    c_k k^{-2\alpha} &\leq 2n^{-\frac{2\alpha}{1+2\talpha}}  \qquad \forall k \geq t_2,
\end{align*}
where the inequalities $n^{-2\alpha/(1+2\talpha)} < (nm)^{-2\alpha/(1+2\alpha)}$ and $m^{-\alpha/(\alpha-\talpha)} < m^{-1}$ hold.
\end{proof}

\section{Posterior-mean LB for misspecified $g$}
\label{sec:proof_lbpostmeanmisspecified_g}

\begin{theorem}
    For $\beta>0$, let $\pi_{\beta}$ be the mean-zero Gaussian process distribution whose covariance function $\Lambda_{\beta}$ shares the same eigenfunctions as the covariance function $\tilde\Lambda$ and has eigenvalues $\zeta_k \asymp k^{-1-2\beta}$ for all $k\in\mathbb{N}$. 
    If $\tilde{g}_{\beta}$ is the posterior mean associated with the prior $\pi_{\beta}$, then under model \eqref{eq:gwnmodelfreq} we have
    \begin{align*}
        \sup_{g \in \cS(\alpha,R)} \E_g \|\tilde{g}_{\beta}-g\|_2^2
        \gtrsim 
        m^{-1} + (nm)^{-2(\beta\wedge\alpha)/(1+2\beta)}.
    \end{align*}
\end{theorem}

\begin{proof}
Recalling the notation $V_g$ and $B_g$ from Section~\ref{sec:proofbayesriskg}, we will show
\begin{align*}
    \sup_{g \in \cS(\alpha, R)} V_g(\tilde{g}_{\beta}) + B_g(\tilde{g}_{\beta}) \gtrsim m^{-1} + (nm)^{-2(\beta\wedge\alpha)/(1+2\beta)}. 
\end{align*}

From \eqref{eq:gwnbayesriskgvar} we have $V_g(\tilde{g}_{\beta}) \asymp m^{-1} + (nm)^{-2\beta/(1+2\beta)}$ which does not depend on $g$.
If $\beta \leq \alpha$, this result implies the preceding panel holds. 
Now it remains to show for $\beta > \alpha$ that 
\begin{align*}
    \sup_{g \in \cS(\alpha, R)} B_g(\tilde{g}_{\beta}) \gtrsim (nm)^{-2\alpha/(1+2\beta)}.
\end{align*}

By \eqref{eq:gwnbayesriskgbias} we have
\begin{align*} 
    B_g(\tilde{g}_{\beta})
    = \sum_{k=1}^{\infty} \left(\frac{n\tlambda_k + 1}{nm\zeta_k + n\tlambda_k + 1}\right)^2 g_k^2.
\end{align*}
We follow the same strategy as in Section~\ref{sec:proofbayesriskg} of partitioning $\mathbb{N}$ into three intervals:
\begin{align*}
    B_g(\tilde{g}_{\beta})
    &\asymp m^{-2} \sum_{k=1}^{t_1} \Big(k^{4\beta-4\tilde\alpha} + n^{-2} k^{2+4\beta} \Big)g_k^2
    + \sum_{k=t_1}^{t_2} \Big(1+n^{-2} k^{2+4\tilde\alpha}\Big) g_k^2
    + \sum_{k=t_2}^{\infty} \Big(n^2 k^{-2-4\tilde\alpha} + 1\Big) g_k^2 \\
    &> m^{-2} \Big(t_1^{4\beta-4\tilde\alpha} + n^{-2} t_1^{2+4\beta} \Big)g_{t_1}^2.
\end{align*}
Denote this lower bound as $d_{t_1} g_{t_1}^2$.
Because the preceding inequality implies \[\sup_{g \in \cS(\alpha, R)} B_g(\tilde{g}_{\beta}) \geq \sup_{g \in \cS(\alpha, R)} d_{t_1} g_{t_1}^2 = d_{t_1} t_1^{-2\alpha} R^2,\] it suffices to show $d_{t_1} t_1^{-2\alpha} \geq (nm)^{-2\alpha/(1+2\beta)}$.
We distinguish two cases. 
For $\beta \leq \talpha+(1/2+\talpha)\delta$ the boundary is $t_1=(nm)^{1/(1+2\beta)}$, in which case 
\begin{align*}
    d_{t_1} t_1^{-2\alpha} \geq (nm)^{-2} t_1^{2+4\beta-2\alpha} = (nm)^{-\frac{2\alpha}{1+2\beta}}. 
\end{align*}
For the remaining case $\beta > \talpha+(1/2+\talpha)\delta$ the boundary is $t_1=m^{\frac{1}{2(\beta-\talpha)}}$, in which case
\begin{align*}
    d_{t_1} t_1^{-2\alpha} \geq m^{-2} t_1^{4\beta-4\tilde\alpha-2\alpha} = m^{-\frac{\alpha}{\beta-\tilde\alpha}} > (nm)^{-\frac{2\alpha}{1+2\beta}}. 
\end{align*}
Hence $d_{t_1} t_1^{-2\alpha} \geq (nm)^{-2\alpha/(1+2\beta)}$ whenever $\beta > \alpha$, which finishes the proof. 
\end{proof}

\section{Proof of threshold estimator portion of Theorem~7: Threshold estimator UB for $\fo$}
\label{sec:proof_ubthresh_f}

\begin{proof}
We consider the estimator $\hat{f}^{(1)}_{k_1,k_2}(\cdot) = \sum_{k=1}^{\infty} \hat{f}_k^{(1)} \psi_k(\cdot)$ where
\begin{equation*}
\hat{f}_k^{(1)} =
\begin{cases}
Y_k^{(1)} & \text{if }k\leq k_1, \\
\hat{g}_k & \text{if }k_1< k\leq k_2, \\
0 & \textit{if }k_2<k.
\end{cases}
\end{equation*}
For notational convenience we omit the super index from $f^{(1)}$, i.e., we write $f=f^{(1)}$, and consider sample size $m+1$. 
Recall from Section~\ref{sec:proofbayesriskg} that $\hg_k|g_k\sim\text{N}\big(g_k, \tlambda_k/m+1/(nm)\big)$.
We introduce the following notations
\begin{align*}
B_g'(k_2) \coloneqq \sum_{\ell> k_2} \big(g_\ell^2+\tilde{\lambda}_\ell\big) \quad \text{and} \quad
V_g'(k_1,k_2) \coloneqq \frac{k_1}{n}+\sum_{k=k_1+1}^{k_2}\left\{\tlambda_k\Big(1+\frac{1}{m}\Big)+\frac{1}{nm}\right\}.
\end{align*}
First note
\begin{align}
\E_g\|f-\hf_{k_1,k_2}\|_2^2 &= \E^{\pi_g^{(1)}} \E_{f,g}\|f-\hf_{k_1,k_2}\|_2^2\nonumber \\
&= \E^{\pi_g^{(1)}} \|\E_{f,g}\hf_{k_1,k_2}-f\|_2^2 + \E^{\pi_g^{(1)}}\E_{f,g}\|\E_{f,g}\hf_{k_1,k_2}-\hf_{k_1,k_2}\|_2^2 \nonumber \\
&= \sum_{k=k_1+1}^{k_2} \E^{\pi_g^{(1)}} (f_k-g_k)^2 +  \sum_{k=k_2+1}^{\infty}\E^{\pi_g^{(1)}} f_k^2 + \sum_{k=1}^{k_1}\frac{1}{n} + \sum_{k=k_1+1}^{k_2}\E_g (\hg_k-g_k)^2 \nonumber \\
&= \sum_{k=k_1+1}^{k_2} \tlambda_k+\sum_{k=k_2+1}^{\infty} \big(g_k^2+\tlambda_k\big)+\frac{k_1}{n}+\sum_{k=k_1+1}^{k_2}\left(\frac{\tlambda_k}{m}+\frac{1}{nm}\right)\label{eq:misef}\\
&= B_g'(k_2)+V_g'(k_1,k_2) \nonumber.
\end{align}

By combining the preceding panels and using $\tlambda_k\asymp k^{-1-2\talpha}$, for any $g \in \cS(\alpha,R)$ we get
\begin{align*}
    B_g'(k_2) \leq c_R k_2^{-2(\alpha\wedge\talpha)} \quad \text{and} \quad V_g'(k_1,k_2) \lesssim \frac{k_1}{n} + \frac{\max\{0,k_2-k_1\}}{nm} + k_1^{-2\talpha},
\end{align*}
and thus
\begin{align*}
    \sup_{g \in \cS(\alpha,R)} \E_g \|\hf_{k_1,k_2}-f\|_2^2
    &\lesssim \frac{k_1}{n} + \frac{\max\{0,k_2-k_1\}}{nm}  + k_1^{-2\talpha} +  k_2^{-2\alpha}.
\end{align*}
Plugging in $k_1 \asymp n^{1/(1+2\talpha)}$ and 
\begin{align*}
    k_2 \asymp (k_1 \vee (nm)^{1/(1+2\alpha)}) \asymp 
    \begin{cases}
        (nm)^{1/(1+2\alpha)} & \text{if } \alpha < \talpha + \delta(1/2+\talpha), \\
        n^{1/(1+2\talpha)} & \text{if } \alpha \geq \talpha + \delta(1/2+\talpha),
    \end{cases} 
\end{align*}
into the preceding display we arrive at the stated upper bound $(nm)^{-\frac{2\alpha}{1+2\alpha}} + n^{-\frac{2\talpha}{1+2\talpha}}$.
\end{proof}

\section{Threshold estimator LB for misspecified $\fo$}
\label{sec:proof_lbthreshmisspecified_f}

\begin{theorem} 
    Assuming model \eqref{eq:gwnmodelfreq}, with a slight abuse of notation let $\tf_{\beta,\tilde\beta}$ denote the estimator \eqref{eq:estf} using the thresholds $k_1 \asymp n^{1/(1+2\tilde\beta)}$ and $k_2 \asymp (k_1 \vee (nm)^{1/(1+2\beta)})$. 
    If $\beta=\alpha$, then for any $\tilde\beta>0$ we have 
    \begin{align*}
        \sup_{g \in \cS(\alpha,R)} \E_g \|\hfo_{\alpha,\tilde\beta}-\fo\|_2^2
        & 
        \asymp n^{-\frac{2(\tilde\beta\wedge\talpha)}{1+2\tilde\beta}} + (nm)^{-\frac{2\alpha}{1+2\alpha}}.
    \end{align*}
    If $\tilde\beta = \tilde\alpha$, then for any $\beta>0$ we have 
    \begin{align*}
        \sup_{g \in \cS(\alpha,R)} \E_g \|\hfo_{\beta,\tilde\alpha}-\fo\|_2^2 &\asymp 
        \begin{cases}
            n^{-\frac{2\talpha}{1+2\talpha}} + (nm)^{-\frac{2(\beta\wedge\alpha)}{1+2\beta}} & \text{for } \beta \leq \talpha + \delta(1/2+\talpha), \\
            n^{-\frac{2(\alpha\wedge\talpha)}{1+2\talpha}}  & \text{for } \beta > \talpha + \delta(1/2+\talpha).
        \end{cases} 
    \end{align*}
\end{theorem}

\begin{proof}
We use the same notations as in the previous sections. 
Then, in view of \eqref{eq:misef}, for any two thresholds $k_1 \leq k_2$ we can infer
\begin{align*}
    \sup_{g \in \cS(\alpha,R)} \E_g\|f-\hf_{k_1,k_2}\|_2^2 
    &= \frac{k_1}{n}+ \frac{k_2-k_1}{nm}+\sum_{k=k_1+1}^{\infty} \tlambda_k+\sum_{k=k_1+1}^{k_2}\frac{\tlambda_k}{m} + \sup_{g \in \cS(\alpha,R)} \sum_{k=k_2+1}^{\infty} g_k^2 \\
    &\asymp \frac{k_1}{n}+ \frac{k_2-k_1}{nm}+k_1^{-2\tilde\alpha}+ k_2^{-2\alpha}.
\end{align*}
Let us now use the thresholds $k_1 = n^{1/(1+2\tilde\beta)}$ and $k_2 = (k_1 \vee (nm)^{1/(1+2\beta)})$.

\subsection{For $\beta=\alpha$} 

If $2n^{1/(1+2\tilde\beta)} \leq (nm)^{1/(1+2\alpha)}$, then $k_2 - k_1 \asymp k_2 = (nm)^{1/(1+2\alpha)}$, and thus
\begin{align*}
    \frac{k_1}{n}+ \frac{k_2-k_1}{nm}+k_1^{-2\tilde\alpha}+ k_2^{-2\alpha} 
    \asymp n^{-\frac{2\tilde\beta}{1+2\tilde\beta}} + n^{-\frac{2\tilde\alpha}{1+2\tilde\beta}} + (nm)^{-\frac{2\alpha}{1+2\alpha}}.
\end{align*}
If $2n^{1/(1+2\tilde\beta)} > (nm)^{1/(1+2\alpha)}$ (i.e., if $\alpha > \tilde\beta + (1/2+\tilde\beta)\delta$), then $k_2\asymp k_1 = n^{1/(1+2\tilde\beta)}$ and thus
\begin{align*}
    \frac{k_1}{n} + \frac{k_2-k_1}{nm} + k_1^{-2\tilde\alpha}+ k_2^{-2\alpha} 
    &\asymp n^{-\frac{2\tilde\beta}{1+2\tilde\beta}} + n^{-\frac{2\talpha}{1+2\tilde\beta}}
\end{align*}
in view of the inequalities $n^{-\frac{2\tilde\beta}{1+2\tilde\beta}} > (nm)^{-\frac{2\tilde\beta}{1+2\tilde\beta}} > (nm)^{-\frac{2\alpha}{1+2\alpha}}$.

\subsection{For $\tilde\beta=\tilde\alpha$} 

If $2n^{1/(1+2\talpha)} \leq (nm)^{1/(1+2\beta)}$, then $k_2 - k_1 \asymp k_2 = (nm)^{1/(1+2\beta)}$, and thus
\begin{align*}
    \frac{k_1}{n}+ \frac{k_2-k_1}{nm}+k_1^{-2\tilde\alpha}+ k_2^{-2\alpha} 
    \asymp n^{-\frac{2\talpha}{1+2\talpha}} + (nm)^{-\frac{2\beta}{1+2\beta}} + (nm)^{-\frac{2\alpha}{1+2\beta}}.
\end{align*}
If $2n^{1/(1+2\talpha)} > (nm)^{1/(1+2\beta)}$ (i.e., if $\beta > \tilde\alpha + (1/2+\tilde\alpha)\delta$), then $k_2\asymp k_1 = n^{1/(1+2\tilde\alpha)}$ and thus
\begin{align*}
    \frac{k_1}{n}+ \frac{k_2-k_1}{nm}+k_1^{-2\tilde\alpha}+ k_2^{-2\alpha} 
    \asymp n^{-\frac{2\talpha}{1+2\talpha}} + n^{-\frac{2\alpha}{1+2\talpha}}.
\end{align*}
\end{proof}

\section{Proof of the threshold estimator portion of Theorem~3: Threshold estimator rate for $g$}
\label{sec:proof_ubthresh_g}

Here we prove the threshold estimator portion of Theorem~3, which is a special case of the following theorem.
\begin{theorem} 
    Let $\beta>0$. 
    If $\tilde{g}_{\beta}$ is the estimator \eqref{eq:estg} using the threshold $K \asymp (nm)^{1/(1+2\beta)}$, then under model \eqref{eq:gwnmodelfreq} we have
    \begin{align*}
        \sup_{g \in \cS(\alpha,R)} \E_g \|\tilde{g}_{\beta}-g\|_2^2
        \asymp 
        m^{-1} + (nm)^{-2(\beta\wedge\alpha)/(1+2\alpha)}.
    \end{align*}
\end{theorem}
The threshold estimator portion of Theorem~3 is obtained by setting $\beta=\alpha$.

\begin{proof}
We consider estimators of the form \eqref{eq:estg} 
where $\hg_k = m^{-1}\sum_{j=1}^m Y^{(j)}_k$. 
Recall $\hg_k|g_k \sim \text{N}(g_k, m^{-1}(\tlambda_k+n^{-1}))$.
For any positive integer $K$ we have the following decomposition into a squared bias and variance term 
\begin{align*}
    \E_g\|\tilde{g}^K-g\|_2^2
    = B_g(K) + V_g(K)
\end{align*}
where we define
\begin{align*}
    B_g(K) \coloneqq \|\E_g \tilde{g}^K - g\|_2^2 = \sum_{k>K} g_{k}^2, \quad
    V_g(K) \coloneqq \E_g \|\E_g \tilde{g}^K - \tilde{g}^K\|_2^2 = \sum_{k=1}^{K} m^{-1} (\tlambda_k + n^{-1}).
\end{align*}
Because the eigenvalues satisfy $\tlambda_k\asymp k^{-1-2\talpha}$, we have $\sum_{k=1}^K \tlambda_k =O(1)$ and thus also
$V_g(K) \asymp m^{-1} + (nm)^{-1}K$.
From the two inequalities
\begin{align*}
    \sup_{g \in \cS(\alpha,R)} B_g(K)
    \leq  K^{-2\alpha}\sup_{g \in \cS(\alpha,R)} \sum_{k>K} g_k^2 k^{2\alpha} 
    = R^2 K^{-2\alpha}
\end{align*}
and 
\begin{align*}
    \sup_{g \in \cS(\alpha,R)} B_g(K) \geq \sup_{g \in \cS(\alpha,R)} g_{K+1}^2 = R^2 (K+1)^{-2\alpha},
\end{align*}
we get 
\begin{align*}
    \sup_{g \in \cS(\alpha,R)} B_g(K) + V_g(K) \asymp m^{-1} + (nm)^{-1} K + K^{-2\alpha}.
\end{align*}
Plugging in the threshold $K \asymp (nm)^{1/(1+2\beta)}$ gives the stated rate $m^{-1} + (nm)^{-2\beta/(1+2\beta)} + (nm)^{-2\alpha/(1+2\beta)} \asymp m^{-1} + (nm)^{-2(\beta\wedge\alpha)/(1+2\beta)}$.
\end{proof}

\section{Proof of Theorem~8: Adaptation for $\fo$}
\label{sec:proofadaptationfgwn}

First we recall (a simplified version of) Theorem 3.1.9 of \cite{gine} to be used in the proof of Theorem~\ref{thm:adaptivef} and consider sample size $m+1$ instead of $m$, as earlier.
\begin{lemma}\label{lem:concent}
For $Z_1,\ldots,Z_n \iid \text{N}(0,1)$ and a decreasing sequence $\gamma_1\geq...\geq \gamma_n>0$, we have
\begin{align*}
\prob \Big(\sum_{i=1}^n \gamma_i(Z_i^2-1)\geq t\Big)\leq \exp\Big\{- \frac{t^2}{4(\sum_{i=1}^n \gamma_i^2+ t\gamma_1)}\Big\}.
\end{align*}
\end{lemma}

\begin{proof}[Proof of Theorem~8] 
For notational convenience we omit the super index from $f^{(1)}$, i.e. we write $f=f^{(1)}$. 

Recall from the proof of Theorem~\ref{thm:upperboundf} in Section~\ref{sec:proof_ubthresh_f} that $\E_g\|f-\hf_{k_1,k_2}\|_2^2 = B_g'(k_2)+V_g'(k_1,k_2)$ where
\begin{align*}
    B_g'(k_2) &\coloneqq \sum_{\ell> k_2} (g_\ell^2+\tilde{\lambda}_\ell) \leq k_2^{-2\alpha} \sum_{\ell> k_2} g_\ell^2 \ell^{2\alpha}+  \sum_{\ell > k_2}  \tilde{\lambda}_\ell\leq c_R k_2^{-2(\alpha\wedge\talpha)},\\
    V_g'(k_1,k_2) &\coloneqq \frac{k_1}{n}+\sum_{k=k_1+1}^{k_2}\bigg\{\tlambda_k\Big(1+\frac{1}{m}\Big)+\frac{1}{nm}\bigg\}.
\end{align*}

First we derive the oracle values $k_1^*$, $k_2^*$ which (up to constant multiplier) minimizes the preceding risk.
Let
\begin{equation}\label{def:threshold:oracle}
\begin{split}
k_2^*&=\arg\min_{k_2 \in \mathbb{N}} \bigg\{B_g'(k_2)\leq \frac{k_2}{nm}\bigg\}, \\
k_1^*&=\arg\min_{k_1\leq k_2^*}\bigg\{B_g'(k_2^*)+V_g'(k_1,k_2^*)\leq 2\frac{k_1}{n}\bigg\}.
\end{split}
\end{equation}
Note $k_2^*$ is well defined since, regarding the terms in the set in $k_2^*$'s definition as $k_2$ increases, the left-hand side decreases to zero whereas the right-hand side goes to infinity.
Since $B_g'((nm)^{\frac{1}{1+2(\alpha\wedge\talpha)}}) \leq c_R (nm)^{-\frac{2(\alpha\wedge\talpha)}{1+2(\alpha\wedge\talpha)}}$, the definition of $k_2^*$ implies $k_2^*\lesssim (nm)^{\frac{1}{1+2(\alpha\wedge\talpha)}}$
and hence $k_2^*/(nm)\leq c_0 (nm)^{-\frac{2(\alpha\wedge\talpha)}{1+2(\alpha\wedge\talpha)}}$ for some $c_0>0$ large enough. 
Note $k_1^*$ is well defined since $B_g'(k_2^*)+V_g'(k_2^*,k_2^*)= (1+1/m)k_2^*/n<2k_2^*/n$ and hence the corresponding set is not empty. 
Furthermore, for $k_1^o\geq \big((4c_2)^{\frac{1}{2\talpha}} n^{\frac{1}{1+2\talpha}}\big)\vee \big(4c_0 n(nm)^{-\frac{2(\alpha\wedge\talpha)}{1+2(\alpha\wedge\talpha)}} \big)$ we have 
\begin{align*}
B_g'(k_2^*)+V_g'(k_1^o,k_2^*)
\leq \frac{k_2^*}{nm} + \frac{k_1^o}{n} + 2c_2(k_1^o)^{-2\talpha} + \frac{k_2^*}{nm}
\leq 2\frac{k_1^o}{n},
\end{align*}
which in turn implies that $k_1^*\leq k_1^o$ and 
\begin{align*}
B_g'(k_2^*)+V_g'(k_1^*,k_2^*) 
\lesssim (nm)^{-\frac{2\alpha}{1+2\alpha}}+n^{-\frac{2\talpha}{1+2\talpha}}.
\end{align*}
Also note 
\begin{align*}
B_g'(k_2)+V_g'(k_1,k_2)> \sum_{k=k_1+1}^{\infty}\tlambda_k\asymp k_1^{-2\talpha}. 
\end{align*} 
From the preceding panel and the definition of $k_1^*$, we get $k_1^*\gtrsim n^{1/(1+2\talpha)}$. Similarly, from $B_g'(k_2)\geq \sum_{k>k_2}\tlambda_k \asymp k_2^{-2\talpha}$ and the definition of $k_2^*$, we get $k_2^*\gtrsim (nm)^{1/(1+2\talpha)}$.

We estimate the oracle values using the data-driven thresholds $\hk_2$ and $\hk_1 = \hk_1^{(1)}$ from \eqref{eq:thresholdadaptf}.
Then we define our estimator of $f$ with these plug-in values:
\begin{align}\label{def:adapt:est}
\tilde{f}(x)=\hf_{\hk_1,\hk_1\vee\hk_2}(x)=\sum_{k=1}^{\hk_1} Y^{(1)}_k \psi_k(x) +\sum_{k=\hk_1+1}^{\hk_2}\hg_k \psi_k(x).
\end{align}

The next lemma shows that the data-driven thresholds $\hk_2$ and $\hk_1 = \hk_1^{(1)}$ will overshoot the oracle values $k_1^*$ and $k_2^{*}$ with exponentially small probability.

\begin{lemma}\label{lem:thresholdf}
For the thresholding estimators $\hk_1 = \hk_1^{(1)}$ and $\hk_2$ defined in \eqref{eq:thresholdadaptf} there exists a small enough constant $c>0$ such that for every $k_1>k_1^*$ and $k_2>k_2^*$,
\begin{align*}
\bP_g(\hk_1=k_1)^{1/2}\leq \frac{1}{c}e^{-c k_1^*}\qquad\text{and}\qquad \bP_g(\hk_2=k_2)^{1/2}\leq \frac{1}{c}e^{-c k_2^*}.
\end{align*}
\end{lemma}

\begin{proof}[Proof of Lemma~\ref{lem:thresholdf}]
We start with the $\hk_2$ bound and take arbitrary $k_2>k_2^*$. The union bound implies
\begin{align*}
\bP_g(\hk_2=k_2)
\leq \sum_{\ell=k_2+1}^{\sqrt{nm}}\bP_g\Big(\|\hf_{1,k_2}-\hf_{1,\ell}\|_2^2>\tau_2\frac{\ell}{nm}\Big)
= \sum_{\ell=k_2+1}^{\sqrt{nm}}\bP_g\Big( \sum_{k=k_2+1}^\ell \hg_k^2>\tau_2\frac{\ell}{nm} \Big). \nonumber
\end{align*}
Let $Z_1,\ldots,Z_{\sqrt{n}}\iid \text{N}(0,1)$.
By the inequality $(a+b)^2 \leq 2(a^2+b^2)$ for any $a,b \in \mathbb{R}$ and the inequalities $\sum_{k=k_2+1}^\ell g_k^2\leq B_g'(k_2^*)\leq k_2^*/(nm)$, the right hand side of the preceding display is bounded above by
\begin{align*}
&\sum_{\ell=k_2+1}^{\sqrt{nm}} \prob \Big\{ 2 \sum_{k=k_2+1}^\ell \Big(\frac{\tlambda_k}{m}+\frac{1}{nm}\Big) Z_k^2  >\tau_2\frac{\ell}{nm}-2\sum_{k=k_2+1}^\ell g_k^2\Big\} \\
&\leq  \sum_{\ell=k_2+1}^{\sqrt{nm}} \prob \Big\{ \sum_{k=k_2+1}^\ell \Big(\tlambda_k+\frac{1}{n}\Big) Z_k^2  >\frac{\tau_2-2}{2} \frac{\ell}{n}\Big\}.
\end{align*}
From $\sum_{k=k_2+1}^\ell \tlambda_k\leq k_2^*/(nm)< \ell/n$ and Lemma \ref{lem:concent}, for $\tau_2>6$ the right hand side of the preceding display is further bounded above by
\begin{align*}
\sum_{\ell=k_2+1}^{\sqrt{nm}}& \prob \Big\{\sum_{k=k_2+1}^\ell \Big(\tlambda_k+\frac{1}{n}\Big)\Big(Z_k^2-1\Big)>\frac{\tau_2-6}{2} \frac{\ell}{n}\Big\}\\
&\leq \sum_{\ell=k_2+1}^{\sqrt{nm}} \exp\Big\{ -\frac{(\tau_2/2-3)^2\ell^2/n^2}{\sum_{k=k_2+1}^\ell  (\tlambda_k+1/n)^2+(\tlambda_{k_2+1}+1/n)(\tau_2/2-3)\ell/n} \Big\}\\
&\leq \sum_{\ell=k_2+1}^{\sqrt{nm}} \exp\Big\{- c_{\tau_2}\frac{\ell^2}{\sum_{k=k_2^*}^\ell \tlambda_k^2n^2+ 1+\big(\tlambda_{k_2^*} +1/n \big)\ell n}\Big\}\\
&\leq \sum_{\ell=k_2+1}^{\sqrt{nm}} \exp\Big\{-c_{\tau_2}'\frac{\ell}{(k_2^{*})^{-2-4\talpha}n^2+1+(k_2^{*})^{-1-2\talpha} n  } \Big\}
\lesssim \exp\Big\{-c_{\tau_2}'' k_2^{*}  \Big\},
\end{align*}
for some positive constants $c_{\tau_2}, c_{\tau_2}'$, where the last inequality used $k_2^{*} \gtrsim n^{1/(1+2\talpha)}$, finishing the proof of the $\hk_2$ bound.

We turn now to the $\hk_1$ bound and use similar arguments as before. 
Take arbitrary $k_1>k_1^*$ and $\tau_1>4$, and let $Z_1,\ldots,Z_n \iid \text{N}(0,1)$.
We introduce the notation $k_1^{-}:=k_1-1\geq k_1^*$. 
From $\sum_{k=k_1^*}^{\sqrt{n}}\tlambda_k \leq k_1^*/n$ we get
\begin{align}
\bP_g(\hk_1=k_1)
&\leq \sum_{\ell=k_1}^{\sqrt{n}}\bP_g\Big(\|\hf_{k_1^-,\sqrt{n}}-\hf_{\ell,\sqrt{n}} \|_2^2 >\tau_1\frac{\ell}{n}\Big)\nonumber
= \sum_{\ell=k_1}^{\sqrt{n}}\bP_g\Big(\sum_{k=k_1}^\ell (Y_k^{(1)}-\hg_k)^2 >\tau_1\frac{\ell}{n}\Big)\nonumber\\
&=\sum_{\ell=k_1}^{\sqrt{n}}\prob\Big(\sum_{k=k_1}^\ell Z_k^2 (\tilde
\lambda_k+1/n)(1+1/m) >\tau_1\frac{\ell}{n}\Big)\nonumber\\
&\leq \sum_{\ell=k_1}^{\sqrt{n}}\prob\Big(\sum_{k=k_1}^\ell (Z_k^2-1)(\tilde
\lambda_k+1/n)(1+1/m) >(\tau_1-2-2/m)\frac{\ell}{n}\Big).
\label{eq:UB:j1}
\end{align}
Then applying again Lemma \ref{lem:concent} the preceding display is bounded above by a multiple of
\begin{align*}
 &\sum_{\ell=k_1}^{\sqrt{n}}\exp\Bigg\{\frac{-(\tau_1-2-2/m)^2\ell^2/n^2}{\sum_{k=k_1^-+1}^\ell(\tlambda_k+1/n)^2(1+1/m)^2 +(\tlambda_{k_1^-+1}+1/n)(1+1/m) (\tau_1-2-2/m)\ell/n} \Bigg\}\\
&\qquad\leq \sum_{\ell=k_1^*}^{\sqrt{n}}\exp\Big\{- c_{\tau_1}\frac{\ell}{(k_1^*)^{-2-4\talpha}n^2+1+(k_1^*)^{-1-2\talpha}n } \Big\}
\lesssim\exp\Big\{ -c_{\tau_1}'k_1^*\Big\},
\end{align*}
for some positive constants $c_{\tau_1}, c_{\tau_1}'$, where the last line used $k_1^*\gtrsim n^{1/(1+2\talpha)}$. This finishes the proof of the $\hk_1$ bound and hence concludes the lemma.
\end{proof}

We continue the proof of Theorem~\ref{thm:adaptivef}.
From the above lemma we can derive the adaptive convergence rate for our double-thresholded estimator.

We distinguish four cases $\{\hk_1\leq k_1^*\leq \hk_2\leq k_2^*\}$, $\{\hk_1\vee\hk_2\leq k_1^*\}$, $\{\hk_1> k_1^*\}$, and $\{\hk_2> k_2^*\}$. 
For the first case, from the definition of $k_1^*, k_2^*$ in \eqref{def:threshold:oracle} and $\hk_1, \hk_2$ in \eqref{eq:thresholdadaptf}, the term $\E_g\|\tilde{f}-f\|_2^21_{\{\hk_1\leq k_1^*\leq  \hk_2\leq k_2^* \}}$ is bounded above by
\begin{align*}
& \E_g\|\hf_{k_1^*,\hk_2}-\hf_{\hk_1,\hk_2}\|_2^21_{\{\hk_1\leq k_1^*\leq \hk_2\leq k_2^* \}}
 +\E_g\|\hf_{k_1^*,\hk_2}-\hf_{k_1^*,k_2^*}\|_2^21_{\{\hk_1\leq k_1^*\leq \hk_2\leq k_2^* \}}
 +\E_g\|\hf_{k_1^*,k_2^*}-f\|_2^2\\
 &\leq  \E_g\|\hf_{k_1^*,\sqrt{n}}-\hf_{\hk_1,\sqrt{n}}\|_2^21_{\{ \hk_1\leq k_1^* \}}
 +\E_g\|\hf_{1,\hk_2}-\hf_{1,k_2^*}\|_2^21_{\{\hk_2\leq k_2^* \}}
+ \E_g\|\hf_{k_1^*,k_2^*}-f\|_2^2\\
&\leq  \tau_1\frac{k_1^*}{n} + \tau_2\frac{k_2^*}{nm} + B_g'(k_2^*)+V_g'(k_1^*,k_2^*)
\lesssim  (nm)^{-\frac{2\alpha}{1+2\alpha}}+n^{-\frac{2\talpha}{1+2\talpha}}.
\end{align*}

For the case $\{\hk_1\vee \hk_2\leq k_1^* \}$, the term $\E_g\|\tilde{f}-f\|_2^21_{\{\hk_1\vee \hk_2\leq k_1^* \}}$ is bounded above by
\begin{align*}
&\E_g\|\hf_{\hk_1,\hk_1\vee\hk_2}-\hf_{\hk_1,k_2^*}\|_2^21_{\{\hk_1\vee  \hk_2\leq k_1^*  \}}
+\E_g\|\hf_{\hk_1,k_2^*}-\hf_{k_1^*,k_2^*}\|_2^21_{\{\hk_1\vee  \hk_2\leq k_1^*  \}}
+\E_g\|\hf_{k_1^*,k_2^*}-f\|_2^2\\
 &\leq \E_g\|\hf_{1,\hk_2}-\hf_{1,k_2^*}\|_2^21_{\{\hk_2\leq k_2^* \}}
+ \E_g\|\hf_{k_1^*,\sqrt{n}}-\hf_{\hk_1,\sqrt{n}}\|_2^21_{\{ \hk_1\leq k_1^* \}}
 +\E_g\|\hf_{k_1^*,k_2^*}-f\|_2^2,
\end{align*}
which is further bounded by a multiple of $ (nm)^{-\frac{2\alpha}{1+2\alpha}}+n^{-\frac{2\talpha}{1+2\talpha}}$ per the previous display.

Then we consider the case $\hk_2>k_2^*$. By the Cauchy-Schwarz inequality and Lemma~\ref{lem:thresholdf},
\begin{align*}
\E_g\|\tilde{f}-f\|_2 1_{\{\hk_2>k_2^* \}}
&\leq  \sum_{k_2=k_2^*+1}^{\sqrt{nm}}\sum_{k_1=1}^{k_2\wedge\sqrt{n}}\big(\E_g\|\hf_{k_1,k_2}-f\|_2^2\big)^{1/2}(\E_g1_{\hk_2=k_2})^{1/2} \\
&\lesssim  \sum_{k_2=k_2^*+1}^{\sqrt{nm}}\sum_{k_1=1}^{\sqrt{n}} \Big(\frac{k_1}{n}+ k_1^{-2\talpha}+\frac{k_2}{nm}+ k_2^{-2(\alpha\wedge \talpha)}\Big)^{1/2}\bP_g(\hk_2=k_2)^{1/2}\\
&\lesssim \sqrt{n}\sum_{k_2>k_2^*}\bP_g(\hk_2=k_2)^{1/2}
\lesssim \sqrt{n}e^{-ck_2^*}=\sqrt{n}\exp\Big\{-c' (nm)^{\frac{1}{1+2\talpha}}\Big\}
\end{align*}
which is $o\big( n^{-\frac{\talpha}{1+2\talpha}}+(nm)^{-\frac{\alpha}{1+2\alpha}}\big)$.

Finally, for the last case $\hk_1>k_1^*\gtrsim n^{1/(1+2\talpha)}$, by similar computations
\begin{align*}
\E_g\|\tilde{f}-f\|_2 1_{\{\hk_1>k_1^* \}}
&\leq  \sum_{k_1=k_1^*+1}^{\sqrt{n}}\sum_{k_2= k_1}^{\sqrt{nm}}\big(\E_g\|\hf_{k_1,k_2}-f\|_2^2\big)^{1/2}(\E_g 1_{\hk_1=k_1})^{1/2}\\
&\lesssim  \sum_{k_1=k_1^*+1}^{\sqrt{n}}\sum_{k_2=k_1}^{\sqrt{nm}} \Big(\frac{k_1}{n}+ k_1^{-2\talpha}+\frac{k_2}{nm}+ k_2^{-2(\alpha\wedge \talpha)}\Big)^{1/2}\bP_g(\hk_1=k_1)^{1/2}\\
&\lesssim \sqrt{nm}\sum_{k_1>k_1^*}\bP_g(\hk_1=k_1)^{1/2}
\lesssim \sqrt{nm}\exp\Big\{-c' n^{\frac{1}{1+2\talpha}}\Big\}
\end{align*}
which is again $o\big( n^{-\frac{\talpha}{1+2\talpha}}+(nm)^{-\frac{\alpha}{1+2\alpha}}\big)$, concluding the proof of Theorem~\ref{thm:adaptivef}.
\end{proof}

\section{Proof of Theorem~4: Adaptation for $g$}
\label{sec:proofadaptationggwn}

\begin{proof}[Proof of Theorem~4]
From the threshold-estimator part of the proof of Theorem~\ref{thm:upperboundg} in Section~\ref{sec:proof_ubthresh_g}, recall for any estimator $\tilde{g}^K$ of the form \eqref{eq:estg} the following decomposition into a squared bias and variance term 
\begin{align*}
    \E_g\|g-\tilde{g}^K\|_2^2
    = B_g(K) + V_g(K)
\end{align*}
where $B_g(K) = \sum_{k>K} g_{k}^2$ and $V_g(K) = \sum_{k=1}^{K} m^{-1} (\tlambda_k + n^{-1})$. 
First we derive the oracle value $K^*$ which (up to constant multiplier) minimizes the above integrated squared error.
Here we simply use $K^*=k_2^*$ as defined in \eqref{def:threshold:oracle}.
If $g \in \cS(\alpha, R)$, then $\sum_{k>K} g_{k}^2 \leq R^2 K^{-2\alpha}$ for any $K$, which implies $K^* \asymp (nm)^{1/(1+2\alpha)}$ and thus 
\begin{align}
\E_g\|g-\tilde{g}^{K^*}\|_2^2 \lesssim m^{-1} + (nm)^{-\frac{2\alpha}{1+2\alpha}}.\label{eq:help:adapt:g}
\end{align}
We use the data-driven threshold $\hK = \hk_2$ from \eqref{eq:thresholdadaptf} to define our estimator:
\begin{align}\label{def:adapt:estg}
\tilde{g}^{\hK}(x)
= \sum_{k=1}^{\hK} \hg_k \psi_k(x).
\end{align}

Lemma~\ref{lem:thresholdf} shows the threshold estimator $\hK = \hk_2$ will overshoot the oracle value $K^*$ with exponentially small probability.
From this we can derive the adaptive convergence rate for our thresholded estimator.
We distinguish two cases $\{\hK \leq K^*\}$ and $\{\hK> K^*\}$. 
From the Cauchy-Schwarz inequality and Lemma~\ref{lem:thresholdf}, we get
\begin{align*}
\E_g\|\tilde{g}^{\hK}-g\|_2^2 1_{\{\hK > K^*\}}
&\leq \sum_{k=K^*+1}^{\sqrt{nm}} \big(\E_g\|\hg_{k}-g\|_2^2\big)^{1/2}(\E_g1_{\hK=k})^{1/2} \\
&\lesssim \sum_{k=K^*+1}^{\sqrt{nm}} \big((nm)^{-1}k + m^{-1} + k^{-2\alpha}\big)^{1/2} \mathbb{P}_g(\hK=k)^{1/2} \\
&\lesssim \sum_{k>K^*}\mathbb{P}_g(\hK=k)^{1/2}
\lesssim e^{-c K^*} 
= \exp\Big\{-c' (nm)^{\frac{1}{1+2\alpha}}\Big\} 
\end{align*}
which is $o( (nm)^{-\frac{2\alpha}{1+2\alpha}})$.
From the definition of $\hK=\hk_2$ in \eqref{eq:thresholdadaptf} and $K^*=k_2^*$ in \eqref{def:threshold:oracle}, we get
\begin{align*}
\E_g\|\tilde{g}^{\hK}-g\|_2^2 1_{\{\hK\leq K^*\}}
\lesssim \E_g\|\tilde{g}^{\hK}-\tilde{g}^{K^*}\|_2^2 1_{\{\hK\leq K^*\}} + \E_g\|\tilde{g}^{K^*} - g\|_2^2 
\leq \tau_2 \frac{K^*}{nm} + \E_g\|\tilde{g}^{K^*} - g\|_2^2
\end{align*}
which is $O( (nm)^{-\frac{2\alpha}{1+2\alpha}}+m^{-1})$ in view of \eqref{eq:help:adapt:g}, finishing the proof of the theorem.
\end{proof}

\end{document}